\documentclass{article}

\usepackage{arxiv}

\usepackage{amsmath,amsthm}
\usepackage{amsfonts,amssymb}
\usepackage{cancel}
\usepackage{graphicx}
\usepackage{booktabs}
\usepackage{algorithm}
\usepackage[noend]{algpseudocode}
\usepackage{xcolor}

\usepackage{times}

\def\ci{\perp\!\!\!\perp}
\newcommand{\oo}{{\circ\!{\--}\!\circ}}

\newcommand{\starstar}{{\star\!{\--}\!\star}}
\newcommand{\conflict}{{x\!{\--}\!x}}
\newcommand{\B}{\widehat{\mathcal{B}}^-_t(X^j_t)}
\newcommand{\Pa}{\mathcal{P}^-_t}
\newcommand{\Ct}{\mathcal{C}_t(X^j_t)}
\newcommand{\CI}{\perp\!\!\!\perp}
\newcommand{\nCI}{\not\perp\!\!\!\perp}
\newcommand{\Xlag}{\mathbf{X}^-_t}
\newcommand{\chg}{\textcolor{blue!60!black}{$\star$}}

\newtheorem{myassum}{Assumptions}
\newtheorem{mydef}{Definition}

\newtheorem{mythm}{Theorem}

\newtheorem{mylemma}{Lemma}
\theoremstyle{definition}
\newtheorem{claim}{Claim}

\usepackage{hyperref}

\title{Discovering contemporaneous and lagged causal relations in autocorrelated nonlinear time series datasets\thanks{This version (October 2026) corrects the termination criterion of Alg.~1 and the statement of Lemma~1, see Sect.~S0 in the SM. The consistency of PCMCI$^+$ (Theorem~1) and the numerical results are not affected.}}

\author{ {\bf Jakob Runge} \\
German Aerospace Center \\
Institute of Data Science\\
07745 Jena, Germany \\
and\\
Technische Universit\"at Berlin \\
10623 Berlin, Germany \\
}

\begin{document}

\maketitle

\begin{abstract}
The paper introduces a novel conditional independence (CI) based method for linear and nonlinear, lagged and contemporaneous causal discovery from observational time series in the causally sufficient case. Existing CI-based methods such as the PC algorithm and also common methods from other frameworks suffer from low recall and partially inflated false positives for strong autocorrelation which is an ubiquitous challenge in time series. The novel method, PCMCI$^+$, extends PCMCI [Runge et al., 2019b] to include discovery of contemporaneous links. PCMCI$^+$ improves the reliability of CI tests by optimizing the choice of conditioning sets and even benefits from autocorrelation. The method is order-independent and consistent in the oracle case. A broad range of numerical experiments demonstrates that PCMCI$^+$ has higher adjacency detection power and especially more contemporaneous orientation recall compared to other methods while better controlling false positives. Optimized conditioning sets also lead to much shorter runtimes than the PC algorithm. PCMCI$^+$ can be of considerable use in many real world application scenarios where often time resolutions are too coarse to resolve time delays and strong autocorrelation is present.
\end{abstract}

\section{INTRODUCTION}
A number of frameworks address the problem of causal discovery from observational data utilizing different assumptions. 
Next to Bayesian score-based methods \cite{Chickering2002}, classical Granger causality (GC) \cite{Granger1969}, and the more recent restricted structural causal models (SCM) framework \cite{Peters2018,Spirtes2016},  conditional independence (CI) based network learning algorithms \cite{Spirtes2000} form a main pillar. 
A main representative of the CI framework in the causally sufficient case (no unobserved common drivers) is the PC algorithm \cite{Spirtes1991}. Its advantages lie, firstly, in the flexibility of utilizing a wide and growing class of CI tests, from linear partial correlation (ParCorr) and non-parametric residual-based approaches \cite{Ramsey2014,Runge2018d} to Kernel measures \cite{Zhang2011}, tests based on conditional mutual information \cite{Runge2018a}, and neural networks \cite{Sen2017}.
Secondly, the PC algorithm utilizes sparsity making it applicable also to large numbers of variables while score- and SCM-based methods are more difficult to adapt to nonlinear high-dimensional causal discovery.

Causal discovery in the time series case is partially less and partially more challenging \cite{Runge2019a}. Obviously, time-order greatly helps in identifying causal directions for lagged links (causes precede effects). This forms the basis of GC which, however, cannot deal with contemporaneous links and suffers from the curse of dimensionality \cite{Runge2018d}. SCM-based methods such as LiNGAM \cite{hyvarinen2010estimation} and also CI-based methods \cite{Runge2018d,Entner2010,malinsky2018causal} have been adapted to the time series case. In \cite{moneta2011causal} GC is augmented by the PC algorithm.
However, properties such as non-stationarity and especially autocorrelation can make causal discovery much less reliable. 

Here I show that autocorrelation, an ubiquitous property of time series (e.g., temperature data), is especially detrimental and propose a novel CI-based method, PCMCI$^+$, that extends the PCMCI method from \cite{Runge2018d} to also include discovery of contemporaneous links, which requires substantial changes. PCMCI$^+$ is based on two central ideas that deviate from the PC algorithm and the time-series adaptations of FCI in \cite{Entner2010,malinsky2018causal}: First, an edge removal phase is conducted separately for lagged and contemporaneous conditioning sets and the lagged phase uses much fewer CI tests. Secondly, and more importantly, PCMCI$^+$ optimizes the choice of conditioning sets for the individual CI tests to make them better calibrated under autocorrelation and increase detection power by utilizing the momentary conditional independence idea \cite{Runge2018d}. 
The paper is structured as follows. Section~\ref{sec:causal_discovery} briefly introduces the problem and Sect.~\ref{sec:pcmciplus} describes the method and states theoretical results. Numerical experiments in Sect.~\ref{sec:numerics} show that PCMCI$^+$ benefits from strong autocorrelation and yields much more adjacency detection power and especially more orientation recall for contemporaneous links while better controlling false positives at much shorter runtimes than the PC algorithm. A Supplementary Material (SM) contains proofs and further numerical experiments. 

\section{TIME SERIES CAUSAL DISCOVERY} \label{sec:causal_discovery}

\subsection{PRELIMINARIES}
We are interested in discovering \textit{time series graphs} (e.g., \cite{Runge2018b}) that can represent the temporal dependency structure underlying complex dynamical systems. Consider an underlying discrete-time structural causal process $\mathbf{X}_t=(X^1_t,\ldots,X^N_t)$ with 
\begin{align} \label{eq:causal_model}
X^j_t &= f_j\left(\mathcal{P}(X^j_t),\,\eta^j_t\right)
\end{align}
where $f_j$ are arbitrary measurable functions with non-trivial dependencies on their arguments and $\eta^j_t$ represents mutually ($i\neq j$) and serially ($t'\neq t$) independent dynamical noise. The nodes in a time series graph $\mathcal{G}$ (example in Fig.~\ref{fig:motivation}) represent the variables $X^j_t$ at different lag-times and the set of variables that $X^j_t$ depends on defines the causal parents $\mathcal{P}(X^j_t)\subset \mathbf{X}^-_{t+1}=(\mathbf{X}_{t}, \mathbf{X}_{t-1},\ldots){\setminus} \{ X^j_t\}$. We denote \emph{lagged parents} by $\mathcal{P}^-_t(X^j_t) = \mathcal{P}(X^j_t) \cap \mathbf{X}^-_t$. A lagged ($\tau>0$) or contemporaneous ($\tau=0$) causal link $X^i_{t-\tau} \to X^j_t$ exists if $X^i_{t-\tau}\in \mathcal{P}(X^j_t)$.
Throughout this work the graph $\mathcal{G}$ is assumed \emph{acyclic} and the causal links \emph{stationary} meaning that if $X^i_{t-\tau} \to X^j_t$ for some $t$, then $X^i_{t^\prime - \tau} \to X^j_{t^\prime}$ for all $t^\prime\neq t$. Then we can always fix one variable at $t$ and take $\tau \geq 0$. Note that the stationarity assumption may be relaxed.
The graph is actually infinite in time, but in practice only considered up to some maximum time lag $\tau_{\max}$. 
We define the set of adjacencies $\mathcal{A}(X^j_t)$ of a variable $X^j_t$ to include all $X^i_{t-\tau}$ for $\tau \geq 0$ that have a (lagged or contemporaneous) link with $X^j_t$ in $ \mathcal{G}$. We define contemporaneous adjacencies as $\mathcal{A}_t(X^j_t)=\mathcal{A}(X^j_t) \cap \mathbf{X}_t$. A sequence of $m$ contemporaneous links is called a \emph{directed contemporaneous path} if for all $k\in \{1,\ldots,m\}$ the link $X^{i+k-1}_{t} \to X^{i+k}_t$ occurs.  We call $X^i_t$ a \emph{contemporaneous ancestor} of $X^j_t$ if there is a directed contemporaneous path from $X^i_t$ to $X^j_t$ and we denote the set of all contemporaneous ancestors as $\mathcal{C}_t(X^j_t)$ (which excludes $X^j_t$ itself). We denote separation in the graph by $\bowtie$, see \cite{Runge2018b} for further notation details.

\subsection{PC ALGORITHM}
The PC algorithm is the most wide-spread CI-based causal discovery algorithm for the causally sufficient case and utilizes the Markov and Faithfulness assumptions as formally defined in Sect.~\ref{sec:definitions}. Adapted to time series (analogously to the methods for the latent case in \cite{Entner2010,malinsky2018causal}), it consists of three phases: First, a skeleton of adjacencies is learned based on iteratively testing which pairs of variables (at different time lags) are conditionally independent at some significance level $\alpha_{\rm PC}$ (Alg.~\ref{algo:step2} with the PC option). For lagged links, time-order automatically provides orientations, while for contemporaneous links a collider phase (Alg.~\ref{algo:step3_SM}) and rule phase (Alg.~\ref{algo:step4_SM}) determine the orientation of links. CI-based discovery algorithms can identify the contemporaneous graph structure only up to a Markov equivalence class represented as a completed partially directed acyclic graph (CPDAG). We denote links for which more than one orientation occurs in the Markov equivalence class by $X^i_{t} \oo X^j_t$. Here we consider a modification of PC that removes an undesired dependence on the order of variables, called PC-stable \cite{Colombo2014}. These modifications also include either the \emph{majority} or \emph{conservative} \cite{ramsey2006adjacency} rule for handling ambiguous triples where separating sets are inconsistent, and conflicting links where different triples in the collider or orientation phase lead to conflicting link orientations. With the \emph{conservative} rule the PC algorithm is consistent already under the weaker Adjacency Faithfulness condition \cite{ramsey2006adjacency}. 
Another approach for the time series case (considered in the numerical experiments) is to combine vector-autoregressive modeling to identify lagged links with the PC algorithm for the contemporaneous causal structure \cite{moneta2011causal}.

\subsection{AUTOCORRELATION}
\begin{figure*}[t]  
\centering
\includegraphics[width=1.\linewidth]{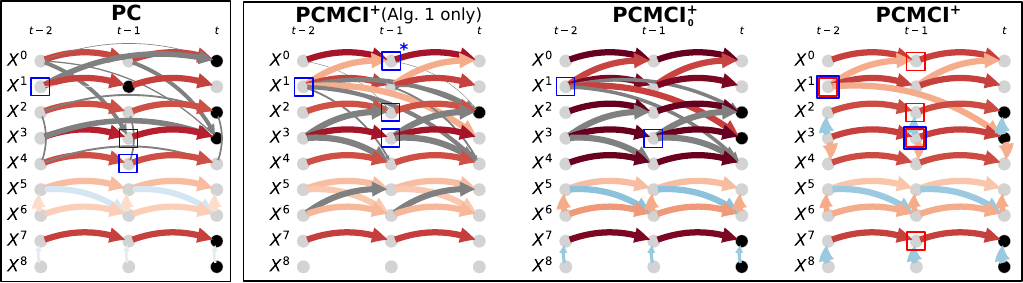}%
\caption{The curse and blessing of autocorrelation. Linear example of model~\eqref{eq:numericalmodel} with ground truth links shown for the PCMCI$^+$ case (right panel). All autodependency coefficients are $0.95$ (except $0.475$ for $X^{5,6}$) and all cross-coupling coefficients are $0.4$ ($\pm$ indicated by red/blue links). The graphs show true and false link detection rates as the link width (if $>$ 0.06) for true (color indicating ParCorr) and incorrect links (grey) for the PC algorithm, Alg.~\ref{algo:step1}, and the variants PCMCI$^+$ and PCMCI$^+_0$ as explained in the text (detection rates based on $500$ realizations run at $\alpha_{\rm PC}=0.01$ for $T=500$).}
\label{fig:motivation}
\end{figure*}
To illustrate the challenge of autocorrelation, in Fig.~\ref{fig:motivation} we consider a linear example with lagged and contemporaneous ground truth links shown for the PCMCI$^+$ case (right panel). The PC algorithm (Alg.~\ref{algo:step2} with ParCorr CI test) starts by testing all unconditional independencies ($p=0$). Here the coupled pairs $(X^5,X^6)$ as well as $(X^7,X^8)$ are independent of the other variables and removed from each others adjacency sets, which shows how PC exploits sparsity and reduces the estimation dimension compared to fitting a full model on the whole past as in the GC framework.
Due to the strong autocorrelation the remaining variables, on the other hand, are almost all adjacent to each other at multiple time lags in this iteration. In the next iteration, CI for all remaining links is tested conditional on all one-dimensional ($p=1$) conditioning sets. Here the PC algorithm removes the true lagged link $X^1_{t-1}\to X^0_t$ (black dots) due to the incorrect CI result $X^1_{t-1} \ci X^0_t|X^1_{t-2}$ (condition marked by blue box). Later this then leads to the false positive $X^1_{t-2}\to X^0_t$ (grey link) since $X^1_{t-1}$ is not conditioned on. In a similar way the true link $X^1_{t-2}\to X^3_t$ is missed leading to the false positive $X^0_{t-1}\to X^3_t$. 
Further, the true contemporaneous link $X^2_{t}\oo X^3_t$ (similarly $X^3_{t}\oo X^4_t$) is removed when conditioning on $\mathcal{S}=(X^4_{t-1},X^3_{t-1})$ (blue boxes), which leads to the false positive autodependencies at lag $2$ for $X^2_{t}, X^4_{t}$, while the false autodependency $X^3_{t-2}\to X^3_t$ is due to missing $X^1_{t-2}\to X^3_t$. This illustrates the pattern of a cascade of false negative errors (missing links) leading to false positives in later stages of the PC algorithm.

What determines the removal of a true link in the finite sample case? Detection power depends on sample size, the significance level $\alpha_{\rm PC}$, the CI test dimension ($p+2$), and effect size, e.g., the absolute ParCorr (population) value, here denoted $I(X^i_{t-\tau};X^j_t|\mathcal{S})$ for some conditioning set $\mathcal{S}$. Within each $p$-iteration the sample size, $\alpha_{\rm PC}$, and the dimension are the same and a link will be removed if $I(X^i_{t-\tau};X^j_t|\mathcal{S})$ falls below the $\alpha_{\rm PC}$-threshold for \emph{any} considered $\mathcal{S}$. Hence, the overall minimum effect size $\min_{\mathcal{S}} [I(X^i_{t-\tau};X^j_t|\mathcal{S})]$ determines whether a link is removed. The PC algorithm will iterate through \emph{all} subsets of adjacencies such that this minimum can become very small. Low effect size can be understood as a low (causal) signal-to-noise ratio: Here $I(X^1_{t-1};X^0_t|X^1_{t-2})$ is small since the signal $X^1_{t-1}$ is reduced by conditioning on its autodependency $X^1_{t-2}$ and the `noise' in $X^0_t$ is large due to its strong autocorrelation. 

But autocorrelation can also be a blessing. The contemporaneously coupled pair $(X^7,X^8)$ illustrates a case where autocorrelation helps to identify the orientation of the link. Without autocorrelation the output of PC would be an unoriented link to indicate the Markov equivalence class. On the other hand, the detection rate here is rather weak since, as above, the signal (link from $X^8_t$) is small compared to the noise (autocorrelation in $X^7$).

This illustrates the curse and blessing of autocorrelation. In summary, the PC algorithm often results in false negatives (low recall) and these then lead to false positives. Another reason for false positives are ill-calibrated tests: To correctly model the null distribution, each individual CI test would need to account for autocorrelation, which is difficult in a complex multivariate and potentially nonlinear setting \cite{Runge2018b}. In the experiments we will see that the PC algorithm features inflated false positives.

As a side comment, the pair $(X^5,X^6)$ depicts a feedback cycle. These often occur in real data and the example shows that time series graphs allow to resolve time-delayed feedbacks while an aggregated \emph{summary graph} would contain a cyclic dependency and summary graph-based methods assuming acyclic graphs would not work. The orientation of the contemporaneous link $X^6_t\to X^5_t$ is achieved via rule R1 in the orientation phase of PC (Alg.~\ref{algo:step4_SM}).

\section{PCMCI$^+$} \label{sec:pcmciplus}
\begin{figure*}[t]  
\centering
\includegraphics[width=1.\linewidth]{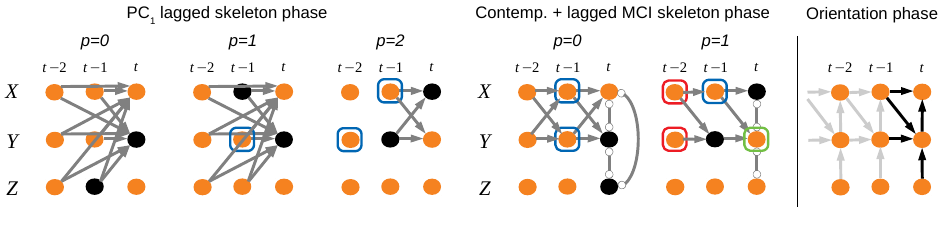}%
\caption{Schematic of PCMCI$^+$. Note that for ease of visualization repeating edges due to stationarity are only partially shown.}
\label{fig:schematic}
\end{figure*}
\subsection{ALGORITHM}
The goal of PCMCI$^+$ is to optimize the choice of conditioning sets in CI tests in order to increase detection power and at the same time maintain well-calibrated tests. The approach is based on two central ideas, (1)~separating the skeleton edge removal phase into a lagged and contemporaneous conditioning phase with much fewer CI tests and (2)~utilizing the momentary conditional independence (MCI) test \cite{Runge2018d} idea in the contemporaneous conditioning phase. Below, I explain the reasoning behind. Figure~\ref{fig:schematic} illustrates the steps.

First, the goal of PC's skeleton phase is to remove all those adjacencies that are due to indirect paths and common causes by conditioning on subsets $\mathcal{S}$ of the variables' neighboring adjacencies in each iteration. Consider a variable $X^j_t$. If we test lagged adjacencies from nodes $X^i_{t-\tau}\in \mathbf{X}^-_t$ conditional on the whole past, i.e., $\mathcal{S}=\mathbf{X}^-_t\setminus\{X^i_{t-\tau}\}$, the only indirect adjacencies remaining are due to paths through contemporaneous parents of $X^j_t$.  This is in contrast to conditioning sets on contemporaneous adjacencies which can also open up paths $X^j_t \to X^k_t \leftarrow X^i_{t-\tau}$ if $X^k_t$ is conditioned on. One reason why the PC algorithm tests \emph{all} combinations of subsets $\mathcal{S}$ is to avoid opening up such collider paths. 
Therefore, one approach would be to start by $\mathcal{S}=\mathbf{X}^-_t\setminus\{X^i_{t-\tau}\}$ and then iterate through contemporaneous conditions. A similar idea lies behind the combination of GC and the PC algorithm in \cite{moneta2011causal}. However, conditioning on large-dimensional conditioning sets strongly affects detection power \cite{Runge2018d}. To avoid this, the lagged conditioning phase of PCMCI$^+$ (Alg.~\ref{algo:step1}, see Fig.~\ref{fig:schematic} left panels) tests all pairs $(X^i_{t-\tau}, X^j_t)$ for $\tau>0$ conditional on only the \emph{strongest} $p$ adjacencies of $X^j_t$ in each $p$-iteration without going through all $p$-dimensional subsets of adjacencies. This choice $(i)$~improves the causal signal-to-noise ratio and recall since for a given test $X^i_{t-\tau} \ci X^j_t~|~\mathcal{S}$ the `noise' in $X^j_t$ due to other lagged adjacencies is conditioned out, $(ii)$~leads to fewer CI tests further improving recall, and $(iii)$~speeds up the skeleton phase. We denote the lagged adjacency set resulting from Alg.~\ref{algo:step1} as $\widehat{\mathcal{B}}^-_t(X^j_t)$. Lemma~\ref{thm:step1b} in Sect.~\ref{sec:consistency} states that the only remaining indirect adjacencies in $\widehat{\mathcal{B}}^-_t(X^j_t)$ are then due to paths passing through contemporaneous parents of $X^j_t$. In the schematic in Fig.~\ref{fig:schematic} this is the link $Y_{t-1}\to X_t$.

Secondly, in Alg.~\ref{algo:step2} (Fig.~\ref{fig:schematic} center panels) the graph $\mathcal{G}$ is  initialized with all contemporaneous adjacencies plus all lagged adjacencies from $\widehat{\mathcal{B}}^-_t(X^j_t)$ for all $X^j_t$. Algorithm~\ref{algo:step2} tests all (unordered lagged and ordered contemporaneous) adjacent pairs $(X^i_{t-\tau}, X^j_t)$ and iterates through contemporaneous conditions $\mathcal{S}\subseteq\mathcal{A}_t(X^j_t)$ with the MCI test
\begin{align}
X^i_{t-\tau}&{\ci} X^j_t~|~\mathcal{S},\widehat{\mathcal{B}}^-_t(X^j_t){\setminus}\{X^i_{t{-}\tau}\},\widehat{\mathcal{B}}^-_{t{-}\tau}(X^i_{t{-}\tau}) \label{eq:pcmci}.
\end{align}
In the schematic in Fig.~\ref{fig:schematic} the condition on $\mathcal{S}=Y_t$, as part of the full conditioning set $\mathcal{S},\widehat{\mathcal{B}}^-_t(X^j_t){\setminus}\{X^i_{t{-}\tau}\},\widehat{\mathcal{B}}^-_{t{-}\tau}(X^i_{t{-}\tau})$, removes the link $X_t\oo Z_t$.
The condition on $\widehat{\mathcal{B}}^-_t(X^j_t)$ blocks paths through lagged parents and the advantage of the additional conditioning on $\widehat{\mathcal{B}}^-_{t{-}\tau}(X^i_{t{-}\tau})$ is discussed in the following. We denote the variant without the condition on $\widehat{\mathcal{B}}^-_{t{-}\tau}(X^i_{t{-}\tau})$ as PCMCI$^+_0$.
Both versions are followed by the collider orientation phase (Alg.~\ref{algo:step3_SM}) and rule orientation phase (Alg.~\ref{algo:step4_SM}) which are deferred to the SM since they are equivalent to the PC algorithm with the modification that the additional CI tests in the collider phase for the conservative or majority rule are also based on the test~\eqref{eq:pcmci} (Fig.~\ref{fig:schematic} right panel). 

We now discuss PCMCI$^+_0$ and PCMCI$^+$ on the example in Fig.~\ref{fig:motivation}.
Algorithm~\ref{algo:step1} tests $X^1_{t-1}\to X^0_t$ conditional on $\mathcal{S}=\{X^0_{t-1}\}$ for $p=1$ and $\mathcal{S}=\{X^0_{t-1},X^1_{t-2}\}$ for $p=2$ as the two strongest adjacencies (as determined by the test statistic value, see pseudo-code). In both of these tests the effect size $I$ (causal signal-to-noise ratio) is much larger than for the condition on $\mathcal{S}=\{X^1_{t-2}\}$ which lead to the removal of $X^1_{t-1}\to X^0_t$ in the PC algorithm.
In Sect.~\ref{sec:consistency} we elaborate more rigorously on effect size. 
In the example $\widehat{\mathcal{B}}^-_t(X^2_t)$ is indicated as blue boxes in the second panel and contains lagged parents as well as adjacencies due to paths passing through contemporaneous parents of $X^2_t$. One false positive, likely due to an ill-calibrated test caused by autocorrelation, is marked by a star.

Based on these lagged adjacencies, Alg.~\ref{algo:step2} with the PCMCI$^+_0$ option then recovers all lagged links (3rd panel), but it still the misses contemporaneous adjacencies $X^2_{t}\oo X^3_t$ and $X^3_{t}\oo X^4_t$ and we also see strong lagged false positives from $X^3$ to $X^2$ and $X^4$. What happened here? The problem are now tests on contemporaneous links: The CI test for PCMCI$^+_0$ in the $p=0$ loop, like the original PC algorithm, will test \emph{ordered} contemporaneous pairs. Hence, first $X^2_{t}\oo X^3_t$ conditional on $\widehat{\mathcal{B}}^-_t(X^3_t)$ and, if the link is not removed, $X^3_{t}\oo X^2_t$ conditional on $\widehat{\mathcal{B}}^-_t(X^2_t)$. Here $X^2_{t}\oo X^3_t$ is removed conditional on $\widehat{\mathcal{B}}^-_t(X^3_t)$ (indicated by blue boxes in the panel) because $I(X^2_{t}; X^3_t|\widehat{\mathcal{B}}^-_t(X^3_t))$ falls below the significance threshold.

The second central idea of PCMCI$^+$ is to improve the effect size of CI tests for contemporaneous links by conditioning on \emph{both} lagged adjacencies $\widehat{\mathcal{B}}^-_t$ in the CI test~\eqref{eq:pcmci} (see blue and red boxes in Fig.~\ref{fig:motivation} right panel). At least for the initial phase $p=0$ one can prove that for non-empty $\widehat{\mathcal{B}}^-_t$ the effect size of the PCMCI$^+$ CI test is always strictly larger than that of the PCMCI$^+_0$ test (Thm.~\ref{thm:effect_size}). I conjecture that this similarly holds for PCMCI$^+$ vs. the PC algorithm. Higher effect size leads to higher recall and PCMCI$^+$ now recovers all lagged as well as contemporaneous links and also correctly removes the lagged false positives that PCMCI$^+_0$ obtains. Also the contemporaneous coupled pair $(X^7,X^8)$ is now much better detected since the MCI effect size $I(X^7_t;X^8_t|X^7_{t-1})$ is larger than $I(X^7_t;X^8_t)$, one of the two PCMCI$^+_0$ and PC algorithm effect sizes tested here.

Another advantage, discussed in \cite{Runge2018d} is that PCMCI$^+$ CI tests are better calibrated, in contrast to PCMCI$^+_0$ and PC algorithm tests, since the condition on both parents removes autocorrelation effects. Note that for lagged links the effect size of PCMCI$^+$ is generally smaller than that of PCMCI$^+_0$ since the extra condition on $\widehat{\mathcal{B}}^-_{t{-}\tau}(X^i_{t{-}\tau})$ can only reduce effect size (see \cite{Runge2012b}). This is the cost of avoiding inflated false positives.

In summary, the central PCMCI$^+$ idea is to increase effect size in individual CI tests to achieve higher detection power and at the same time maintain well-controlled false positives also for high autocorrelation. Correct adjacency information then leads to better orientation recall in Alg.~\ref{algo:step3_SM},\,\ref{algo:step4_SM}. The other advantage of PCMCI$^+$ compared to the PC algorithm is a much faster and, as numerical examples show, also much less variable runtime.

The full algorithm is detailed in pseudo-code Algorithms~\ref{algo:step1},\ref{algo:step2},\ref{algo:step3_SM},\ref{algo:step4_SM} with differences to PC and PCMCI$^+_0$ indicated. Note that pairs $(X^i_{t-\tau},X^j_t)$ in lines 5 and 6 of Alg.~\ref{algo:step2} are ordered for $\tau=0$ and unordered for $\tau>0$. 
One can construct (rather conservative) $p$-values for the skeleton adjacencies $(X^i_{t-\tau},X^j_t)$ by taking the maximum $p$-value over all CI tests conducted in Alg.~\ref{algo:step2}. A link strength can be defined corresponding to the test statistic value of the maximum $p$-value. Based on the PC stable variant, PCMCI$^+$ is fully order-independent. Here shown is the majority-rule implementation of the collider phase, the version without handling ambiguous triples and for the conservative rule are detailed in Alg.~\ref{algo:step3_SM}. Note that the tests in the collider phase also use the CI tests~\eqref{eq:pcmci}.

Like other CI-based methods, PCMCI$^+$ has the free parameters $\alpha_{\rm PC}$, $\tau_{\max}$, and the choice of the CI test. $\alpha_{\rm PC}$ can be chosen based on cross-validation or an information criterion (implemented in \texttt{tigramite}). $\tau_{\max}$ should be larger or equal to the maximum true time lag of any parent and can in practice also be chosen based on model selection. However, the numerical experiments indicate that, in contrast to GC, a too large $\tau_{\max}$ does not degrade performance much and $\tau_{\max}$ can also be chosen based on the lagged dependence functions, see \cite{Runge2018d}. PCMCI$^+$ can flexibly be combined with different CI tests for nonlinear causal discovery, and for different variable types (discrete or continuous, univariate or multivariate).


The computational complexity of PCMCI$^+$ strongly depends on the network structure. The sparser the causal dependencies, the faster the convergence. Compared to the original PC algorithm with worst-case exponential complexity, the complexity is much reduced since Alg.~\ref{algo:step1} only has polynomial complexity \cite{Runge2018d} and Alg.~\ref{algo:step2} only iterates through contemporaneous conditioning sets, hence the worst-case exponential complexity only applies to $N$ and not to $N \tau_{\max}$. 

\begin{algorithm*}[h!]
\caption{(PCMCI$^+$ / PCMCI$^+_0$ lagged skeleton phase) \emph{Correction (2026): Only the lines marked \chg{} change. The loop no longer stops directly after an iteration in which candidates were removed. Instead, it repeats the last iteration at $p = |\B|-1$ until no candidate is removed.}}
\begin{algorithmic}[1]
\Require Time series dataset $\mathbf{X}=(X^1,\,\ldots,X^N)$, max. time lag $\tau_{\max}$, significance threshold $\alpha_{\rm PC}$, CI test ${\rm CI}(X,\,Y,\,\mathbf{Z})$ returning $p$-value and test statistic value $I$
\ForAll{$X^j_{t}$ in $\mathbf{X}_t$} 
    \State Initialize $\widehat{\mathcal{B}}^-_t(X^j_t){=}\mathbf{X}^{-}_{t}{=}(\mathbf{X}_{t-1},\dots,\mathbf{X}_{t-\tau_{\max}})$ and  $I^{\min}(X^i_{t-\tau},X^j_t)=\infty ~~\forall~ X^i_{t-\tau} \in \widehat{\mathcal{B}}^-_t(X^j_t)$
    \State Let $p=0$ and $r={\rm false}$ \hfill \chg
    \While {$\widehat{\mathcal{B}}^-_t(X^j_t)\neq\emptyset$ and ($|\widehat{\mathcal{B}}^-_t(X^j_t)|{-}1 \geq p$ or $r$)} \hfill \chg
        \State Let $p={\min}(p,\,|\widehat{\mathcal{B}}^-_t(X^j_t)|{-}1)$ \hfill \chg
        \State Let $r={\rm false}$ \hfill \chg
        \ForAll{$X^i_{t-\tau}$ in $\widehat{\mathcal{B}}^-_t(X^j_t)$ satisfying $|\widehat{\mathcal{B}}^-_t(X^j_t){\setminus}\{X^i_{t-\tau}\}| \geq p$} 
            \State $\mathcal{S}=$ first $p$ variables in $\widehat{\mathcal{B}}^-_t(X^j_{t})\setminus \{X^i_{t-\tau}\}$
            \State $(\text{$p$-value},\,I) \gets$ \Call{CI}{$X^i_{t-\tau},\,X^j_{t},\,\mathcal{S}$}
            \State $I^{\min}(X^i_{t-\tau},X^j_t)=\min(|I|,I^{\min}(X^i_{t-\tau},X^j_t))$
            \If{$p$-value $> \alpha_{\rm PC}$} mark $X^i_{t-\tau}$ for removal and let $r={\rm true}$ \hfill \chg
            \EndIf
        \EndFor
        \State Remove non-significant entries and sort $\widehat{\mathcal{B}}^-_t(X^j_t)$ by $I^{\min}(X^i_{t-\tau},X^j_t)$ from largest to smallest
        \State Let $p=p+1$
    \EndWhile
\EndFor
\State \Return $\widehat{\mathcal{B}}^-_t(X^j_t)$ for all $X^j_{t}$ in $\mathbf{X}_t$
\end{algorithmic}
 \label{algo:step1}
\end{algorithm*}

\begin{algorithm*}[h!]
\caption{(PCMCI$^+$ / PCMCI$^+_0$ contemporaneous skeleton phase / PC full skeleton phase)}
\begin{algorithmic}[1]
\Require Time series dataset $\mathbf{X}=(X^1,\,\ldots,X^N)$, max. time lag $\tau_{\max}$, significance threshold $\alpha_{\rm PC}$, ${\rm CI}(X,\,Y,\,\mathbf{Z})$, PCMCI$^+$ / PCMCI$^+_0$: $\widehat{\mathcal{B}}^-_t(X^j_t)$ for all $X^j_{t}$ in $\mathbf{X}_t$
\State PCMCI$^+$ / PCMCI$^+_0$: Form time series graph $\mathcal{G}$ with lagged links from $\widehat{\mathcal{B}}^-_t(X^j_t)$ for all $X^j_{t}$ in $\mathbf{X}_t$ and fully connect all contemporaneous variables, i.e., add  $X^i_t \oo X^j_t$ for all $X^i_t\neq X^j_t\in \mathbf{X}_t$
\item[] PC: Form fully connected time series graph $\mathcal{G}$ with lagged and contemporaneous links
\State PCMCI$^+$ / PCMCI$^+_0$: Initialize contemporaneous adjacencies $\widehat{\mathcal{A}}(X^j_t):=\widehat{\mathcal{A}}_t(X^j_t)=\{X^i_{t}{\neq} X^j_t \in \mathbf{X}_t: X^i_t \oo X^j_t ~\text{in $\mathcal{G}$}\}$
\item[] PC: Initialize full adjacencies $\widehat{\mathcal{A}}(X^j_t)$ for all (lagged and contemporaneous) links in $\mathcal{G}$
\State Initialize $I^{\min}(X^i_{t-\tau},X^j_t)=\infty$ for all links in $\mathcal{G}$
\State Let $p=0$
\While {any adjacent pairs $(X^i_{t-\tau}, X^j_t)$ for $\tau \geq 0$ in $\mathcal{G}$ satisfy $|\widehat{\mathcal{A}}(X^j_t){\setminus}\{X^i_{t-\tau}\}| \geq p$}
    \State Select new adjacent pair $(X^i_{t-\tau}, X^j_t)$ for $\tau \geq 0$ satisfying $|\widehat{\mathcal{A}}(X^j_t){\setminus}\{X^i_{t-\tau}\}| \geq p$
    \While {$(X^i_{t-\tau}, X^j_t)$ are adjacent in $\mathcal{G}$ and not all $\mathcal{S}\subseteq\widehat{\mathcal{A}}(X^j_t){\setminus}\{X^i_{t-\tau}\}$ with $|\mathcal{S}|=p$ have been considered}
        \State Choose new $\mathcal{S}{\subseteq}\widehat{\mathcal{A}}(X^j_t){\setminus}\{X^i_{t-\tau}\}$ with $|\mathcal{S}|{=}p$ 
        \State PCMCI$^+$: Set $\mathbf{Z}{=}(\mathcal{S},\widehat{\mathcal{B}}^-_t(X^j_t){\setminus}\{X^i_{t{-}\tau}\},\widehat{\mathcal{B}}^-_{t{-}\tau}(X^i_{t{-}\tau}))$
        \item[] \hspace*{2.7em} PCMCI$^+_0$: Set $\mathbf{Z}{=}(\mathcal{S},\widehat{\mathcal{B}}^-_t(X^j_t){\setminus}\{X^i_{t{-}\tau}\})$
        \item[] \hspace*{2.7em} PC: Set $\mathbf{Z}{=}\mathcal{S}$
        \State $(\text{$p$-value},\,I) \gets$ \Call{CI}{$X^i_{t{-}\tau},X^j_{t},\mathbf{Z}}$
        \State $I^{\min}(X^i_{t-\tau},X^j_t)=\min(|I|,I^{\min}(X^i_{t-\tau},X^j_t))$
        \If{$p$-value $> \alpha_{\rm PC}$}
            \State Delete link $X^i_{t-\tau}\to X^j_t$ for $\tau>0$ (or $X^i_{t}\oo X^j_t$ for $\tau=0$) from $\mathcal{G}$
            \State Store (unordered) ${\rm sepset}(X^i_{t-\tau}, X^j_t)=\mathcal{S}$ 
        \EndIf
    \EndWhile
    \State Let $p=p+1$
    \State Re-compute $\widehat{\mathcal{A}}(X^j_t)$ from $\mathcal{G}$ and sort by $I^{\min}(X^i_{t-\tau},X^j_t)$ from largest to smallest
\EndWhile
\State \Return $\mathcal{G}$, sepset
\end{algorithmic}
 \label{algo:step2}
\end{algorithm*}

\subsection{THEORETICAL RESULTS} \label{sec:consistency}
This section states asymptotic consistency, finite sample order-independence, and further results regarding effect size and false positive control.
The consistency of network learning algorithms is separated into \emph{soundness}, i.e., the returned graph has correct adjacencies, and \emph{completeness}, i.e., the returned graph is also maximally informative (links are oriented as much as possible).
We start with the following assumptions.
\begin{myassum}[Asymptotic case] \label{assum:asymptotic}
Throughout this paper we assume Causal Sufficiency, the Causal Markov Condition, the Adjacency Faithfulness Conditions, and consistent CI tests (oracle). In the present time series context we also assume stationarity and time-order and that the maximum time lag $\tau_{\max}\geq \tau^{\mathcal{P}}_{\max}$, where $ \tau^{\mathcal{P}}_{\max}$ is the maximum time lag of any parent in the SCM~\eqref{eq:causal_model}.
Furthermore, we rule out \emph{selection variables} and \emph{measurement error}.
\end{myassum}
Definitions of these assumptions, adapted from \cite{Spirtes2000} to the time series context, are in Sect.~\ref{sec:definitions} and all proofs are in Sect.~\ref{sec:proofs}. We start with the following lemma.
\begin{mylemma} \label{thm:step1b}
Under Assumptions~\ref{assum:asymptotic} Alg.~\ref{algo:step1} returns a set that always contains the lagged parents of $X^j_t$. If Adjacency Faithfulness is replaced by standard Faithfulness, the set contains exactly the lagged parents of $X^j_t$ and of all its contemporaneous ancestors $\Ct$, i.e.,
\begin{align*}
\B = \bigcup_{X^i_t \in \{X^j_t\}\cup\Ct} \Pa(X^i_t).
\end{align*}
\end{mylemma}

This establishes that the conditions $\widehat{\mathcal{B}}^-_t(X^j_t)$ estimated in the first phase of PCMCI$^+$ will suffice to block all lagged confounding paths that do not go through contemporaneous links.  This enables to prove the soundness of Alg.~\ref{algo:step2}, even though Alg.~\ref{algo:step2} is a variant of the PC algorithm that only iterates through contemporaneous conditioning sets.
\begin{mythm}[Soundness of PCMCI$^+$] \label{thm:soundness}
Algorithm~\ref{algo:step2} returns the correct adjacencies under Assumptions~\ref{assum:asymptotic}, i.e., $\widehat{\mathcal{G}^*}=\mathcal{G}^*$, where the $\mathcal{G}^*$ denotes the skeleton of the time series graph.
\end{mythm}
\noindent

To prove the completeness of PCMCI$^+$, we start with the following observation.
\begin{mylemma} \label{thm:orientation_triples}
Due to time-order and the stationarity assumption, the considered triples in the collider phase (Alg.~\ref{algo:step3_SM}) and rule orientation phase (Alg.~\ref{algo:step4_SM}) can be restricted as follows: In the collider orientation phase only unshielded triples $X^i_{t-\tau}\to X^k_t \oo X^j_t$ (for $\tau>0$) or $X^i_{t}\oo X^k_t \oo X^j_t$ (for $\tau=0$) in $\mathcal{G}$ where $(X^i_{t-\tau}, X^j_t)$ are not adjacent are relevant. For orientation rule R1 triples $X^i_{t-\tau}\to X^k_t \oo X^j_t$  where $(X^i_{t-\tau}, X^j_t)$ are not adjacent, for orientation rule R2 triples $X^i_{t}\to X^k_t \to X^j_t$ with $X^i_{t} \oo X^j_t$, and for orientation rule R3 pairs of  triples $X^i_{t}\oo X^k_t \to X^j_t$ and $X^i_{t}\oo X^l_t \to X^j_t$ where $(X^k_{t}, X^l_t)$ are not adjacent and $X^i_{t} \oo X^j_t$ are relevant. These restrictions imply that only contemporaneous parts of separating sets are relevant for the collider phase.
\end{mylemma}

\begin{mythm}[PCMCI$^+$ is complete] \label{thm:completeness}
PCMCI$^+$ (Algorithms~\ref{algo:step1},\ref{algo:step2},\ref{algo:step3_SM},\ref{algo:step4_SM}) when used with the conservative rule for orienting colliders in Alg.~\ref{algo:step3_SM} returns the correct CPDAG under Assumptions~\ref{assum:asymptotic}. Under standard Faithfulness also PCMCI$^+$ when used with the majority rule or the standard orientation rule is complete.
\end{mythm}

Also the proof of order-independence follows straightforwardly from the proof in \cite{Colombo2014}. Of course, order independence does not apply to time-order. 
\begin{mythm}[Order independence] \label{thm:order_independence}
Under Assumptions~\ref{assum:asymptotic} PCMCI$^+$ with the conservative or majority rule in Alg.~\ref{algo:step3_SM} is independent of the order of variables $(X^1,\ldots,X^N)$.
\end{mythm}

Next, we consider effect size.
The toy example showed that a major problem of PCMCI$^+_0$ (and also PC) is lack of detection power for contemporaneous links. A main factor of statistical detection power is effect size, i.e., the population value of the test statistic considered (e.g., absolute partial correlation). In the following, I will base my argument in an information-theoretic framework and consider the conditional mutual information as a general test statistic, denoted $I$. In Alg.~\ref{algo:step2} PCMCI$^+_0$ will test a contemporaneous dependency $X^i_{t} \oo X^j_t$ first with the test statistic $I(X^i_{t}; X^j_t|\widehat{\mathcal{B}}^-_t(X^j_t))$, and, if that test was positive, secondly with $I(X^i_{t}; X^j_t|\widehat{\mathcal{B}}^-_{t}(X^i_{t})))$. If either of these tests finds (conditional) independence, the adjacency is removed. Therefore, the minimum test statistic value determines the relevant effect size. On the other hand, PCMCI$^+$ treats both cases symmetrically since the test statistic is always $I(X^i_{t}; X^j_t|\widehat{\mathcal{B}}^-_t(X^j_t),\widehat{\mathcal{B}}^-_{t}(X^i_{t}))$. 

\begin{mythm}[Effect size of MCI tests for $p=0$] \label{thm:effect_size}
Under Assumptions~\ref{assum:asymptotic} the PCMCI$^+$ oracle case CI tests in Alg.~\ref{algo:step2} for $p=0$ for contemporaneous true links $X^i_{t}\to X^j_t \in \mathcal{G}$  have an effect size that is always greater than that of the PCMCI$^+_0$ CI tests, i.e., $I(X^i_{t}; X^j_t|\widehat{\mathcal{B}}^-_t(X^j_t),\widehat{\mathcal{B}}^-_{t}(X^i_{t}))> \min(I(X^i_{t}; X^j_t|\widehat{\mathcal{B}}^-_t(X^j_t)),\,I(X^i_{t}; X^j_t|\widehat{\mathcal{B}}^-_{t}(X^i_{t})))$ if both $X^i_{t}$ and $X^j_t$ have parents that are not shared with the other.
\end{mythm}
I conjecture that this result holds similarly for $p>0$ and also that PCMCI$^+$ has greater effect sizes than the PC algorithm since the latter iterates over \emph{all} subsets of adjacencies and, hence, the minimum is taken generally over an even larger set leading to even smaller effect sizes.
For lagged links the effect size of the PCMCI$^+$ tests is always smaller (or equal) than that of the PCMCI$^+_0$  tests (see \cite{Runge2012b}).

Last, we discuss false positive control.
While the effect size result regards detection power, in the following I give a mathematical intuition why the MCI tests are better calibrated than the PC algorithm CI tests and control false positives below the expected significance level. 
Lemma~\ref{thm:step1b} implies that even though Alg.~\ref{algo:step1} does not aim to estimate the contemporaneous parents, it still yields a set of conditions that shields $X^j_t$ from the `infinite' past $\mathbf{X}^-_t$, either by blocking the parents of $X^j_t$ or by blocking indirect contemporaneous paths through contemporaneous ancestors of $X^j_t$.
Blocking paths from the infinite past, I conjecture, is key to achieve well-calibrated CI tests in Alg.~\ref{algo:step2}. 
The authors in \cite{Runge2018d} showed that under certain model assumptions the MCI tests reduce to CI tests among the noise terms $\eta$ from model~\eqref{eq:causal_model} which are assumed to be i.i.d. and help to achieve well-calibrated CI tests. In the numerical experiments below we can see that the PC algorithm has inflated false positive for high autocorrelation, while PCMCI$^+$ well controls false positives, but a formal proof of correct false positive control for this challenging nonlinear, high-dimensional setting is beyond the scope of this paper.





\section{NUMERICAL EXPERIMENTS} \label{sec:numerics}
We consider a number of typical challenges \cite{Runge2019a}, contemporaneous and time lagged causal dependencies, strong autocorrelation, large number of variables and considered time lags, different noise distributions and nonlinearity, in the following additive variant of model~\eqref{eq:causal_model}:
\begin{align} \label{eq:numericalmodel}
X_t^j &= a_j X^j_{t-1} + \textstyle{\sum_i} c_i f_{i}(X^i_{t-\tau_i}) + \eta^j_t
\end{align}
for $j\in\{1,\ldots,N\}$. Autocorrelations $a_j$ are uniformly drawn from $[\max(0, a-0.3),\, a]$ for  $a$ as indicated in Fig.~\ref{fig:experiments} and $\eta^j$ is \emph{i.i.d.} and follows a zero-mean Gaussian $\mathcal{N}$ or Weibull $\mathcal{W}$ (scale parameter $2$) distribution (depending on setup) with standard deviation drawn from $[0.5,\, 2]$. In addition to autodependency links, for each model $L=\lfloor1.5\cdot N\rfloor$ (except for $N=2$ with $L=1$) cross-links are chosen whose functional dependencies are linear or  $f_i(x)=f^{(2)}(x)=(1+5 x e^{-x^2/20})x$ (depending on setup), with $f^{(2)}$ designed to yield more stationary dynamics. Coefficients $c_i$ are drawn uniformly from $\pm[0.1, 0.5]$. 30\% of the links are contemporaneous ($\tau_i=0$) and the remaining $\tau_i$ are drawn from $[1,\,5]$. Only stationary models are considered. We have an average cross-in-degree of $d=1.5$ for all network sizes (plus an auto-dependency) implying that models become sparser for larger $N$. We consider several model setups: linear Gaussian, linear mixed noise (among the $N$ variables: 50\% Gaussian, 50\% Weibull),  and nonlinear mixed noise (50\% linear, 50\% $f^{(2)}(x)$; 66\% Gaussian, 34\% Weibull).

For the linear model setups we consider the PC algorithm and PCMCI$^+$ in the majority-rule variant with ParCorr and compare these with GCresPC \cite{moneta2011causal}, a combination of GC with PC applied to residuals, and a autoregressive model version of LiNGAM \cite{hyvarinen2010estimation}, a representative of the SCM framework (implementation details in Sect.~\ref{sec:implementation_SM}). For the LiNGAM implementation I could not find a way to set a significance level and used the LASSO option which prunes `non-active' links to zero. Both GCresPC and LiNGAM assume linear dependencies and LiNGAM also non-Gaussianity. For the nonlinear setup the PC algorithm and PCMCI$^+$ are implemented with the GPDC test \cite{Runge2018d} that is based on Gaussian process regression and a distance correlation test on the residuals, which is suitable for a large class of nonlinear dependencies with additive noise.

Performance is evaluated as follows: True (TPR) and false positive rates (FPR, shown to evaluate false positive control, not applicable to LiNGAM) for adjacencies are distinguished between lagged cross-links ($i\neq j$), contemporaneous, and autodependency links. Due to time order, lagged links (and autodependencies) are automatically oriented. Contemporaneous orientation precision is measured as the fraction of correctly oriented links ($\oo$ or $\to$) among all estimated adjacencies, and recall as the fraction of correct orientations among all true contemporaneous links. Further shown is the fraction of conflicting links among all detected contemporaneous adjacencies (not applicable to LiNGAM). All metrics (and their std. errors) are computed across all estimated graphs from $500$ realizations of model~\eqref{eq:numericalmodel} at time series length $T$. The average runtimes were evaluated on Intel Xeon Platinum 8260.
In Fig.~\ref{fig:experiments} results for the linear Gaussian setup with default model parameters $N=5,\,T=500,\,a=0.95$ and method parameters $\tau_{\max}=5$ and $\alpha=0.01$ (not applicable to LiNGAM) are shown. Each of the four panels shows results for varying one of $a,\,N,\,T,\,\tau_{\max}$.
The insets show ANOVA statistics $r\pm \bar{\Delta} r$ [per unit], where $r$ is the performance metric at the leftmost parameter on the $x$-axis ($a,\,N,\,T,\,\tau_{\max}$, respectively) and $\bar{\Delta} r$ denotes the average change per parameter unit. In the adjacency subplots the statistics refer to lagged links.

\begin{figure*}[t]  
\centering
\includegraphics[width=.5\linewidth]{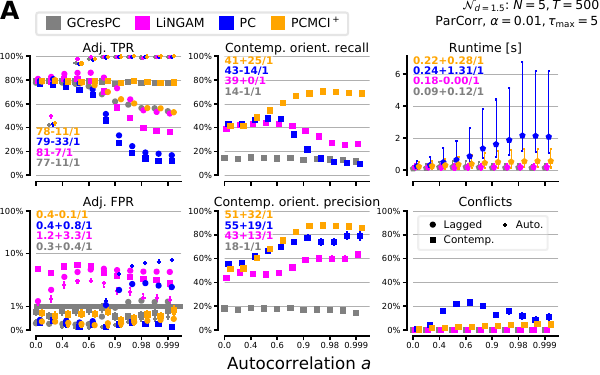}\hspace*{3pt}
\includegraphics[width=.5\linewidth]{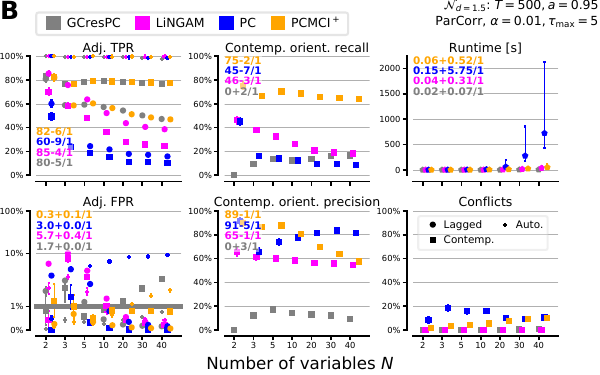}
\includegraphics[width=.5\linewidth]{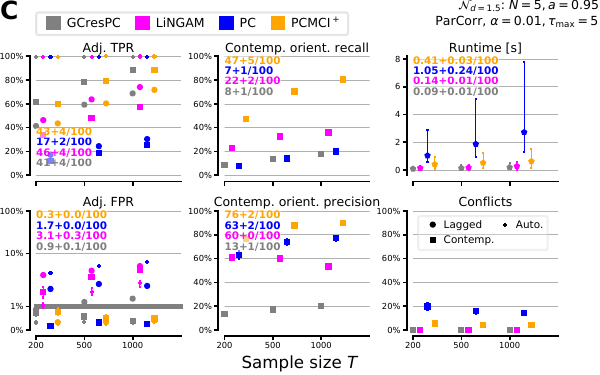}\hspace*{3pt}
\includegraphics[width=.5\linewidth]{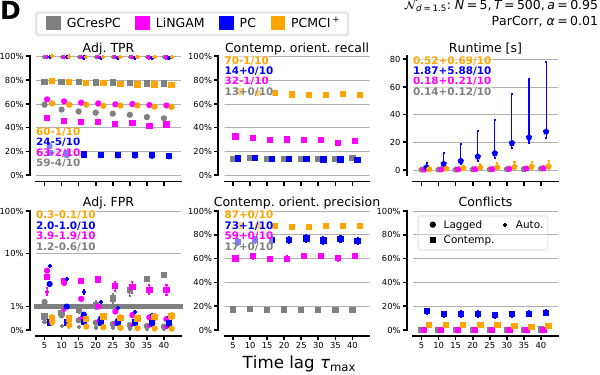}%
\caption{
Numerical experiments with linear Gaussian setup for varying (\textbf{A}) autocorrelation strength $a$ (\textbf{B}) number of variables $N$ (\textbf{C}) sample size $T$ and (\textbf{D}) maximum time lag $\tau_{\max}$. All remaining setup parameters indicated in the top right. Errorbars show std. errors or the 90\% range (for runtime). The insets show ANOVA statistics.}
\label{fig:experiments}
\end{figure*}

Figure~\ref{fig:experiments}A demonstrates that the TPR of PCMCI$^+$ and GCresPC for contemporaneous links is stable even under high autocorrelation while PC and LiNGAM show strong declines. Since LiNGAM has no $\alpha_{\rm PC}$ for FPR-control we focus on its relative changes rather than absolute performance. Lagged TPR decreases strongly for PC while the other methods are more robust.
FPR is well-controlled for PCMCI$^+$ while PC and slightly also GCresPC show inflated lagged FPR for high autocorrelation. LiNGAM features a strong increase of lagged FPR.
These adjacency results translate into higher contemporaneous orientation recall for PCMCI$^+$ which increases with autocorrelation, while it decreases for all other methods. GCresPC has steady low recall since it does not use lagged links in the orientation phase.
Except for GCresPC, all methods have increasing precision with PCMCI$^+$ and PC outperforming LiNGAM.
PCMCI$^+$ shows almost no conflicts while PC's conflicts increase with autocorrelation until low power reduces them again. Finally, runtimes are almost constant for GCresPC and LiNGAM, while they increase for PCMCI$^+$ and much stronger for PC.

Figure~\ref{fig:experiments}B shows that PCMCI$^+$ and GCresPC have the highest TPR for increasing number of variables $N$, especially for contemporaneous links.
FPR is well controlled only for PCMCI$^+$ while PC has false positives for small $N$ where model connectivity is denser and false negatives are more likely leading to false positives. For high $N$ PC has false positives only regarding autodependencies while inflated FPR appears for GCresPC. 
PCMCI$^+$ has more than twice as much contemporaneous recall compared to the other methods and is almost not affected by higher $N$.
Orientation precision is decreasing for all methods (except PC) with a higher decrease for PCMCI$^+$.
Runtime is increasing at a much smaller rate for PCMCI$^+$ compared to PC, which also has a very high runtime variability across the different model realizations. LiNGAM and especially GCresPC are fastest.

PCMCI$^+$, GCresPC, and LiNGAM benefit similarly and PC less so for increasing sample size regarding TPR (Fig.~\ref{fig:experiments}C). FPR is still not controlled for PC for large sample sizes, lagged FPR increases for GCresPC. PCMCI$^+$ shows the highest increases in contemporaneous recall and precision. Runtime increases are moderate compared to PC, conflicts decrease.

Last, Fig.~\ref{fig:experiments}D shows that all methods are relatively robust to large maximum time lags $\tau_{\max}$ (beyond the true max. time lag $5$) for the considered sample size $T=500$. Contemporaneous FPR and runtime increase for PC. 

In the SM further results are shown. For too large $N\tau_{\max}$ (relative to $T$) GCresPC and LiNGAM (despite LASSO-regularization) sharply drop in performance.
For the linear mixed noise setup (Fig.~\ref{fig:linearmixed_SM_overview}) results are almost unchanged for all methods except for LiNGAM for which recall and precision rise, as expected. Recall is then higher than PCMCI$^+$ for low autocorrelation, but still much lower for high autocorrelation and large $N$ or $\tau_{\max}$, at similar precision.

In the nonlinear mixed noise setup (Fig.~\ref{fig:nonlinearmixed_SM_overview}), the difference between PC and PCMCI$^+$ is similar. We observe slight FPR inflation for high autocorrelation. GPDC seems to not work well in high-dimensional, highly autocorrelated settings. Runtime for GPDC compared to ParCorr is orders of magnitude longer, especially for PC.
Further figures in the SM show many combinations of $a,\,N,\,T,\,\tau_{\max}$ and $\alpha_{\rm PC}$ for the model setups and demonstrate that the above findings are robust.

\section{CONCLUSIONS}
PCMCI$^+$ improves the reliability of CI tests by optimizing the choice of conditioning sets and yields much higher recall, well-controlled false positives, and faster runtime than the original PC algorithm for highly autocorrelated time series, while maintaining similar performance for low autocorrelation. 
The algorithm well exploits sparsity in high-dimensional settings and can flexibly be combined with different CI tests for nonlinear causal discovery, and for different variable types (discrete or continuous, univariate or multivariate).
Autocorrelation is actually key to increase contemporaneous orientation recall since it creates triples $X^i_{t-1}\to X^i_t \oo X^j_t$ that can often be oriented while an isolated link $X^i_t \oo X^j_t$  stays undirected in the Markov equivalence class, a drawback of CI-based methods. If the data is at least non-Gaussian, a SCM method like LiNGAM can exploit this property and recover directionality in such cases. Still, we saw that LiNGAM suffers from large autocorrelation.
PCMCI$^+$ is available as part of the \emph{tigramite} Python package at \url{https://github.com/jakobrunge/tigramite}. A next step will be to extend the present ideas to an algorithm accounting for latent confounders and to explore combinations between SCM-based methods and PCMCI$^+$. The numerical results will be contributed to the causality benchmark platform \verb|www.causeme.net| \cite{Runge2019a} to facilitate a further expanded method evaluation.


\subsubsection*{Acknowledgments}
DKRZ provided computational resources (grant no. 1083). I thank Andreas Gerhardus for helpful comments.

\clearpage
\renewcommand{\refname}{\selectfont\normalsize References} 





\clearpage
\appendix

\newtheorem{mythmm}{Theorem}

\renewcommand\theequation{S\arabic{equation}}
\setcounter{equation}{0}
\renewcommand\thefigure{S\arabic{figure}}    
\setcounter{figure}{0}   
\renewcommand\thesection{S\arabic{section}}    
\setcounter{section}{0} 
\renewcommand\thealgorithm{S\arabic{algorithm}}   
\setcounter{table}{0} 
\renewcommand\thetable{S\arabic{table}}    
\setcounter{algorithm}{1} 

\renewcommand\themyprop{S\arabic{myprop}}    
\renewcommand\themydef{S\arabic{mydef}}    
\setcounter{mydef}{0}
\renewcommand\themylemma{S\arabic{mylemma}}    
\setcounter{mylemma}{0}
\renewcommand\themycor{S\arabic{mycor}}    
\setcounter{mycor}{0}
\renewcommand\themythmm{S\arabic{mythmm}}    
\setcounter{mythmm}{0}

\newtheorem{innercustomgeneric}{\customgenericname}
\providecommand{\customgenericname}{}
\newcommand{\newcustomtheorem}[2]{%
  \newenvironment{#1}[1]
  {%
   \renewcommand\customgenericname{#2}%
   \renewcommand\theinnercustomgeneric{##1}%
   \innercustomgeneric
  }
  {\endinnercustomgeneric}
}

\newcustomtheorem{customlemma}{Lemma}
\newcustomtheorem{customtheorem}{Theorem}

\theoremstyle{definition}

\setcounter{section}{0}
\renewcommand{\thesection}{S\arabic{section}}
\setcounter{section}{-1}
\section{Erratum (October 2026)}

This section corrects an error in the termination criterion of Alg.~1 and, as a consequence, in Lemma~1. The error does not affect the consistency of PCMCI$^+$. Theorem~1 only relies on the superset property of Lemma~1 (Lemma~S1), which holds for both the published and the corrected algorithm. What the error affects is the upper bound in Lemma~1: the published Alg.~1 can return conditioning sets $\B$ that contain more than the lagged parents of $X^j_t$ and of its contemporaneous ancestors. The numerical results in Sect.~4 are also practically unchanged (see below).

\paragraph{Error in Alg.~1.}
The published Alg.~1 stopped as soon as $|\B\setminus\{X^i_{t-\tau}\}| < p$ for all candidates, also directly after an iteration in which candidates were removed. In such an iteration (at $p-1$) each remaining candidate was only tested conditional on the first $p-1$ variables of $\B$ before the removal, which may include the candidates removed in the same iteration. The remaining candidates are then never tested conditional on the final set $\B\setminus\{X^i_{t-\tau}\}$, and a lagged variable that is neither a lagged parent of $X^j_t$ nor of one of its contemporaneous ancestors can remain in $\B$. In the proof of Lemma~1 this enters in the step ``Alg.~1 removes $X^i_{t-\tau}$ iff $X^i_{t-\tau} \CI X^j_t \,|\, \mathcal{S}$ for some $\mathcal{S}\subseteq \B\setminus\{X^i_{t-\tau}\}$'', which holds in the ``only if'' direction, but not in the ``if'' direction. Lemma~S1 only uses the ``only if'' direction.

\paragraph{Example.}
Consider the linear model with independent standard Gaussian noise terms
\begin{align*}
X^1_t &= 0.5\,X^1_{t-1} + 0.5\,X^3_t - 0.3\,X^4_t + \eta^1_t, &
X^2_t &= -0.5\,X^2_{t-1} + \eta^2_t,\\
X^3_t &= -0.8\,X^2_t + \eta^3_t, &
X^4_t &= \eta^4_t,
\end{align*}
and $\tau_{\max}=1$. The contemporaneous ancestors of $X^1_t$ are $X^2_t,X^3_t,X^4_t$, and Lemma~1 states $\widehat{\mathcal{B}}^-_t(X^1_t) = \{X^1_{t-1}, X^2_{t-1}\}$. With oracle CI tests (and the population conditional mutual information as test statistic for the sorting) the published Alg.~1 removes no candidate at $p=0$ and $p=1$. At $p=2$, $X^3_{t-1}$ is removed (given $X^1_{t-1},X^2_{t-1}$), while $X^4_{t-1}$ is only tested given $X^1_{t-1},X^3_{t-1}$ and remains dependent via the opened collider path $X^4_{t-1}\to X^1_{t-1}\leftarrow X^1_{t-2}\leftarrow X^3_{t-2}\leftarrow X^2_{t-2}\to X^2_{t-1}\to X^2_t \to X^3_t \to X^1_t$. Now $|\widehat{\mathcal{B}}^-_t(X^1_t)|-1 = 2 < 3 = p$ and the algorithm stops with $\widehat{\mathcal{B}}^-_t(X^1_t) = \{X^1_{t-1}, X^2_{t-1}, X^4_{t-1}\}$. The corrected Alg.~1 repeats the iteration at $p=2$, removes $X^4_{t-1}$ (given $X^1_{t-1},X^2_{t-1}$), and stops after an iteration at $p=1$ without removals with $\widehat{\mathcal{B}}^-_t(X^1_t) = \{X^1_{t-1}, X^2_{t-1}\}$.

\paragraph{Correction.}
In the corrected Alg.~1 the last iteration is repeated at $p=|\B|-1$ whenever the previous iteration removed a candidate, and the algorithm only stops after an iteration without removals (marked lines in Alg.~1). The number of additional iterations is bounded by the number of removals, hence the complexity remains polynomial. Lemma~1 now states the upper bound under standard Faithfulness.

Note that the upper bound requires standard Faithfulness also for the corrected algorithm. Under Adjacency Faithfulness a lagged parent of a contemporaneous ancestor can be independent of $X^j_t$ (it is not adjacent to $X^j_t$) and be removed, which leaves the paths through it open. For example, take $X^1_t = -0.75\,X^2_{t-1} + 0.5\,X^4_t + \eta^1_t$, $X^3_t = X^2_t + X^5_{t-1} + \eta^3_t$, $X^4_t = 0.5\,X^3_{t-1} + \eta^4_t$, $X^2_t=\eta^2_t$, and $X^5_t=\eta^5_t$ with independent standard Gaussian noise terms ($\tau_{\max}=2$). Then $\mathrm{Cov}(X^3_{t-1},X^1_t) = -0.75 + 0.25\cdot 3 = 0$, so $X^3_{t-1}$ is removed at $p=0$, while $X^5_{t-2}$ remains dependent on $X^1_t$ via $X^5_{t-2}\to X^3_{t-1}\to X^4_t\to X^1_t$. Both versions of Alg.~1 return $\{X^2_{t-1}, X^5_{t-2}\}$, which is not a subset of $\{X^2_{t-1},X^3_{t-1}\}$.

\paragraph{Numerical experiments.}
The results in Sect.~4 and Figs.~3, S1--S15 were obtained with the published Alg.~1. I re-ran the experiments of Sect.~4 with the same model setup (linear Gaussian, linear mixed noise, and a reduced nonlinear setup with GPDC at $N=3$, $T=200$) for $\alpha_{\mathrm{PC}}\in\{0.01,0.05\}$ and 100--500 realizations per configuration. Each realization was analyzed with both the published and the corrected Alg.~1. The final PCMCI$^+$ graphs differ in almost no realizations for $a\leq 0.5$, and in up to 12\% of the realizations at $N=5$ (59\% at $N=20$) for $a\geq 0.9$. Still, in the linear setups no performance metric changes by more than $0.01$, and the paired differences are not significant beyond the rate expected by chance. The only systematic effect is a slightly lower lagged FPR for the corrected algorithm at $\alpha_{\mathrm{PC}}=0.05$ and $a\geq 0.9$ (0.5--5\% relative, median 3\%), since spurious lagged candidates retained by the published Alg.~1 occasionally passed the MCI tests. In the nonlinear setup (400 realizations) differences of up to $0.02$ occur, none of them significant (exact sign test, all $p\geq 0.13$). With oracle CI tests both versions recover the true Markov equivalence class in all tested models. The corrected algorithm requires 1--6\% more CI tests. Table~\ref{tab:erratum} shows the default configuration of Fig.~3. In summary, the conclusions of Sect.~4 remain unchanged.

\begin{table}[h]
\centering
\caption{Default configuration of Fig.~3 ($N=5$, $T=500$, $a=0.95$, $\tau_{\max}=5$, linear Gaussian, ParCorr, majority rule, 500 realizations) for the published and corrected Alg.~1 on the same data. $\Delta$ denotes the difference corrected $-$ published with paired bootstrap standard error.}
\label{tab:erratum}
\small
\setlength{\tabcolsep}{4pt}
\begin{tabular}{lccrccr}
\toprule
 & \multicolumn{3}{c}{$\alpha_{\mathrm{PC}}=0.01$} & \multicolumn{3}{c}{$\alpha_{\mathrm{PC}}=0.05$}\\
\cmidrule(lr){2-4}\cmidrule(lr){5-7}
 & published & corrected & $\Delta$ ($\pm$ s.e.) & published & corrected & $\Delta$ ($\pm$ s.e.)\\
\midrule
Lagged TPR            & 0.657  & 0.657  & $+0.0004\pm0.0007$ & 0.736  & 0.736  & $0.0000\pm0.0011$\\
Lagged FPR            & 0.0039 & 0.0039 & $0.0000\pm0.0001$  & 0.0072 & 0.0068 & $-0.0003\pm0.0001$\\
Contemp.\ TPR         & 0.869  & 0.868  & $-0.0010\pm0.0010$ & 0.916  & 0.916  & $0.0000\pm0.0000$\\
Contemp.\ FPR         & 0.0063 & 0.0063 & $0.0000\pm0.0000$  & 0.0370 & 0.0370 & $0.0000\pm0.0005$\\
Orient.\ precision    & 0.843  & 0.841  & $-0.0024\pm0.0029$ & 0.766  & 0.765  & $-0.0009\pm0.0025$\\
Orient.\ recall       & 0.754  & 0.751  & $-0.0030\pm0.0028$ & 0.815  & 0.814  & $-0.0010\pm0.0018$\\
Conflicts             & 0.039  & 0.039  & $0.0000\pm0.0023$  & 0.061  & 0.062  & $+0.0009\pm0.0017$\\
\# CI tests (mean)    & 375    & 388    &                    & 423    & 437    & \\
Final graph differs   & \multicolumn{3}{c}{3.2\% of realizations} & \multicolumn{3}{c}{8.2\% of realizations}\\
\bottomrule
\end{tabular}
\end{table}

\paragraph{Software.}
The corrected Alg.~1 is implemented in Tigramite from version~6.0 (\url{https://github.com/jakobrunge/tigramite}).

\section{Definitions} \label{sec:definitions}
The following definitions are adaptations of the standard assumptions of causal discovery to the time series case. Here we consider the causally sufficient case and assume that all variables $\mathbf{X}=(X^1,\ldots,X^N)$ of the underlying SCM~\eqref{eq:causal_model} are observed. Additionally, we assume that the maximum PCMCI$^+$ time lag $\tau_{\max}\geq \tau^{\mathcal{P}}_{\max}$, where $ \tau^{\mathcal{P}}_{\max}$ is the maximum time lag of any parent in the SCM~\eqref{eq:causal_model}.


\begin{mydef}[Causal Markov Condition] \label{defs:markov}
The joint distribution of a process $\mathbf{X}$ whose causal structure can be represented in a time series graph $\mathcal{G}$ fulfills the Causal Markov Condition iff for all $X^j_t\in \mathbf{X}_t$ every non-descendent of $X^j_t$ in $\mathcal{G}$ is independent of $X^j_t$ given the parents $\mathcal{P}(X^j_t)$. In particular, $\mathbf{X}_{t}^- {\setminus} \mathcal{P}(X^j_t) \ci X^j_t ~|~ \mathcal{P}(X^j_t)$ since all variables in $\mathbf{X}_{t}^-$ are non-descendants of $X^j_t$ by time-order.
\end{mydef}
Note that for the SCM~\eqref{eq:causal_model} with independent noise terms the Causal Markov Condition is automatically fulfilled.

\begin{mydef}[Adjacency and standard faithfulness Condition] \label{defs:faithfulness}
The joint distribution of a process $\mathbf{X}$ whose causal structure can be represented in a time series graph $\mathcal{G}$ fulfills the Adjacency Faithfulness Condition iff for all disjoint $X^i_{t-\tau}, X^j_t, \mathcal{S} \in \mathbf{X}^-_{t+1}$ with $\tau> 0$ 
\begin{align*}
    X^i_{t-\tau} &\ci X^j_t~|~\mathcal{S} ~\Rightarrow~ X^i_{t-\tau} \to X^j_t \notin \mathcal{G} \\
    X^i_{t-\tau} &\to X^j_t \in \mathcal{G}  ~\Rightarrow~  X^i_{t-\tau} \cancel{\ci} X^j_t~|~\mathcal{S}~~\text{(contrapositive)}
\end{align*} 
and with $\tau=0$ 
\begin{align*}
    X^i_{t} &\ci X^j_t~|~\mathcal{S} ~\Rightarrow~ X^i_{t} \oo X^j_t \notin \mathcal{G} \\
    X^i_{t} &\oo X^j_t \in \mathcal{G}  ~\Rightarrow~  X^i_{t} \cancel{\ci} X^j_t~|~\mathcal{S}~~\text{(contrapositive)}\,.
\end{align*} 
Furthermore, the variables fulfill the (standard) Faithfulness Condition iff for $\tau \geq 0$
\begin{align*}
    X^i_{t-\tau} &\ci X^j_t~|~\mathcal{S} ~\Rightarrow~ X^i_{t-\tau} \bowtie X^j_t ~|~\mathcal{S} \\
    X^i_{t-\tau} &\cancel{\bowtie} X^j_t ~|~\mathcal{S}  ~\Rightarrow~  X^i_{t-\tau} \cancel{\ci} X^j_t~|~\mathcal{S}~~\text{(contrapositive)}\,.
\end{align*} 
\end{mydef}

\section{Proofs} \label{sec:proofs}

\subsection{Proof of Lemma~\ref{thm:step1b}}
We first consider the following Lemma:
\begin{mylemma} \label{thm:step1}
Algorithm~\ref{algo:step1} returns a superset of lagged parents under Assumptions~\ref{assum:asymptotic}, i.e., $\mathcal{P}^-_t(X^j_t)  \subseteq \widehat{\mathcal{B}}^-_t(X^j_t)$ for all $X^j_{t}$ in $\mathbf{X}_t$.
\end{mylemma}
\begin{proof}  
We need to show that for arbitrary $(X^i_{t-\tau}, X^j_t)$ with $\tau>0$ we have $X^i_{t-\tau} \notin \widehat{\mathcal{B}}^-_t(X^j_t)~\Rightarrow~X^i_{t-\tau} \notin \mathcal{P}^-_t(X^j_t)$. 
Algorithm~\ref{algo:step1} removes $X^i_{t-\tau}$ from $\widehat{\mathcal{B}}^-_t(X^j_t)$ only if $X^i_{t-\tau} \ci X^j_t~|~\mathcal{S}$ for some $\mathcal{S} \subseteq \widehat{\mathcal{B}}^-_t(X^j_t){\setminus}\{X^i_{t{-}\tau}\}$ in the iterative CI tests. Then Adjacency Faithfulness directly implies that $X^i_{t-\tau}$ is not adjacent to $X^j_t$ and in particular $X^i_{t-\tau} \notin \mathcal{P}^-_t(X^j_t)$.
\end{proof}
With this step we can prove Lemma~\ref{thm:step1b}.

\paragraph{Proof of Lemma~1.}
Let $\mathcal{A} = \{X^j_t\}\cup\Ct$ and $\mathcal{L} = \bigcup_{X^i_t\in\mathcal{A}}\Pa(X^i_t)$. All conditioning sets in Alg.~1 are subsets of the lagged variables $\Xlag$. The superset property is Lemma~S1, the proof is unchanged. The upper bound follows from three claims.

\begin{claim}[Separation]
Let $x\in\Xlag\setminus\mathcal{L}$ and $\mathcal{S}\subseteq\Xlag$ with $\mathcal{L}\subseteq\mathcal{S}$ and $x\notin\mathcal{S}$. Then $x$ and $X^j_t$ are d-separated given $\mathcal{S}$ in the time series graph.
\end{claim}
\begin{proof}
Let $\pi$ be a path between $X^j_t$ and $x$. Since $x$ is lagged, $x\notin\mathcal{A}$. Let $w$ be the first vertex on $\pi$ (starting from $X^j_t$) that is not in $\mathcal{A}$, and $v\in\mathcal{A}$ its predecessor on $\pi$. (i) $w\to v$: Then $w$ is a parent of $v\in\mathcal{A}$. A contemporaneous parent of $v$ would be a contemporaneous ancestor of $X^j_t$, i.e., in $\mathcal{A}$. Hence $w$ is lagged and $w\in\Pa(v)\subseteq\mathcal{L}\subseteq\mathcal{S}$ (in particular $w\neq x$). Since $w$ is a non-collider on $\pi$, the path is blocked. (ii) $v\to w$: Then $w$ lies at a time $\geq t$. Since directed edges never point backward in time and $x$ lies at a time $<t$, the subpath from $w$ to $x$ cannot be directed away from $w$. Hence $\pi$ contains a collider $c$ that is a descendant of $w$. The collider $c$ and all its descendants lie at times $\geq t$ and are therefore not in $\mathcal{S}\subseteq\Xlag$, and the path is blocked.
\end{proof}

\begin{claim}[No member of $\mathcal{L}$ is removed]
Under standard Faithfulness no $\ell\in\mathcal{L}$ is ever removed by Alg.~1 (published or corrected).
\end{claim}
\begin{proof}
Let $\ell\in\Pa(a)$ with $a\in\mathcal{A}$. The path from $\ell$ via $a$ along a directed contemporaneous path to $X^j_t$ only has non-colliders at time $t$, none of which is in any $\mathcal{S}\subseteq\Xlag$. Hence $\ell$ and $X^j_t$ are d-connected given every $\mathcal{S}\subseteq\Xlag\setminus\{\ell\}$, and Faithfulness implies $\ell\nCI X^j_t\,|\,\mathcal{S}$.
\end{proof}

\begin{claim}[The last iteration tests the final set]
If the corrected Alg.~1 stops with $\B\neq\emptyset$, every $x\in\B$ has been tested conditional on $\B\setminus\{x\}$ in the last iteration without being removed.
\end{claim}
\begin{proof}
At termination $|\B|-1<p$ and $r=\mathrm{false}$, i.e., the last iteration removed nothing and $\B$ did not change during it. This iteration ran at $p' = p-1 \geq |\B|-1$, and by the reset $p'\leq|\B|-1$, hence $p'=|\B|-1$. For $p'=|\B|-1$ the conditioning set of each $x\in\B$ consists of all variables in $\B\setminus\{x\}$.
\end{proof}

\noindent\emph{Upper bound.} Assume standard Faithfulness and suppose $x\in\B\setminus\mathcal{L}$ when Alg.~1 stops. By Claim~2, $\mathcal{L}\subseteq\B\setminus\{x\}$. By Claim~1 and the Causal Markov Condition, $x\CI X^j_t\,|\,\B\setminus\{x\}$, and the test in the last iteration (Claim~3) would have removed $x$, a contradiction. Hence $\B\subseteq\mathcal{L}$, and together with Claim~2, $\B=\mathcal{L}$. \qed

\subsection{Proof of Theorem~\ref{thm:soundness}}
\begin{proof}  
The theorem states that under Assumptions~\ref{assum:asymptotic} $\widehat{\mathcal{G}^*}=\mathcal{G}^*$, where the $\mathcal{G}^*$ denotes the skeleton of the time series graph. We denote the two types of skeleton links $\to$ and $\oo$ here generically as $\starstar$ and can assume $\tau_{\max}\geq \tau \geq 0$. We need to show that for arbitrary $X^i_{t-\tau}, X^j_t \in \mathbf{X}^-_{t+1}$: (1)~$X^i_{t-\tau}\starstar X^j_t \notin \widehat{\mathcal{G}^*} ~\Rightarrow~ X^i_{t-\tau}\starstar X^j_t \notin \mathcal{G}^*$ and (2)~$X^i_{t-\tau}\starstar X^j_t \notin \mathcal{G}^* ~\Rightarrow~ X^i_{t-\tau}\starstar X^j_t \notin \widehat{\mathcal{G}^*}$. 

Ad (1): Algorithm~\ref{algo:step2} deletes a link $X^i_{t-\tau}\starstar X^j_t$ from $\widehat{\mathcal{G}^*}$ iff $X^i_{t-\tau} \ci X^j_t~|~\mathcal{S},\widehat{\mathcal{B}}^-_t(X^j_t){\setminus}\{X^i_{t{-}\tau}\},\widehat{\mathcal{B}}^-_{t{-}\tau}(X^i_{t{-}\tau})$ for some $\mathcal{S} \subseteq \widehat{\mathcal{A}}_t(X^j_t)$ in the iterative CI tests with $\widehat{\mathcal{B}}^-_t(X^j_t)$ estimated in Alg.~\ref{algo:step1}. $\widehat{\mathcal{A}}_t(X^j_t)$ denotes the contemporaneous adjacencies. Then Adjacency Faithfulness directly implies that $X^i_{t-\tau}$ is not adjacent to $X^j_t$: $X^i_{t-\tau} \starstar X^j_t \notin \mathcal{G}^*$.

Ad (2): By Lemma~\ref{thm:step1b} we know that $\widehat{\mathcal{B}}^-_t(X^j_t)$ is a superset of the lagged parents of $X^j_t$. Denote the lagged, extra conditions occurring in the CI tests of Alg.~\ref{algo:step2} as $W^-_t=(\widehat{\mathcal{B}}^-_t(X^j_t)\setminus \{X^i_{t-\tau}\},\widehat{\mathcal{B}}^-_{t-\tau}(X^i_{t-\tau}))\setminus \mathcal{P}(X^j_t)$. $W^-_t$ does not contain parents of $X^j_t$ and by the assumption also $X^i_{t-\tau}$ is not a parent of $X^j_t$. We further assume that for $\tau=0$ $X^i_{t}$ is also not a descendant of $X^j_t$ since that case is covered if we exchange $X^i_t$ and $X^j_t$. Then the Causal Markov Condition implies $(X^i_{t-\tau},W^-_t) \ci X^j_t~|~\mathcal{P}(X^j_t)$. By the weak union property of conditional independence this leads to $X^i_{t-\tau} \ci X^j_t~|~\mathcal{P}(X^j_t), W^-_t$ which is equivalent to $X^i_{t-\tau} \ci X^j_t~|~\mathcal{P}(X^j_t), \widehat{\mathcal{B}}^-_t(X^j_t)\setminus \{X^i_{t-\tau}\},\widehat{\mathcal{B}}^-_{t-\tau}(X^i_{t-\tau})$. Now Alg.~\ref{algo:step2} iteratively tests $X^i_{t-\tau} \ci X^j_t~|~\mathcal{S}, \widehat{\mathcal{B}}^-_t(X^j_t)\setminus \{X^i_{t-\tau}\},\widehat{\mathcal{B}}^-_{t-\tau}(X^i_{t-\tau})$ for all $\mathcal{S}\subseteq \widehat{\mathcal{A}_t}(X^j_t)$. By the first part of this proof, the estimated contemporaneous adjacencies are always a superset of the true contemporaneous adjacencies, i.e.,  $\mathcal{A}_t(X^j_t)\subseteq \widehat{\mathcal{A}_t}(X^j_t)$, and by Lemma~\ref{thm:step1b} $\widehat{\mathcal{B}}^-_t(X^j_t)$ is a superset of the lagged parents. Hence, at some iteration step $\mathcal{S}=\mathcal{P}_t(X^j_t)$ and Alg.~\ref{algo:step2} will find $X^i_{t-\tau} \ci X^j_t~|~\mathcal{P}(X^j_t), \widehat{\mathcal{B}}^-_t(X^j_t)\setminus \{X^i_{t-\tau}\},\widehat{\mathcal{B}}^-_{t-\tau}(X^i_{t-\tau})$ and remove $X^i_{t-\tau}\starstar X^j_t$ from $\widehat{\mathcal{G}^*}$.
\end{proof}
For empty conditioning sets $\mathcal{S}$ ($p=0$), Alg.~\ref{algo:step2} is equivalent to the MCI algorithm \cite{Runge2018d} with the slight change that the latter is initialized with a fully connected (lagged) graph, which has no effect asymptotically. In \cite{Runge2018d} the authors prove the consistency of PCMCI assuming no contemporaneous causal links under the standard Faithfulness Condition. The proof above implies that PCMCI is already consistent under the weaker Adjacency Faithfulness Condition.

\subsection{Proof of Lemma~\ref{thm:orientation_triples}}
\begin{proof}  
Time order and stationarity can be used to constrain the four cases as follows. Let us first consider a generic triple $X^i_{t_i} \starstar X^k_{t_k} \starstar X^j_{t_j}$. By stationarity we can fix $t=t_j$. We only need to consider cases with $t_i,t_k \leq t$. If $t_k > t_j$, the triple is oriented already by time order and the case  $t_i > t_j$ is symmetric.

The possible triples in the collider phase of the original PC algorithm are $X^i_{t_i}\starstar X^k_{t_k} \starstar X^j_t$ where $(X^i_{t_i}, X^j_t)$ are not adjacent.  
For $t_k<t$ the time-order constraint automatically orients $X^k_{t_k} \to X^j_{t}$ and hence $X^k_{t_k}$ is a parent of $X^j_t$ and must always be in the separating set that makes $X^i_{t_i}$ and $X^j_t$ independent. Hence we only need to consider $t_k=t$ and can set $\tau=t-t_i$ ($\tau_{\max}\geq \tau \geq 0$), leaving the two cases of unshielded triples $X^i_{t-\tau}\to X^k_t \oo X^j_t$ (for $\tau>0$) or $X^i_{t}\oo X^k_t \oo X^j_t$ (for $\tau=0$) in $\mathcal{G}$ where $(X^i_{t-\tau}, X^j_t)$ are not adjacent. Since $X^k_t$ is contemporaneous to $X^j_t$, this restriction implies that only contemporaneous parts of separating sets are relevant for the collider orientation phase.

For rule R1 in the orientation phase the original PC algorithm considers the remaining triples with $X^i_{t-\tau}\to X^k_t$ that were not oriented by the collider phase (or by time order). This leaves $X^i_{t-\tau}\to X^k_t \oo X^j_t$ where $\tau_{\max}\geq \tau \geq 0$.

For rule R2 the original PC algorithm considers $X^i_{t_i}\to X^k_{t_k} \to X^j_t$ with $X^i_{t_i} \oo X^j_t$. The latter type of link leads to $t_i=t$ and time order restricts the triples to $X^i_{t}\to X^k_t \to X^j_t$ with $X^i_{t} \oo X^j_t$.

For rule R3 the original PC algorithm considers $X^i_{t_i}\oo X^k_{t_k} \to X^j_t$ and $X^i_{t_i}\oo X^l_{t_l} \to X^j_t$ where $(X^k_{t_k}, X^l_{t_l})$ are not adjacent and $X^i_{t_i} \oo X^j_t$. The latter constraint leads to $t_i=t$ and $X^i_{t_i}\oo X^k_{t_k}$ and $X^i_{t_i}\oo X^k_{t_l}$ imply $t_k=t_l=t$. Hence we only need to check triples $X^i_{t}\oo X^k_t \to X^j_t$ and $X^i_{t}\oo X^l_t \to X^j_t$ where $(X^k_{t}, X^l_t)$ are not adjacent and $X^i_{t} \oo X^j_t$.
\end{proof}

\subsection{Proof of Theorem~\ref{thm:completeness}}
\begin{proof}  
We first consider the case under Assumptions~\ref{assum:asymptotic} with Adjacency Faithfulness and PCMCI$^+$ in conjunction with the conservative collider orientation rule in Alg.~\ref{algo:step3_SM}. We need to show that all separating sets estimated in Alg.~\ref{algo:step3_SM} during the conservative orientation rule are correct. From the soundness (Theorem~\ref{thm:soundness}) and correctness of the separating sets follows the correctness of the collider orientation phase and the rule orientation phase which implies the completeness.

By Lemma~\ref{thm:orientation_triples} we only need to prove that in Alg.~\ref{algo:step3_SM} for unshielded triples $X^i_{t-\tau}\to X^k_t \oo X^j_t$ (for $\tau>0$) or $X^i_{t}\oo X^k_t \oo X^j_t$ (for $\tau=0$) the separating sets among subsets of contemporaneous neighbors of $X^j_t$ and, if $\tau=0$, of $X^i_t$, are correct. Algorithm~\ref{algo:step3_SM} tests $X^i_{t-\tau}\ci X^j_t~|~\mathcal{S},\widehat{\mathcal{B}}^-_t(X^j_t){\setminus}\{X^i_{t{-}\tau}\},\widehat{\mathcal{B}}^-_{t{-}\tau}(X^i_{t{-}\tau})$ for all $\mathcal{S}{\subseteq}\widehat{\mathcal{A}}_t(X^j_t){\setminus}\{X^i_{t-\tau}\}$ and for all $\mathcal{S}{\subseteq}\widehat{\mathcal{A}}_t(X^i_t){\setminus}\{X^j_{t}\}$ (if $\tau=0$). Since PCMCI$^+$ is sound, all adjacency information is correct and since all CI tests are assumed correct, all information on  separating sets is correct. Furthermore, with the conservative rule those triples where only Adjacency Faithfulness, but not standard Faithfulness, holds will be correctly marked as ambiguous triples. 

Under standard Faithfulness the completeness requires to prove that PCMCI$^+$ without the conservative orientation rule yields correct separating set information. By Lemma~\ref{thm:orientation_triples} also here we need to consider only separating sets among subsets of contemporaneous neighbors of $X^j_t$. Algorithm~\ref{algo:step2} tests $X^i_{t-\tau}\ci X^j_t~|~\mathcal{S},\widehat{\mathcal{B}}^-_t(X^j_t){\setminus}\{X^i_{t{-}\tau}\},\widehat{\mathcal{B}}^-_{t{-}\tau}(X^i_{t{-}\tau})$ for all $\mathcal{S}{\subseteq}\widehat{\mathcal{A}}_t(X^j_t){\setminus}\{X^i_{t-\tau}\}$. And again, since PCMCI$^+$ is sound, all adjacency information is correct and since all CI tests are assumed correct, all information on separating sets is correct, from which the completeness for this case follows.
\end{proof}

\subsection{Proof of Theorem~\ref{thm:order_independence}}
\begin{proof} 
Order-independence follows straightforwardly from sticking to the PC algorithm version in \cite{Colombo2014}. In particular, Alg.~\ref{algo:step1} and Alg.~\ref{algo:step2} are order-independent since they are based on PC stable where adjacencies are removed only after each loop over conditions of cardinality $p$. Furthermore, the collider phase (Alg.~\ref{algo:step3_SM}) and rule orientation phase (Alg.~\ref{algo:step4_SM}) are order-independent by marking triples with inconsistent separating sets as ambiguous and consistently marking conflicting link orientations by  $\conflict$.
\end{proof}

\subsection{Proof of Theorem~\ref{thm:effect_size}}
\begin{proof} 
The theorem states that under Assumptions~\ref{assum:asymptotic} the effect size for the PCMCI$^+$ oracle case CI tests in Alg.~\ref{algo:step2} for $p=0$ for contemporaneous true links $X^i_{t}\to X^j_t \in \mathcal{G}$  is greater than that of PCMCI$^+_0$: $I(X^i_{t}; X^j_t|\widehat{\mathcal{B}}^-_t(X^j_t),\widehat{\mathcal{B}}^-_{t}(X^i_{t}))> \min(I(X^i_{t}; X^j_t|\widehat{\mathcal{B}}^-_t(X^j_t)),\,I(X^i_{t}; X^j_t|\widehat{\mathcal{B}}^-_{t}(X^i_{t})))$ if both $X^i_{t}$ and $X^j_t$ have parents that are not shared with the other. We will use an information-theoretic framework here and consider the conditional mutual information.

To prove this statement, we denote by $\mathcal{B}_i=\widehat{\mathcal{B}}^-_t(X^i_t)\setminus \widehat{\mathcal{B}}^-_t(X^j_t)$ the lagged conditions of $X^i_t$ that are not already contained in those of $X^j_t$ and, correspondingly, $\mathcal{B}_j=\widehat{\mathcal{B}}^-_t(X^j_t)\setminus \widehat{\mathcal{B}}^-_t(X^i_t)$. Since both $X^i_{t}$ and $X^j_t$ have parents that are not shared with the other and we assume the oracle case, both these sets are non-empty. Further, we denote the common lagged conditions as $\mathcal{B}_{ij}=\widehat{\mathcal{B}}^-_t(X^j_t)\cap \widehat{\mathcal{B}}^-_t(X^i_t)$ and make use of the following conditional independencies, which hold by the Markov assumption:
(1)~$\mathcal{B}_i \ci X^j_t | \mathcal{B}_j, \mathcal{B}_{ij}, X^i_t$ and (2)~$\mathcal{B}_j \ci X^i_t | \mathcal{B}_i, \mathcal{B}_{ij}$. We first prove that, given a contemporaneous true link $X^i_{t}\to X^j_t \in \mathcal{G}$, $I(X^i_{t}; X^j_t|\mathcal{B}_{ij},\mathcal{B}_{j}) > I(X^i_{t}; X^j_t|\mathcal{B}_{ij},\mathcal{B}_{i})$ by using the following two ways to apply the chain rule of conditional mutual information:
\begin{align}
& I(X^i_{t}, \mathcal{B}_{i}; X^j_t, \mathcal{B}_{j}|\mathcal{B}_{ij})= \nonumber\\
&=I(X^i_{t}, \mathcal{B}_{i};\mathcal{B}_{j}|\mathcal{B}_{ij}) + I(X^i_{t}, \mathcal{B}_{i}; X^j_t|\mathcal{B}_{ij} \mathcal{B}_{j}) \nonumber\\
&=I(\mathcal{B}_{i};\mathcal{B}_{j}|\mathcal{B}_{ij}) + \underbrace{I(X^i_{t};\mathcal{B}_{j}|\mathcal{B}_{ij},\mathcal{B}_{i})}_{=0~~\text{(Markov)}} \nonumber\\ 
&\phantom{=}+I(X^i_{t};X^j_{t}|\mathcal{B}_{ij},\mathcal{B}_{j}) + \underbrace{I(\mathcal{B}_{i};X^j_{t}|\mathcal{B}_{ij},\mathcal{B}_{j},X^i_t)}_{=0~~\text{(Markov)}} \label{one}
\end{align}
and
\begin{align}
& I(X^i_{t}, \mathcal{B}_{i}; X^j_t, \mathcal{B}_{j}|\mathcal{B}_{ij})= \nonumber\\
&=I(\mathcal{B}_{i};X^j_t,\mathcal{B}_{j}|\mathcal{B}_{ij}) + I(X^i_{t}; X^j_t,\mathcal{B}_{j}|\mathcal{B}_{ij} \mathcal{B}_{i}) \nonumber\\
&=I(\mathcal{B}_{i};\mathcal{B}_{j}|\mathcal{B}_{ij}) + \underbrace{I(\mathcal{B}_{i};X^j_t|\mathcal{B}_{ij},\mathcal{B}_{j})}_{> 0~~\text{since $X^i_t\to X^j_t$}} \nonumber\\ 
&\phantom{=}+I(X^i_{t};X^j_{t}|\mathcal{B}_{ij},\mathcal{B}_{i}) + \underbrace{I(X^i_{t};\mathcal{B}_{j}|\mathcal{B}_{ij},\mathcal{B}_{i},X^j_t)}_{> 0~~\text{since $X^i_t\to X^j_t$}} \label{two}
\end{align}
where \eqref{one} and \eqref{two} denote two different applications of the chain rule. From this is follows that $I(X^i_{t}; X^j_t|\mathcal{B}_{ij},\mathcal{B}_{j}) > I(X^i_{t}; X^j_t|\mathcal{B}_{ij},\mathcal{B}_{i})$.

Hence, it remains to prove that $I(X^i_{t}; X^j_t|\mathcal{B}_{ij},\mathcal{B}_{j},\mathcal{B}_{i}) > I(X^i_{t}; X^j_t|\mathcal{B}_{ij},\mathcal{B}_{i})$, which we also do by the chain rule: 
\begin{align}
& I(X^i_{t}; X^j_t, \mathcal{B}_{j}|\mathcal{B}_{ij},\mathcal{B}_{i})= \nonumber\\
&=I(X^i_{t};X^j_t|\mathcal{B}_{ij},\mathcal{B}_{i}) + \underbrace{I(X^i_{t};\mathcal{B}_{j}|\mathcal{B}_{ij},\mathcal{B}_{i},X^j_t)}_{> 0~~\text{since $X^i_t\to X^j_t$}} \\ 
&=\underbrace{I(X^i_{t};\mathcal{B}_{j}|\mathcal{B}_{ij},\mathcal{B}_{i})}_{=0~~\text{(Markov)}} + I(X^i_{t};X^j_{t}|\mathcal{B}_{ij},\mathcal{B}_{i},\mathcal{B}_{j})
\end{align}
\end{proof}


\section{Further pseudo code}
Algorithms~\ref{algo:step3_SM} and \ref{algo:step4_SM} detail the pseudo-code for the PCMCI$^+$ / PCMCI$^+_0$ / PC collider phase with different collider rules and the orientation phase.

\begin{algorithm*}[h!]
\caption{(Detailed PCMCI$^+$ / PCMCI$^+_0$ / PC collider phase with different collider rules)}
\begin{algorithmic}[1]
\Require $\mathcal{G}$ and sepset from Alg.~\ref{algo:step2}, rule $=\{$'none', 'conservative', 'majority'$\}$, time series dataset $\mathbf{X}=(X^1,\,\ldots,X^N)$, significance threshold $\alpha_{\rm PC}$, ${\rm CI}(X,\,Y,\,\mathbf{Z})$, PCMCI$^+$ / PCMCI$^+_0$: $\widehat{\mathcal{B}}^-_t(X^j_t)$ for all $X^j_{t}$ in $\mathbf{X}_t$
\ForAll{unshielded triples $X^i_{t-\tau}\to X^k_t \oo X^j_t$ ($\tau>0$) or $X^i_{t}\oo X^k_t \oo X^j_t$ ($\tau=0$) in $\mathcal{G}$ where $(X^i_{t-\tau}, X^j_t)$ are not adjacent} 
    \If{rule $=$ 'none'}
        \If{$X^k_t$ is not in sepset$(X^i_{t-\tau}, X^j_t)$}
            \State Orient $X^i_{t-\tau}\to X^k_t \oo X^j_t$ ($\tau>0$) or  $X^i_{t-\tau}\oo X^k_t \oo X^j_t$ ($\tau=0$) as $X^i_{t-\tau} \to X^k_t \leftarrow X^j_t$ 
        \EndIf
    \Else
    \State PCMCI$^+$ / PCMCI$^+_0$: Define contemporaneous adjacencies $\widehat{\mathcal{A}}(X^j_t)=\widehat{\mathcal{A}}_t(X^j_t)=\{X^i_{t}{\neq} X^j_t \in \mathbf{X}_t: X^i_t \oo X^j_t ~\text{in $\mathcal{G}$}\}$
    \item[] \hspace*{2.7em} PC: Define full adjacencies $\widehat{\mathcal{A}}(X^j_t)$ for all (lagged and contemporaneous) links in $\mathcal{G}$
        \ForAll{for all $\mathcal{S}{\subseteq}\widehat{\mathcal{A}}(X^j_t){\setminus}\{X^i_{t-\tau}\}$ and for all $\mathcal{S}{\subseteq}\widehat{\mathcal{A}}(X^i_t){\setminus}\{X^j_{t}\}$ (if $\tau=0$) }
          \State Evaluate \Call{CI}{$X^i_{t{-}\tau},X^j_{t},\mathbf{Z}}$ with
          \State PCMCI$^+$: $\mathbf{Z}{=}(\mathcal{S},\widehat{\mathcal{B}}^-_t(X^j_t){\setminus}\{X^i_{t{-}\tau}\},\widehat{\mathcal{B}}^-_{t{-}\tau}(X^i_{t{-}\tau}))$
          \item[] \hspace*{4.2em} PCMCI$^+_0$: $\mathbf{Z}{=}(\mathcal{S},\widehat{\mathcal{B}}^-_t(X^j_t){\setminus}\{X^i_{t{-}\tau}\})$
          \item[] \hspace*{4.2em} PC: $\mathbf{Z}{=}\mathcal{S}$
          \State Store all subsets $\mathcal{S}$ with $p$-value $> \alpha_{\rm PC}$ as separating subsets          
        \EndFor
        \If{no separating subsets are found} \State Mark triple as ambiguous
        \Else
            \State Compute fraction $n_k$ of separating subsets that contain $X^k_t$
            \If{rule $=$ 'conservative'}
                \State Orient triple as collider if $n_k{=}0$, leave unoriented if $n_k{=}1$, and mark as ambiguous if $0{<}n_k{<}1$
            \ElsIf{rule $=$ 'majority'}
                \State Orient triple as collider if $n_k{<}0.5$, leave unoriented if $n_k{>}0.5$, and mark as ambiguous if $n_k{=}0.5$
            \EndIf
            \State Mark links in $\mathcal{G}$ with conflicting orientations as $\conflict$
        \EndIf
    \EndIf
\EndFor
\State \Return $\mathcal{G}$, sepset, ambiguous triples, conflicting links
\end{algorithmic}
 \label{algo:step3_SM}
\end{algorithm*}

\begin{algorithm*}[h]
\caption{(Detailed PCMCI$^+$ / PCMCI$^+_0$ / PC rule orientation phase)}
\begin{algorithmic}[1]
\Require $\mathcal{G}$, ambiguous triples, conflicting links
\While{any unambiguous triples suitable for rules R1-R3 are remaining}
    \State Apply rule R1 (orient unshielded triples that are not colliders):
    \ForAll{unambiguous triples $X^i_{t-\tau}\to X^k_t \oo X^j_t$  where $(X^i_{t-\tau}, X^j_t)$ are not adjacent} 
        \State Orient as $X^i_{t-\tau}\to X^k_t \to X^j_t$
        \State Mark links  with conflicting orientations as  $\conflict$
    \EndFor
    \State Apply rule R2 (avoid cycles):
    \ForAll{unambiguous triples $X^i_{t}\to X^k_t \to X^j_t$ with $X^i_{t} \oo X^j_t$} 
        \State Orient as $X^i_{t} \to X^j_t$
        \State Mark links with conflicting orientations as  $\conflict$
    \EndFor
    \State Apply rule R3 (orient unshielded triples that are not colliders and avoid cycles):
    \ForAll{pairs of unambiguous triples $X^i_{t}\oo X^k_t \to X^j_t$ and $X^i_{t}\oo X^l_t \to X^j_t$ where $(X^k_{t}, X^l_t)$ are not adjacent and $X^i_{t} \oo X^j_t$} 
        \State Orient as $X^i_{t} \to X^j_t$
        \State Mark links with conflicting orientations as  $\conflict$
    \EndFor
\EndWhile
\State \Return $\mathcal{G}$, conflicting links
\end{algorithmic}
 \label{algo:step4_SM}
\end{algorithm*}

\clearpage

\section{Implementation details} \label{sec:implementation_SM}

In the linear and nonlinear numerical experiments PCMCI$^+$ is compared with the PC algorithm, both implemented with the appropriate CI test (ParCorr for the linear case, GPDC for the nonlinear case). For the linear numerical experiments we additionally consider representatives from two further frameworks: GCresPC, a combination of GC with PC applied to residuals, and an autoregressive model version of LiNGAM \cite{hyvarinen2010estimation}, a representative of the SCM framework. Their implementations are as follows.

\subsection{LiNGAM}
For LiNGAM the code was taken from \texttt{https://github.com/cdt15/lingam} which provides a class \texttt{VARLiNGAM}. The method was called follows:
\begin{verbatim}
Input: data, tau_max

model = lingam.VARLiNGAM(lags=tau_max, criterion=None, prune=True)
model.fit(data)
val_matrix = model.adjacency_matrices_.transpose(2,1,0)
graph = (val_matrix != 0.).astype('int') 

Output: graph
\end{verbatim}
The causal graph \verb|graph| encodes the causal relations in an array of shape \verb|(N, N, tau_max + 1)|.
The option \verb|criterion=None| just ignores the optional automatic selection of \verb|lags|, which is here set to the same \verb|tau_max| for all methods. 
I could not find a way to obtain p-values in the \texttt{VARLiNGAM} implementation, but with the parameter setting \verb|prune=True| the resulting adjacency matrices are regularized with an adaptive LASSO approach using the BIC criterion to find the optimal regularization hyper-parameter (\verb|sklearn.LassoLarsIC(criterion='bic')|). Non-zero adjacencies were then evaluated as causal links. Note that all other methods can be intercompared at different $\alpha_{\rm PC}$ levels while for comparison against LiNGAM we focus on its relative changes rather than absolute performance.

\subsection{GCresPC}
There was no code available for the method proposed in \cite{moneta2011causal}. The present implementation first fits a VAR model up to $\tau_{\max}$ and applies the PC algorithm on the residuals. To remove spurious lagged links (due to contemporaneous paths), the PC algorithm was additionally run on significant lagged and contemporaneous links, but the orientation phase was restricted to contemporaneous links, as proposed in \cite{moneta2011causal}.
The following Python pseudo-code utilizes functionality from the \texttt{tigramite} package, \texttt{numpy}, and \texttt{statsmodels}:
\begin{verbatim}
Input: data, tau_max, alpha
import functions/classes ParCorr, PCMCI, DataFrame from tigramite

graph = np.zeros((N, N, tau_max + 1))

# 1. Estimate lagged adjacencies (to be updated in step 3.)    
tsamodel = tsa.var.var_model.VAR(data)
results = tsamodel.fit(maxlags=tau_max, trend='nc')
pvalues = results.pvalues
values = results.coefs
residuals = results.resid

lagged_parents = significant lagged links at alpha

# 2. Run PC algorithm on residuals (with tau_max=0)
pcmci = PCMCI(dataframe=DataFrame(residuals), cond_ind_test=ParCorr())
pcmcires = pcmci.run_pcalg(pc_alpha=alpha,
                           tau_min=0,
                           tau_max=0)

# Update contemporaneous graph
graph[:,:,0] = pcmcires['graph'][:,:,0]

# 3. Run PC algorithm on significant lagged and contemporaneous adjacencies
# to remove spurious lagged links due to contemporaneous parents 

selected_links = lagged_parents + significant contemporaneous adjacencies
pcmci = PCMCI(dataframe=DataFrame(data), cond_ind_test=ParCorr())
pcmcires = pcmci.run_pcalg(selected_links=selected_links, 
                           pc_alpha=alpha,
                           tau_min=0,
                           tau_max=tau_max)

# Update lagged part of graph
graph[:,:,1:] = pcmcires['graph'][:,:,1:]

Output: graph
\end{verbatim}
Note that the contemporaneous graph structure in \verb|graph| comes only from applying the PC algorithm to the residuals and, hence, does not utilize triples containing lagged adjacencies. Step 3 is necessary to remove spurious lagged links due to contemporaneous parents.
The output of GCresPC depends on $\alpha_{\rm PC}$ as for PCMCI$^+$ and the PC algorithm.

\section{Further numerical experiments} \label{sec:numerics_SM}

Next to repeating the overview figure for the linear Gaussian model setup from the main text in  Fig.~\ref{fig:lineargaussian_SM_overview}, in Fig.~\ref{fig:linearmixed_SM_overview} we show the linear mixed noise setup, and in Fig.~\ref{fig:nonlinearmixed_SM_overview} the nonlinear mixed noise setup.
The remaining pages contain results of further numerical experiments that evaluate different $a,\,N,\,T,\,\tau_{\max}$ and $\alpha_{\rm PC}$ for the linear model setups.
All results and more will be contributed to the causality benchmark platform \verb|www.causeme.net| \cite{Runge2019a} to facilitate a further expanded method evaluation.




\clearpage
\begin{figure*}[t]  
\centering
\includegraphics[width=.5\linewidth]{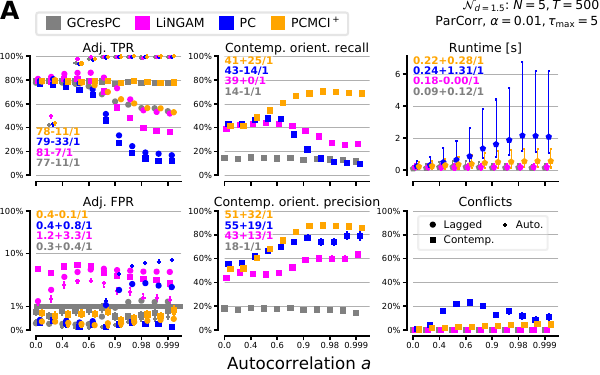}\hspace*{3pt}
\includegraphics[width=.5\linewidth]{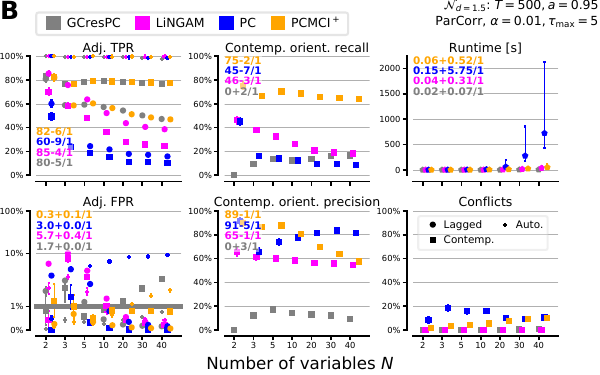}
\includegraphics[width=.5\linewidth]{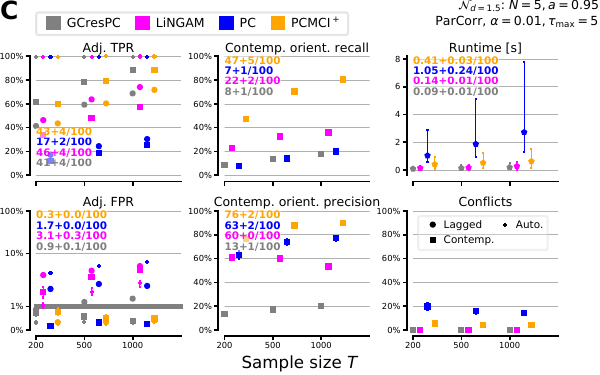}\hspace*{3pt}
\includegraphics[width=.5\linewidth]{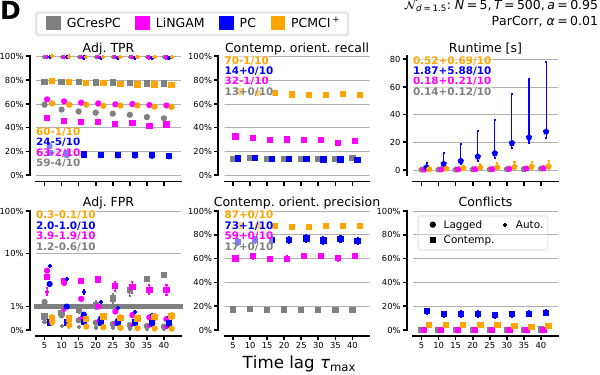}%
\caption{
Numerical experiments with linear Gaussian setup for varying (\textbf{A}) autocorrelation strength $a$ (\textbf{B}) number of variables $N$ (\textbf{C}) sample size $T$ and (\textbf{D}) maximum time lag $\tau_{\max}$. All remaining setup parameters indicated in the top right. Errorbars show std. errors or the 90\% range (for runtime). The insets show ANOVA statistics.}
\label{fig:lineargaussian_SM_overview}
\end{figure*}

\clearpage
\begin{figure*}[t]  
\centering
\includegraphics[width=.5\linewidth]{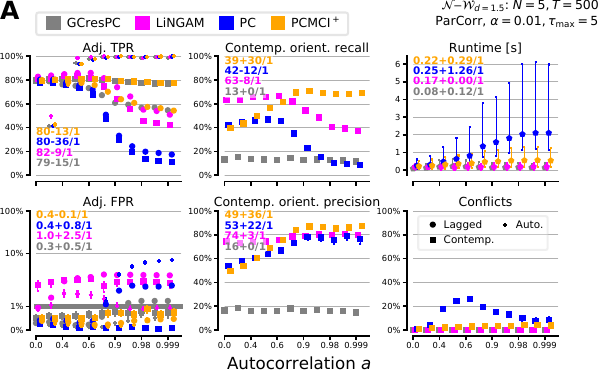}\hspace*{3pt}
\includegraphics[width=.5\linewidth]{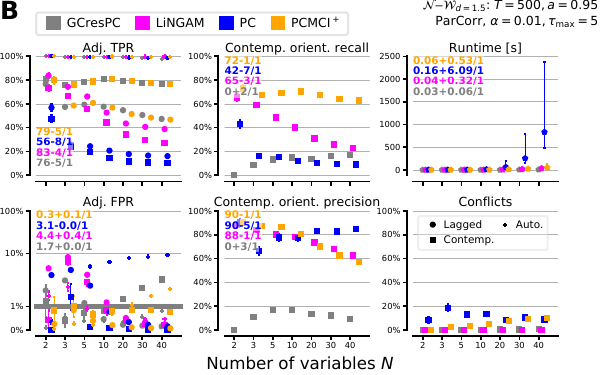}
\includegraphics[width=.5\linewidth]{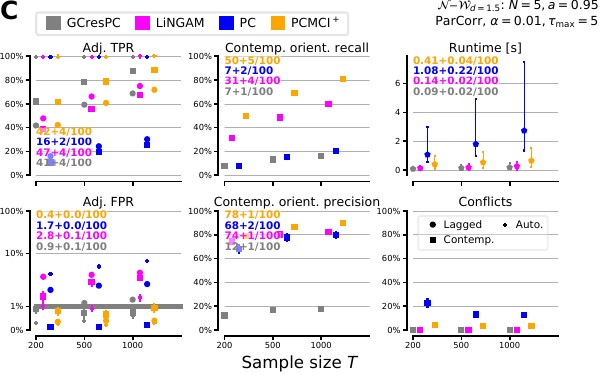}\hspace*{3pt}
\includegraphics[width=.5\linewidth]{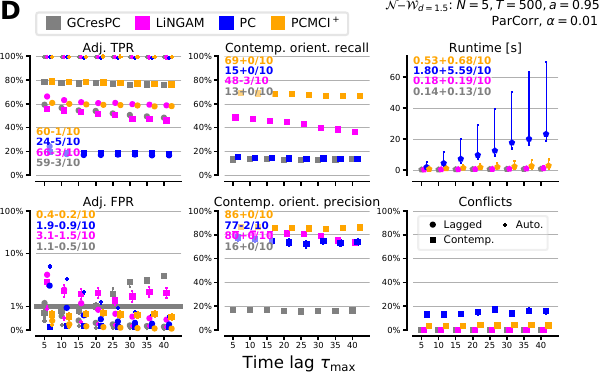}%
\caption{
Numerical experiments with linear mixed noise setup for varying (\textbf{A}) autocorrelation strength $a$ (\textbf{B}) number of variables $N$ (\textbf{C}) sample size $T$ and (\textbf{D}) maximum time lag $\tau_{\max}$. All remaining setup parameters indicated in the top right. Errorbars show std. errors or the 90\% range (for runtime). The insets show ANOVA statistics.}
\label{fig:linearmixed_SM_overview}
\end{figure*}

\clearpage
\begin{figure*}[t]  
\centering
\includegraphics[width=.5\linewidth]{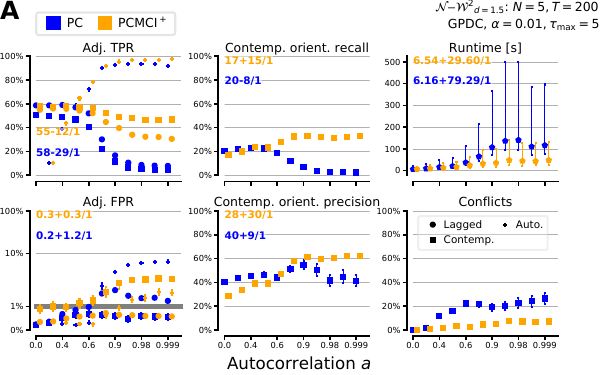}\hspace*{3pt}
\includegraphics[width=.5\linewidth]{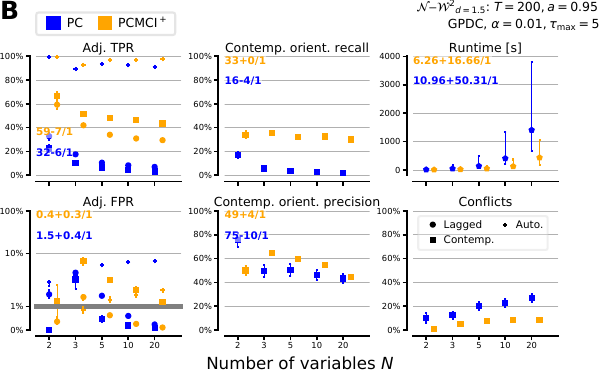}
\includegraphics[width=.5\linewidth]{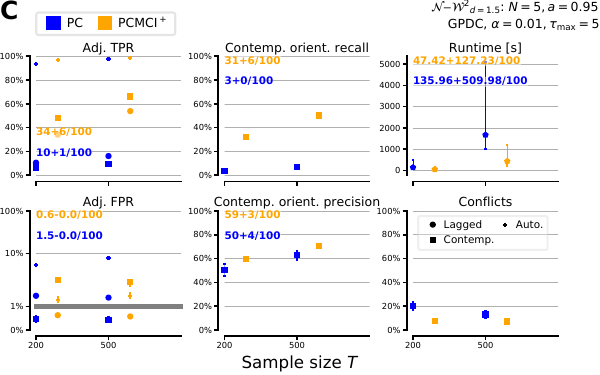}\hspace*{3pt}
\includegraphics[width=.5\linewidth]{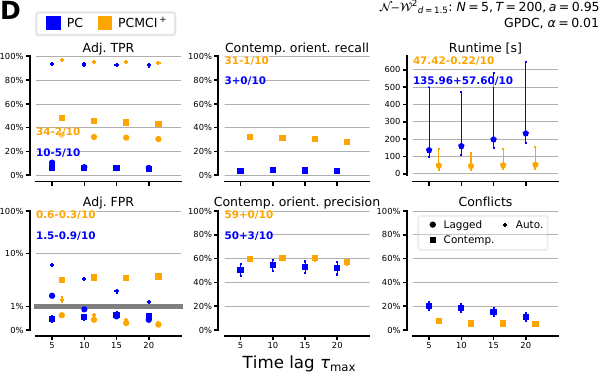}%
\caption{
Numerical experiments with nonlinear mixed noise setup for varying (\textbf{A}) autocorrelation strength $a$ (\textbf{B}) number of variables $N$ (\textbf{C}) sample size $T$ and (\textbf{D}) maximum time lag $\tau_{\max}$. All remaining setup parameters indicated in the top right. Errorbars show std. errors or the 90\% range (for runtime). The insets show ANOVA statistics.}
\label{fig:nonlinearmixed_SM_overview}
\end{figure*}


\clearpage
\begin{figure*}[t]  
\centering
\includegraphics[width=1\linewidth,page=1]{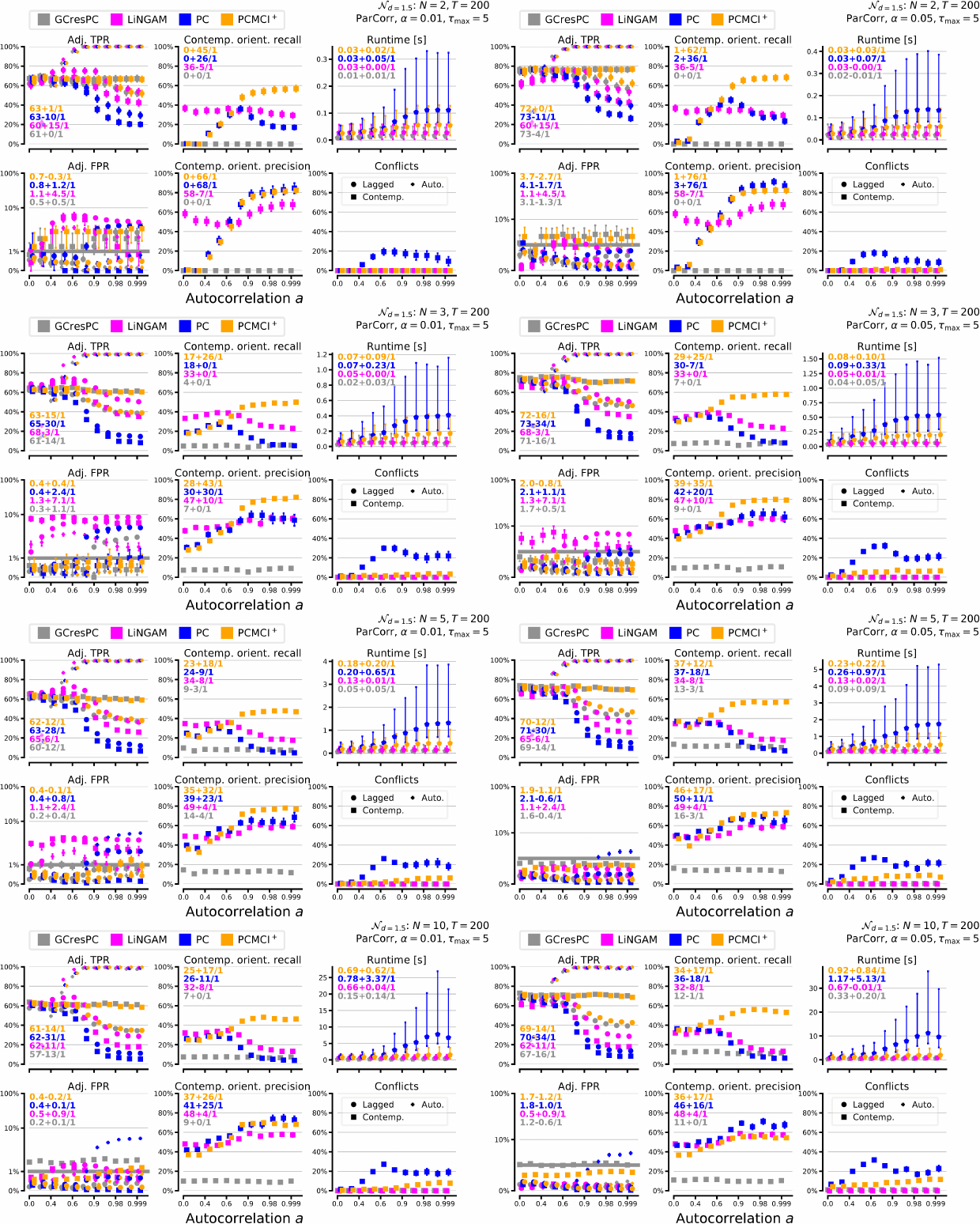}
\caption{
Numerical experiments with 
linear Gaussian setup
for varying 
autocorrelation $a$
 and $T=200$
. The left (right) column shows results for significance level $\alpha=0.01$ ($\alpha=0.05$). 
The rows depict results for $N=2,\,3,\,5,\, 10$ (top to bottom).
All model and method parameters are indicated in the upper right of each panel.
}
\label{fig:lineargaussian_SM_autocorr1}
\end{figure*}

\clearpage
\begin{figure*}[t]  
\centering
\includegraphics[width=1\linewidth,page=2]{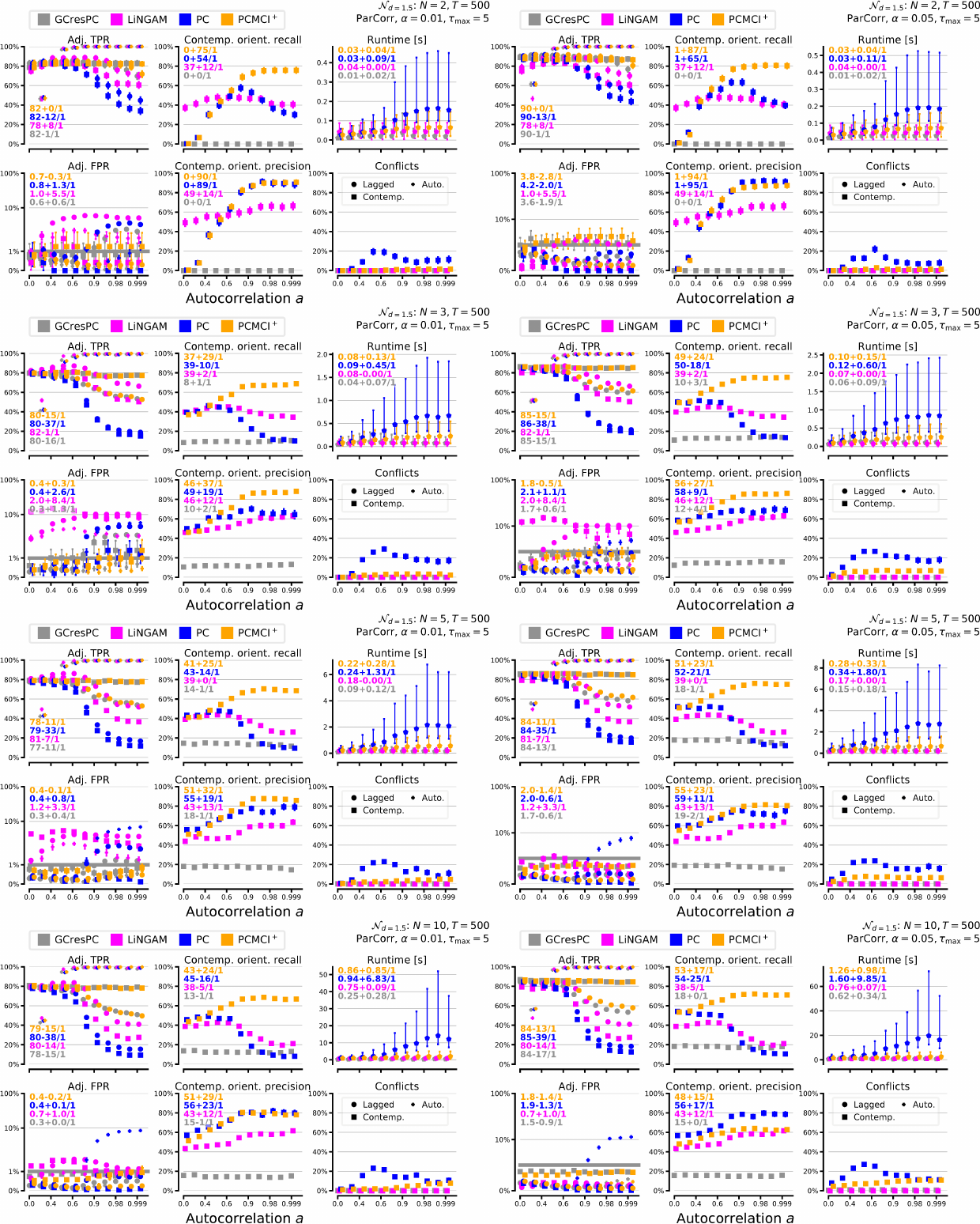}
\caption{
Numerical experiments with 
linear Gaussian  setup
for varying 
autocorrelation $a$
  and $T=500$
. The left (right) column shows results for significance level $\alpha=0.01$ ($\alpha=0.05$). 
The rows depict results for $N=2,\,3,\,5,\, 10$ (top to bottom).
All model and method parameters are indicated in the upper right of each panel.
}
\label{fig:lineargaussian_SM_autocorr2}
\end{figure*}

\clearpage
\begin{figure*}[t]  
\centering
\includegraphics[width=1\linewidth,page=3]{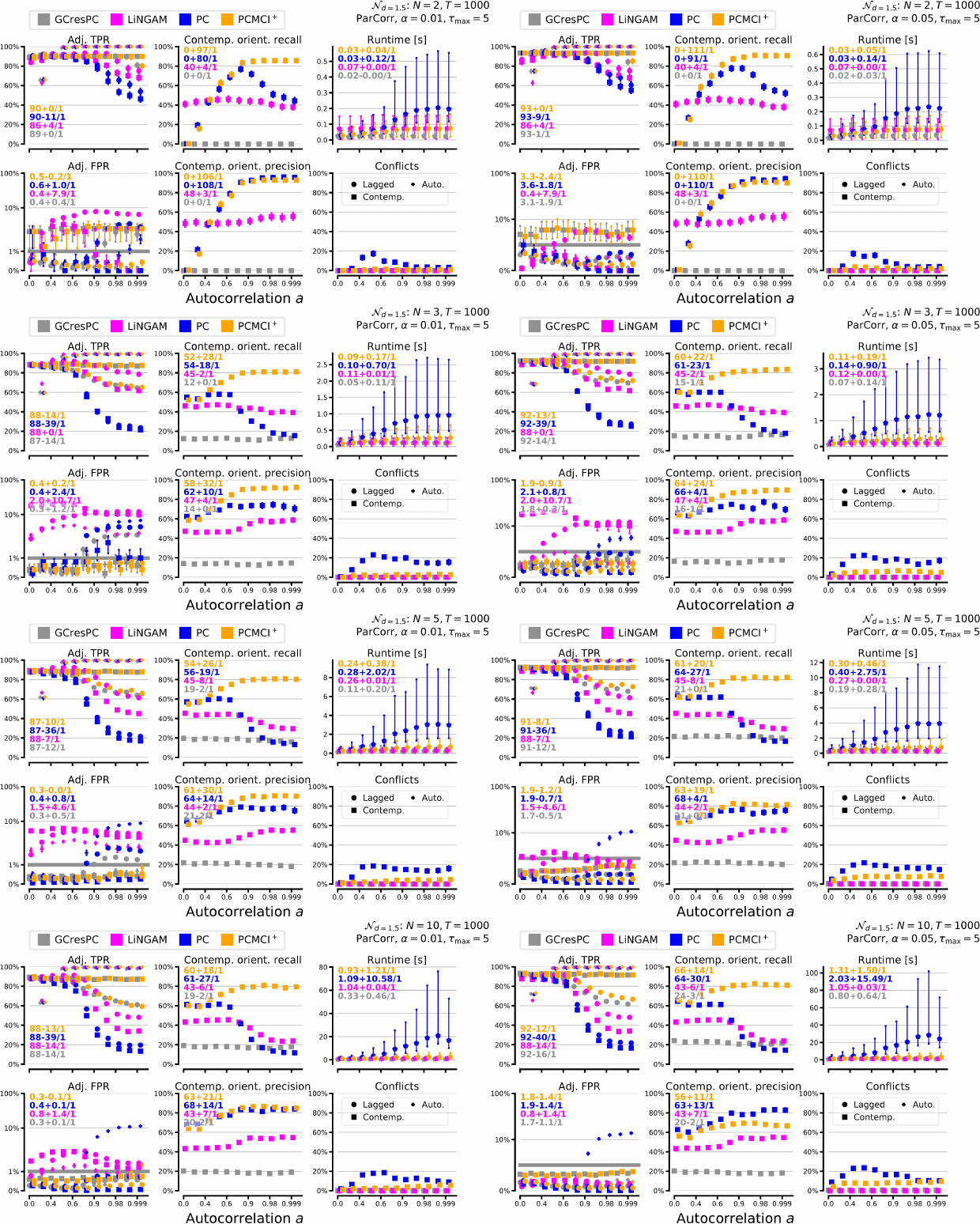}
\caption{
Numerical experiments with
linear Gaussian  setup
for varying 
autocorrelation $a$
   and $T=1000$
. The left (right) column shows results for significance level $\alpha=0.01$ ($\alpha=0.05$). 
The rows depict results for $N=2,\,3,\,5,\, 10$ (top to bottom).
All model and method parameters are indicated in the upper right of each panel.
}
\label{fig:lineargaussian_SM_autocorr3}
\end{figure*}

\clearpage
\begin{figure*}[t]  
\centering
\includegraphics[width=1\linewidth,page=1]{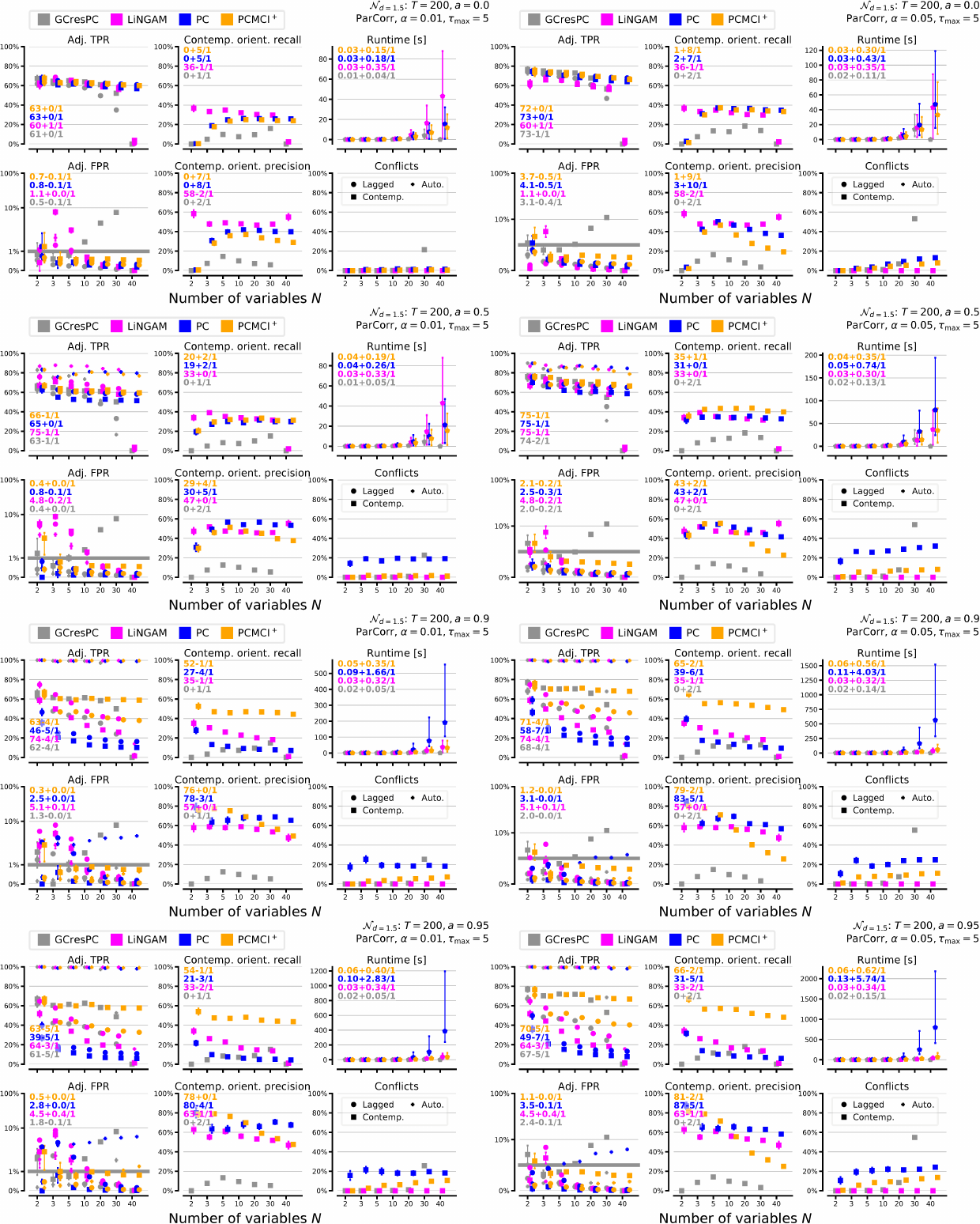}
\caption{
Numerical experiments with 
linear Gaussian  setup
for varying 
 number of variables $N$
 and $T=200$
. The left (right) column shows results for significance level $\alpha=0.01$ ($\alpha=0.05$). 
The rows depict results for increasing autocorrelations $a$ (top to bottom).
All model and method parameters are indicated in the upper right of each panel.
}
\label{fig:lineargaussian_SM_highdim1}
\end{figure*}

\clearpage
\begin{figure*}[t]  
\centering
\includegraphics[width=1\linewidth,page=2]{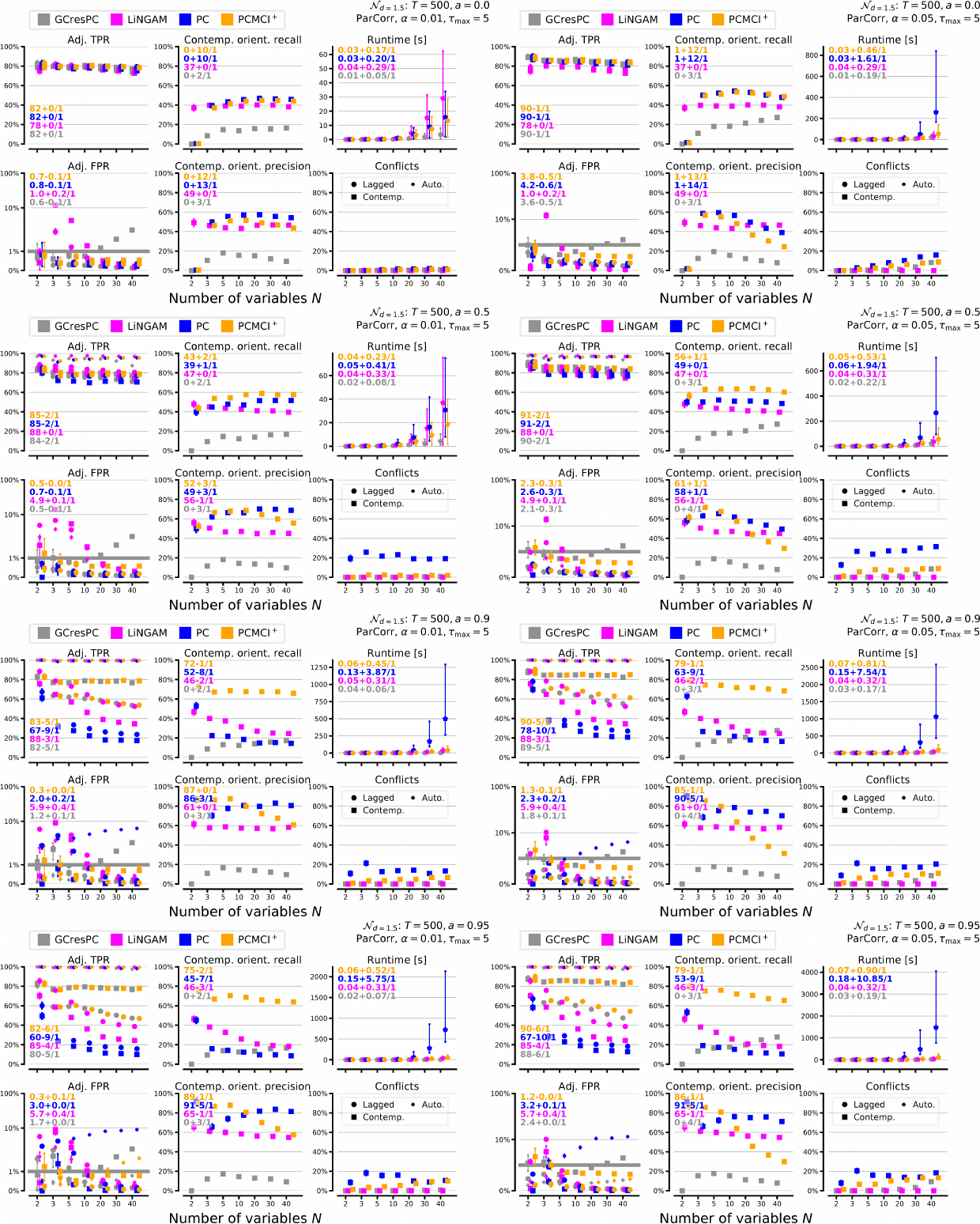}
\caption{
Numerical experiments with 
linear Gaussian  setup
for varying 
 number of variables $N$
  and $T=500$
. The left (right) column shows results for significance level $\alpha=0.01$ ($\alpha=0.05$). 
The rows depict results for increasing autocorrelations $a$ (top to bottom).
All model and method parameters are indicated in the upper right of each panel.
}
\label{fig:lineargaussian_SM_highdim2}
\end{figure*}

\clearpage
\begin{figure*}[t]  
\centering
\includegraphics[width=1\linewidth,page=3]{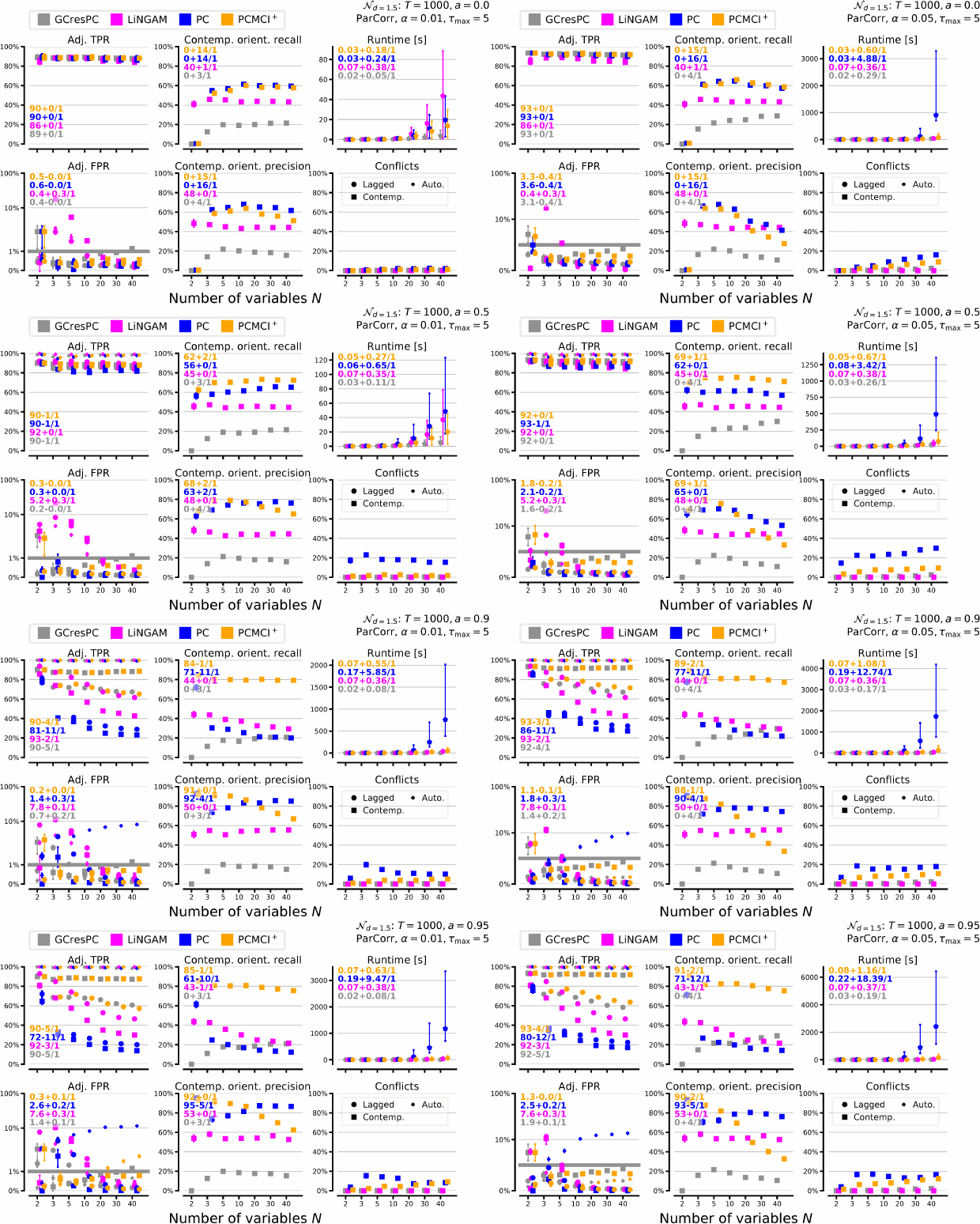}
\caption{
Numerical experiments with 
linear Gaussian  setup
for varying 
 number of variables $N$
   and $T=1000$
. The left (right) column shows results for significance level $\alpha=0.01$ ($\alpha=0.05$). 
The rows depict results for increasing autocorrelations $a$ (top to bottom).
All model and method parameters are indicated in the upper right of each panel.
}
\label{fig:lineargaussian_SM_highdim3}
\end{figure*}

\clearpage
\begin{figure*}[t]  
\centering
\includegraphics[width=1\linewidth,page=1]{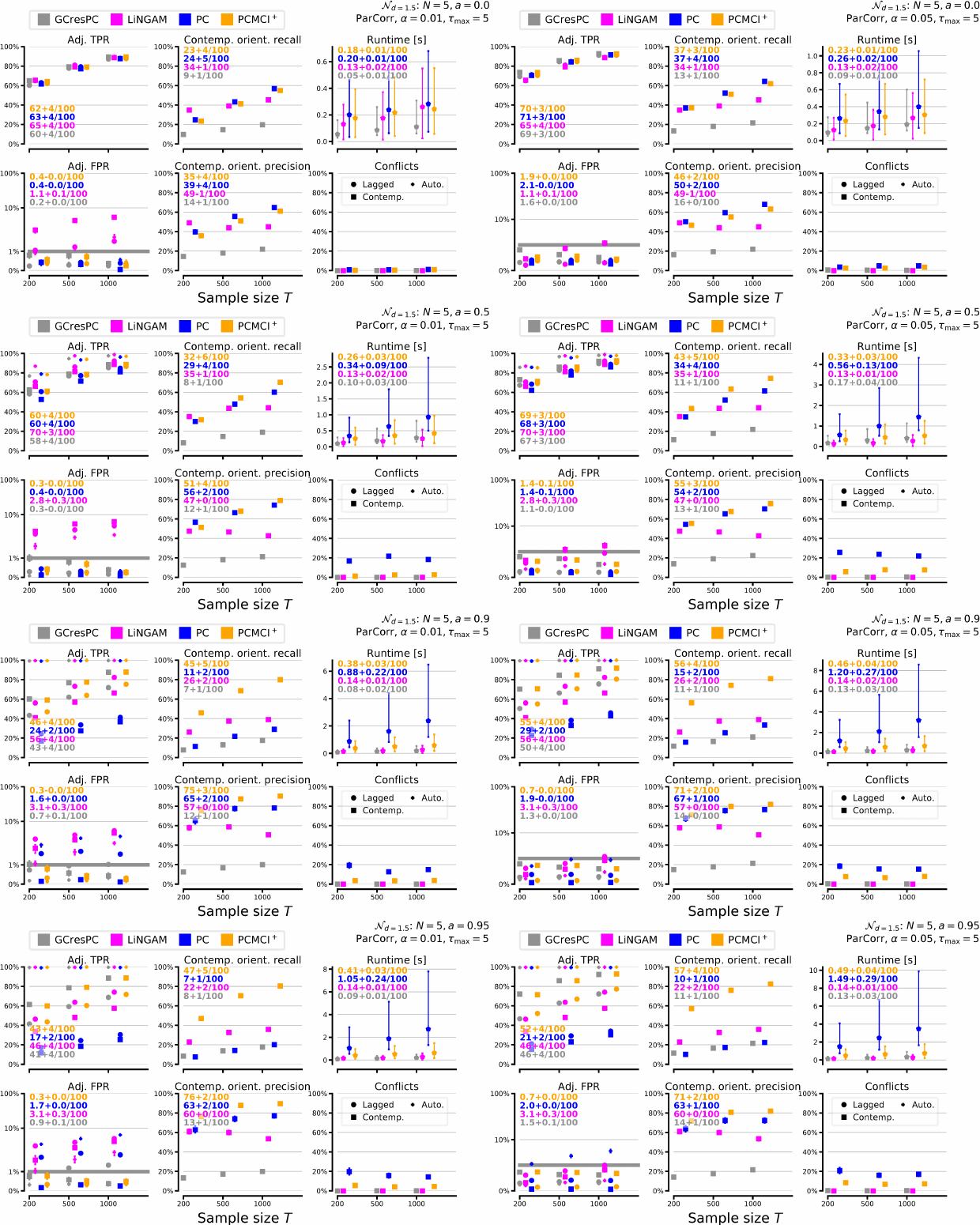}
\caption{
Numerical experiments with 
linear Gaussian  setup
for varying 
 sample size $T$
 for $N=5$ 
. The left (right) column shows results for significance level $\alpha=0.01$ ($\alpha=0.05$). 
The rows depict results for increasing autocorrelations $a$ (top to bottom).
All model and method parameters are indicated in the upper right of each panel.
}
\label{fig:lineargaussian_SM_sample_size1}
\end{figure*}

\clearpage
\begin{figure*}[t]  
\centering
\includegraphics[width=1\linewidth,page=2]{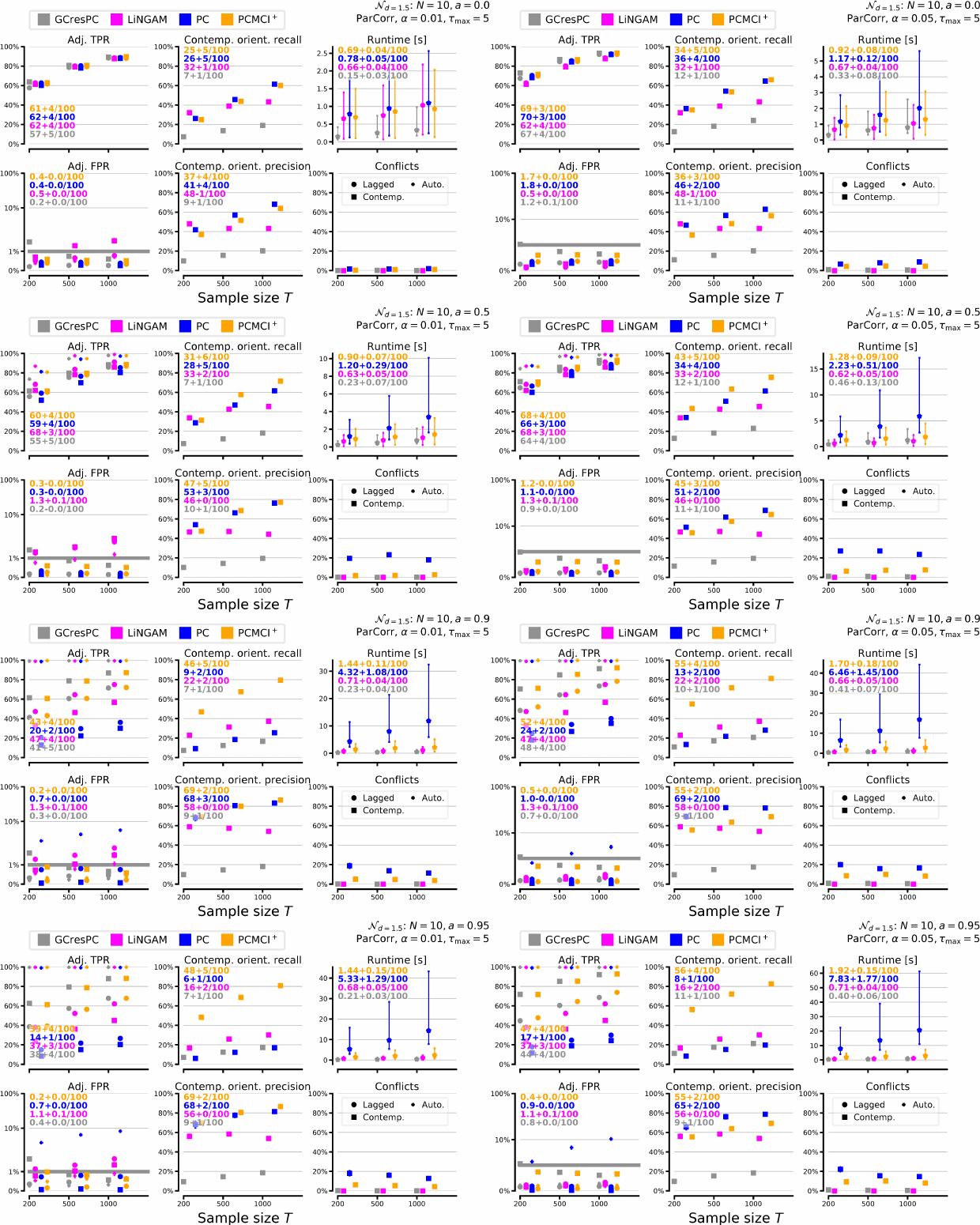}
\caption{
Numerical experiments with 
linear Gaussian  setup
for varying 
 sample size $T$
 for $N=10$ 
. The left (right) column shows results for significance level $\alpha=0.01$ ($\alpha=0.05$). 
The rows depict results for increasing autocorrelations $a$ (top to bottom).
All model and method parameters are indicated in the upper right of each panel.
}
\label{fig:lineargaussian_SM_sample_size2}
\end{figure*}

\clearpage
\begin{figure*}[t]  
\centering
\includegraphics[width=1\linewidth,page=3]{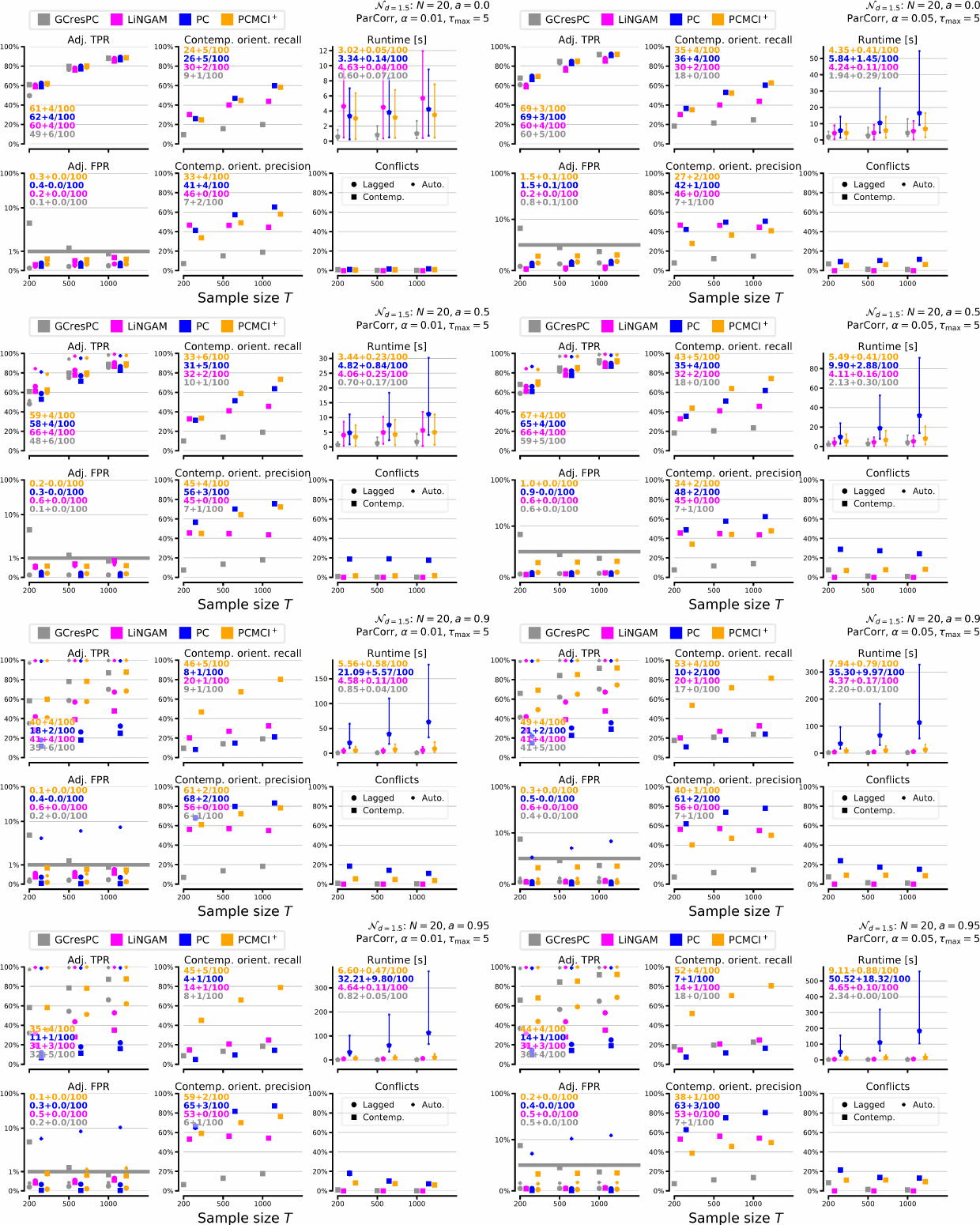}
\caption{
Numerical experiments with 
linear Gaussian  setup
for varying 
 sample size $T$
 for $N=20$ 
. The left (right) column shows results for significance level $\alpha=0.01$ ($\alpha=0.05$). 
The rows depict results for increasing autocorrelations $a$ (top to bottom).
All model and method parameters are indicated in the upper right of each panel.
}
\label{fig:lineargaussian_SM_sample_size3}
\end{figure*}

\clearpage
\begin{figure*}[t]  
\centering
\includegraphics[width=1\linewidth,page=1]{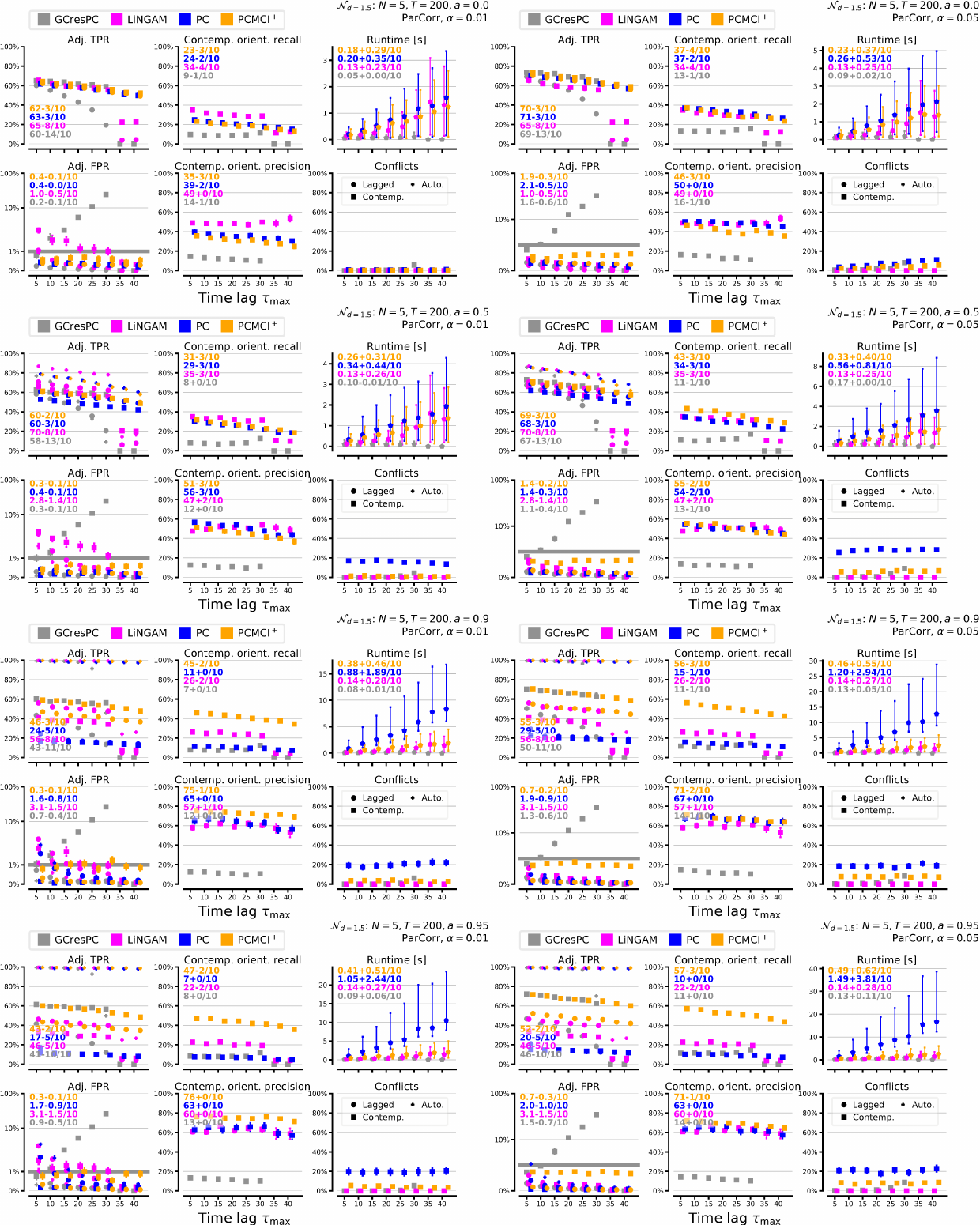}
\caption{
Numerical experiments with 
linear Gaussian  setup
for varying 
 maximum time lag $\tau_{\max}$ 
 and $T=200$
. The left (right) column shows results for significance level $\alpha=0.01$ ($\alpha=0.05$). 
The rows depict results for increasing autocorrelations $a$ (top to bottom).
All model and method parameters are indicated in the upper right of each panel.
}
\label{fig:lineargaussian_SM_tau_max1}
\end{figure*}

\clearpage
\begin{figure*}[t]  
\centering
\includegraphics[width=1\linewidth,page=2]{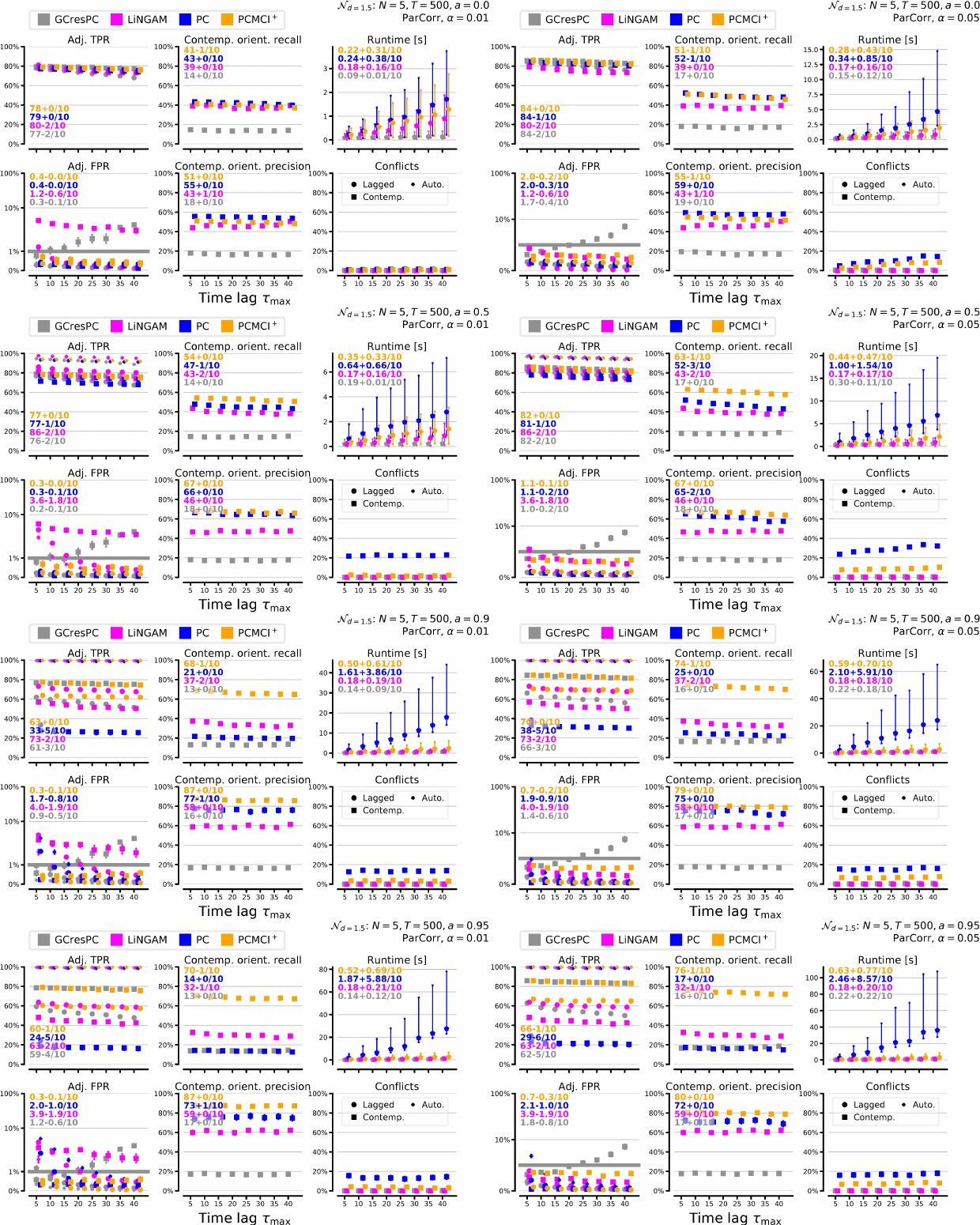}
\caption{
Numerical experiments with 
linear Gaussian  setup
for varying 
 maximum time lag $\tau_{\max}$ 
  and $T=500$
. The left (right) column shows results for significance level $\alpha=0.01$ ($\alpha=0.05$). 
The rows depict results for increasing autocorrelations $a$ (top to bottom).
All model and method parameters are indicated in the upper right of each panel.
}
\label{fig:lineargaussian_SM_tau_max2}
\end{figure*}

\clearpage
\begin{figure*}[t]  
\centering
\includegraphics[width=1\linewidth,page=3]{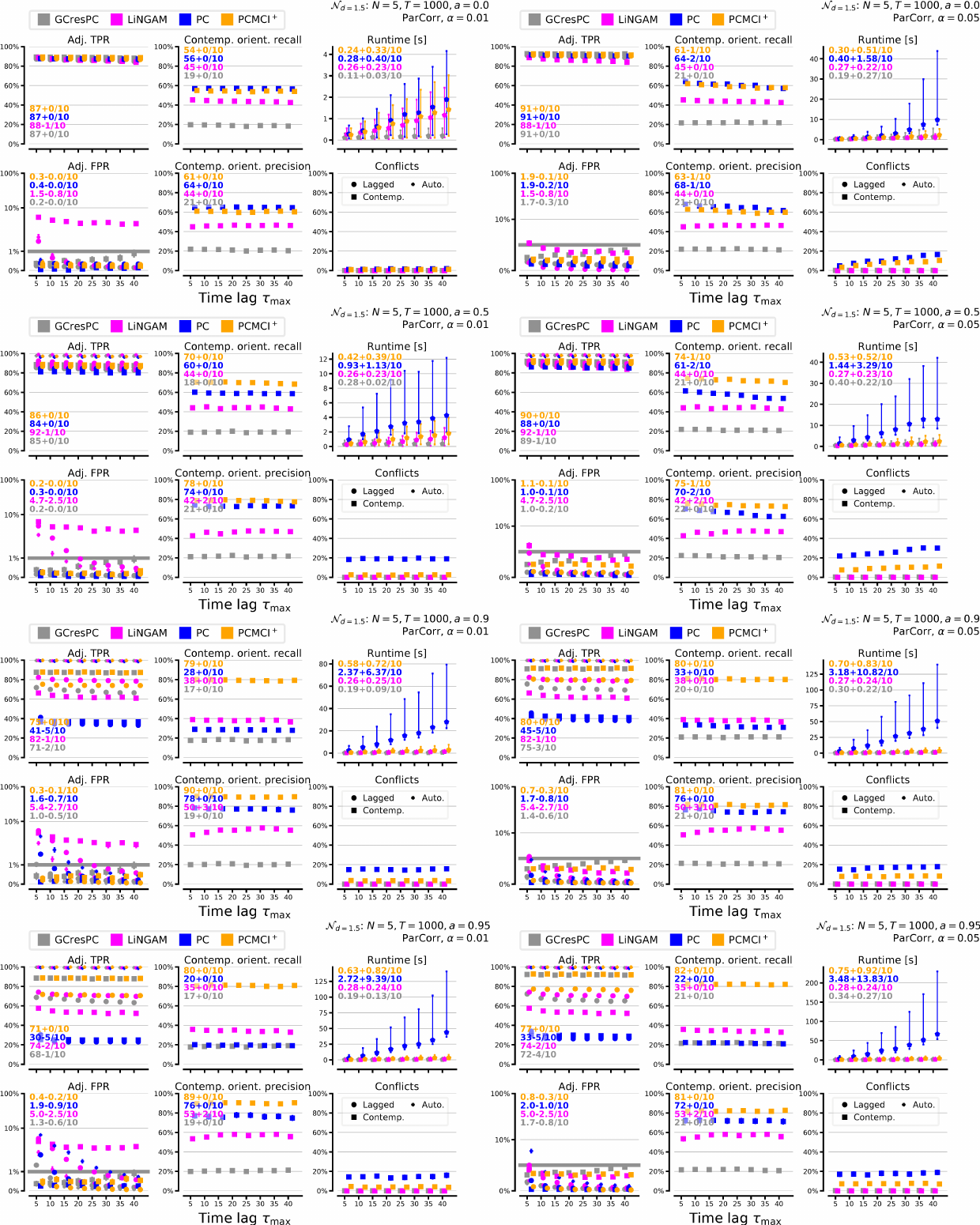}
\caption{
Numerical experiments with 
linear Gaussian  setup
for varying 
 maximum time lag $\tau_{\max}$ 
   and $T=1000$
. The left (right) column shows results for significance level $\alpha=0.01$ ($\alpha=0.05$). 
The rows depict results for increasing autocorrelations $a$ (top to bottom).
All model and method parameters are indicated in the upper right of each panel.
}
\label{fig:lineargaussian_SM_tau_max3}
\end{figure*}

\clearpage
\begin{figure*}[t]  
\centering
\includegraphics[width=1\linewidth,page=1]{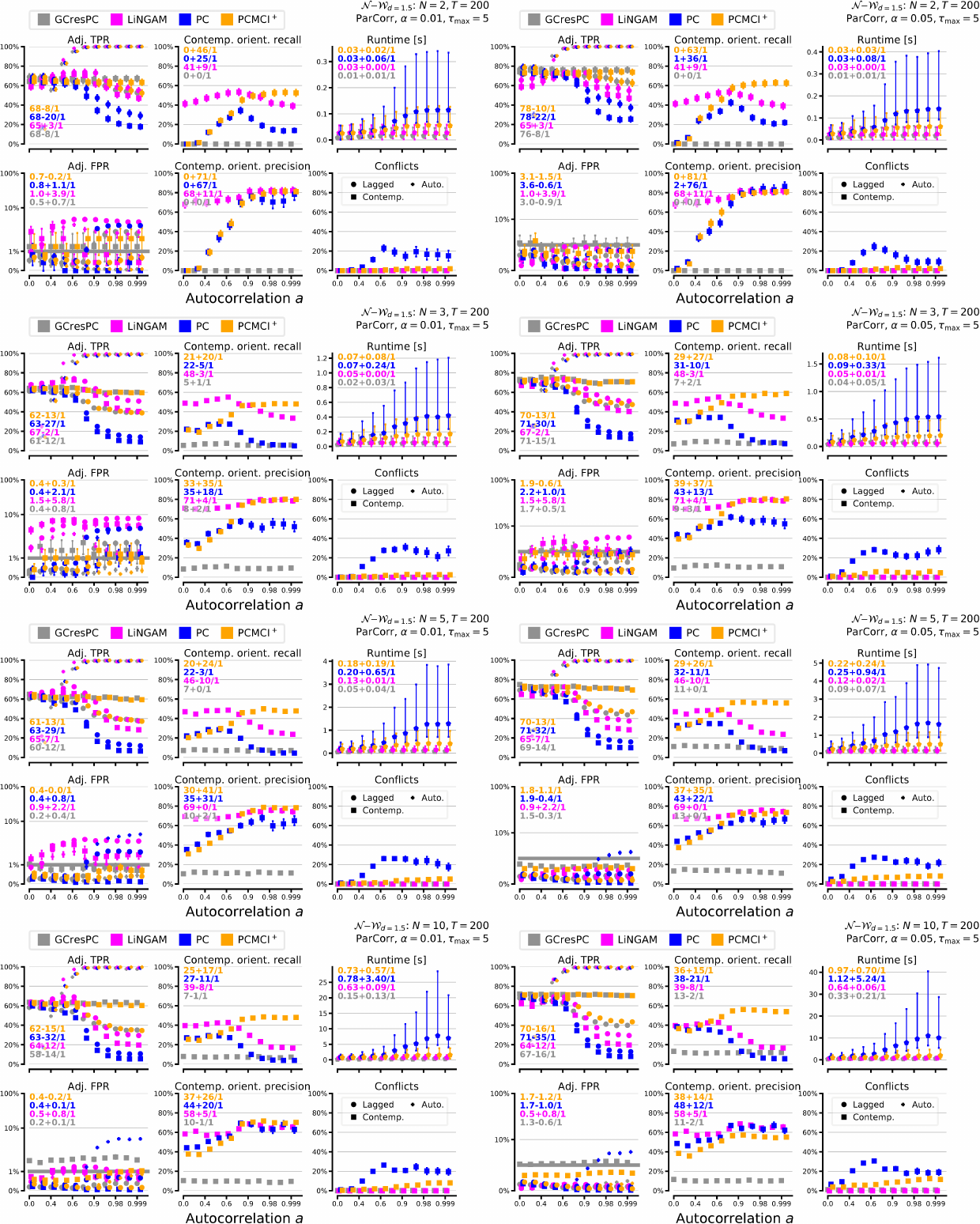}
\caption{
Numerical experiments with 
linear mixed noise setup 
for varying 
autocorrelation $a$
 and $T=200$
. The left (right) column shows results for significance level $\alpha=0.01$ ($\alpha=0.05$). 
The rows depict results for $N=2,\,3,\,5,\, 10$ (top to bottom).
All model and method parameters are indicated in the upper right of each panel.
}
\label{fig:linearmixed_SM_autocorr1}
\end{figure*}

\clearpage
\begin{figure*}[t]  
\centering
\includegraphics[width=1\linewidth,page=2]{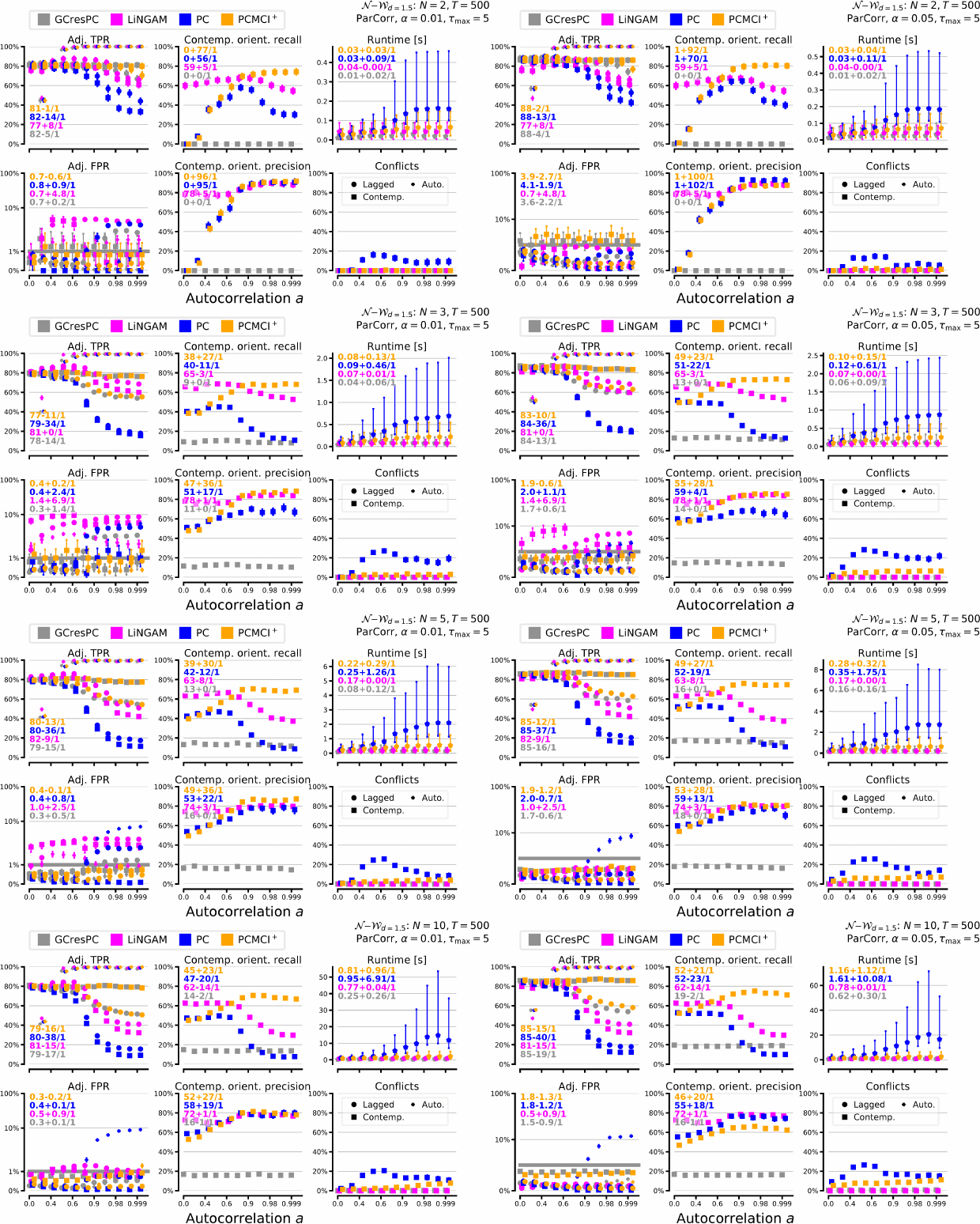}
\caption{
Numerical experiments with 
linear mixed noise setup 
for varying 
autocorrelation $a$
  and $T=500$
. The left (right) column shows results for significance level $\alpha=0.01$ ($\alpha=0.05$). 
The rows depict results for $N=2,\,3,\,5,\, 10$ (top to bottom).
All model and method parameters are indicated in the upper right of each panel.
}
\label{fig:linearmixed_SM_autocorr2}
\end{figure*}

\clearpage
\begin{figure*}[t]  
\centering
\includegraphics[width=1\linewidth,page=3]{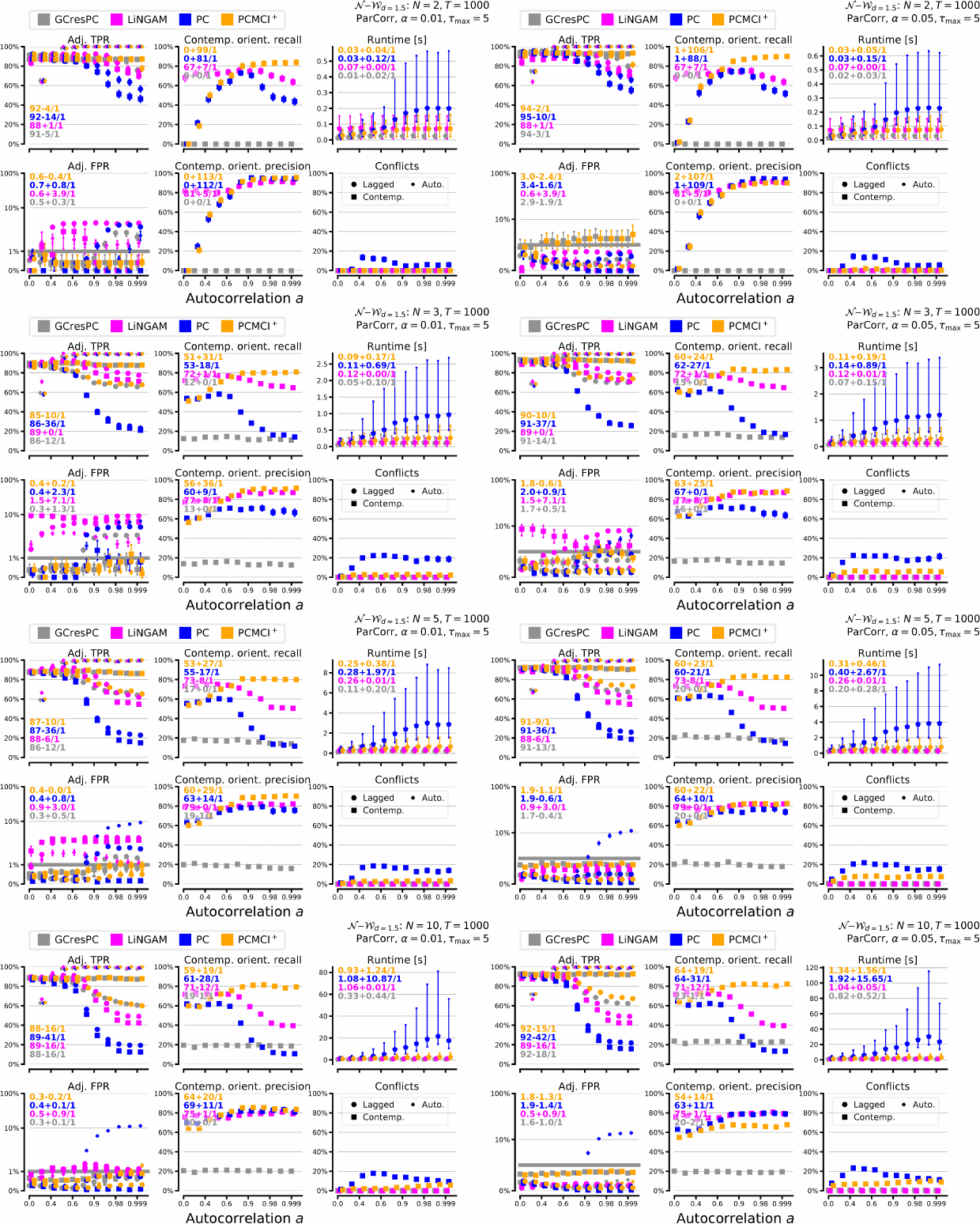}
\caption{
Numerical experiments with
linear mixed noise setup  
for varying 
autocorrelation $a$
   and $T=1000$
. The left (right) column shows results for significance level $\alpha=0.01$ ($\alpha=0.05$). 
The rows depict results for $N=2,\,3,\,5,\, 10$ (top to bottom).
All model and method parameters are indicated in the upper right of each panel.
}
\label{fig:linearmixed_SM_autocorr3}
\end{figure*}

\clearpage
\begin{figure*}[t]  
\centering
\includegraphics[width=1\linewidth,page=1]{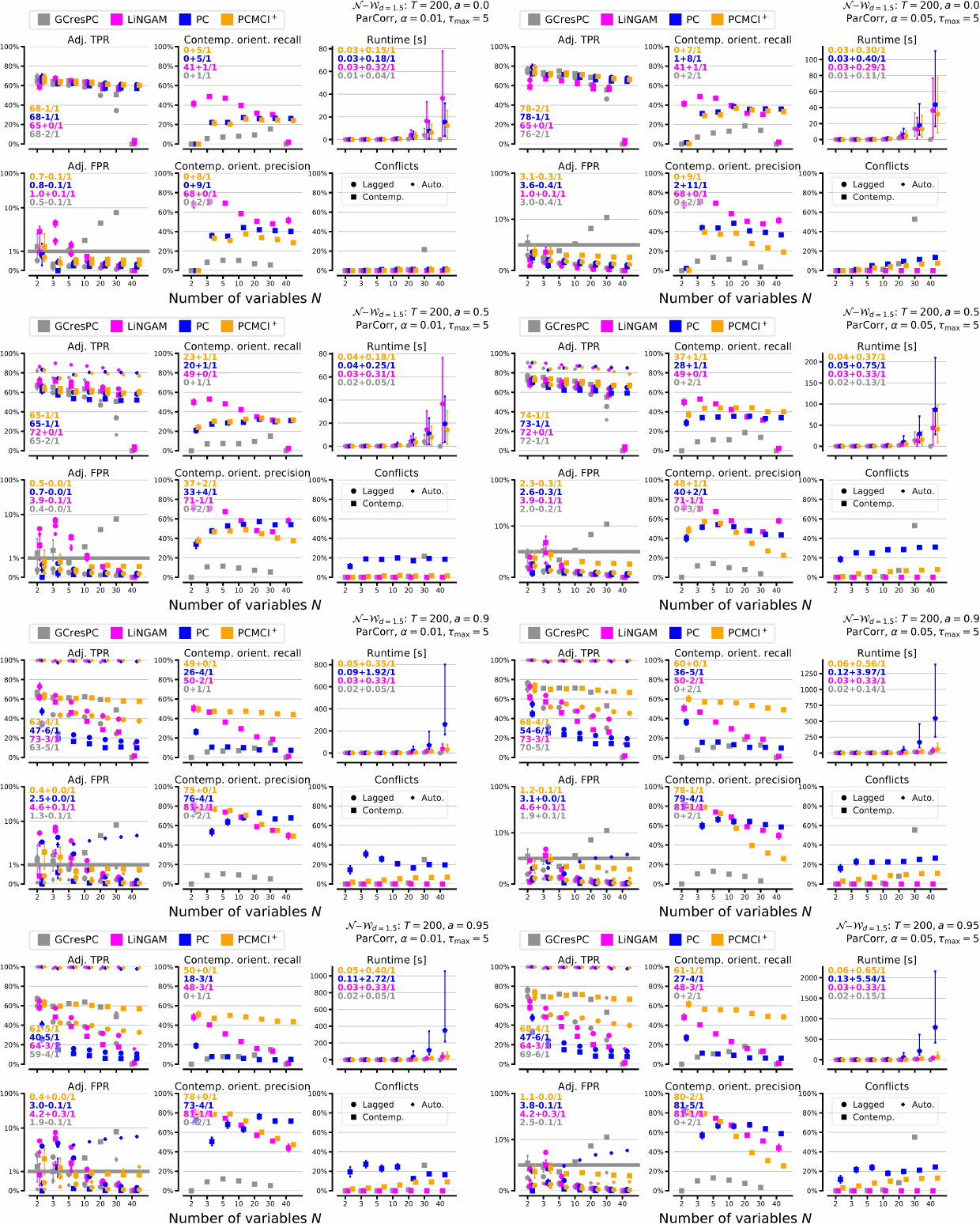}
\caption{
Numerical experiments with 
linear mixed noise setup 
for varying 
 number of variables $N$
 and $T=200$
. The left (right) column shows results for significance level $\alpha=0.01$ ($\alpha=0.05$). 
The rows depict results for increasing autocorrelations $a$ (top to bottom).
All model and method parameters are indicated in the upper right of each panel.
}
\label{fig:linearmixed_SM_highdim1}
\end{figure*}

\clearpage
\begin{figure*}[t]  
\centering
\includegraphics[width=1\linewidth,page=2]{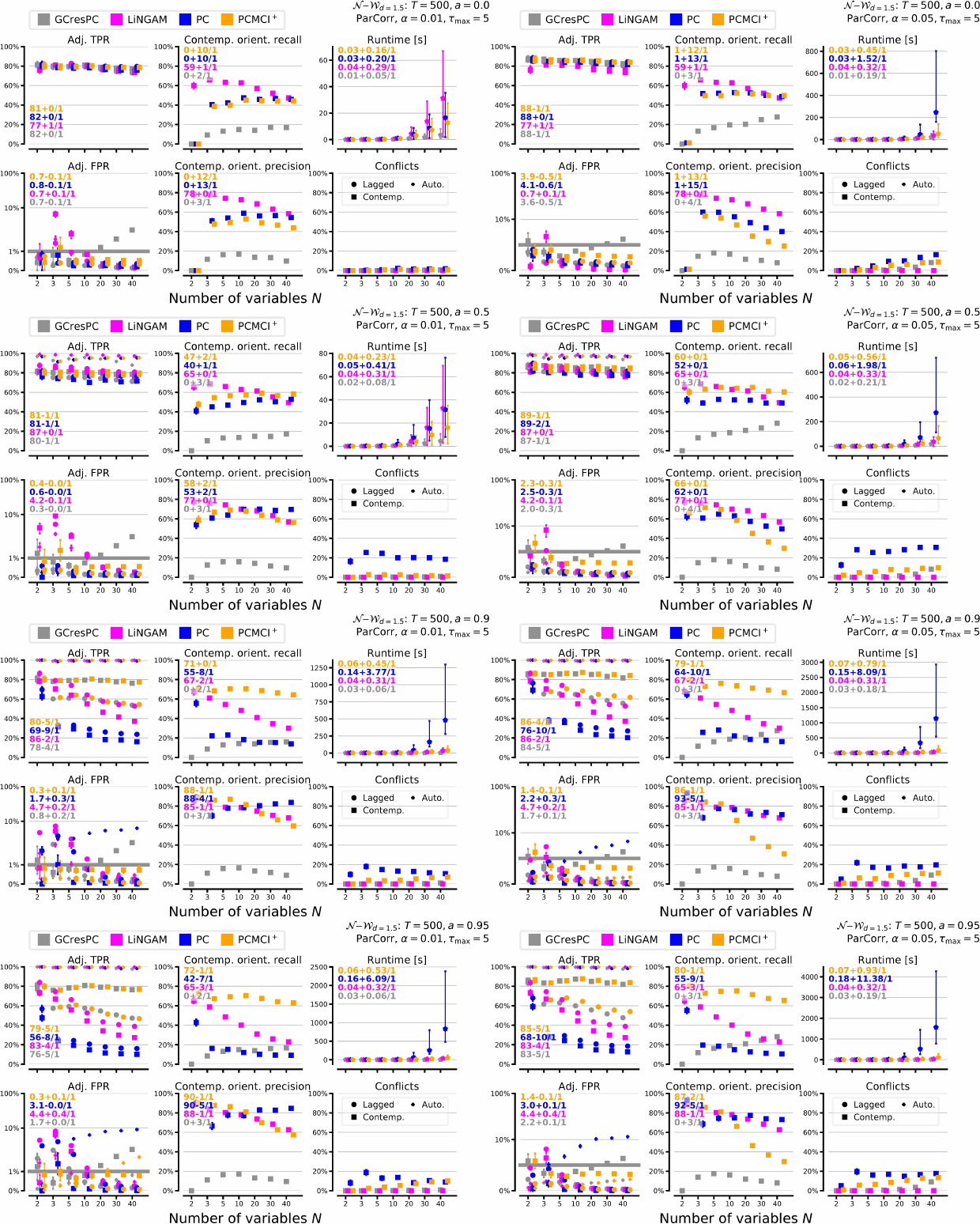}
\caption{
Numerical experiments with 
linear mixed noise setup 
for varying 
 number of variables $N$
  and $T=500$
. The left (right) column shows results for significance level $\alpha=0.01$ ($\alpha=0.05$). 
The rows depict results for increasing autocorrelations $a$ (top to bottom).
All model and method parameters are indicated in the upper right of each panel.
}
\label{fig:linearmixed_SM_highdim2}
\end{figure*}

\clearpage
\begin{figure*}[t]  
\centering
\includegraphics[width=1\linewidth,page=3]{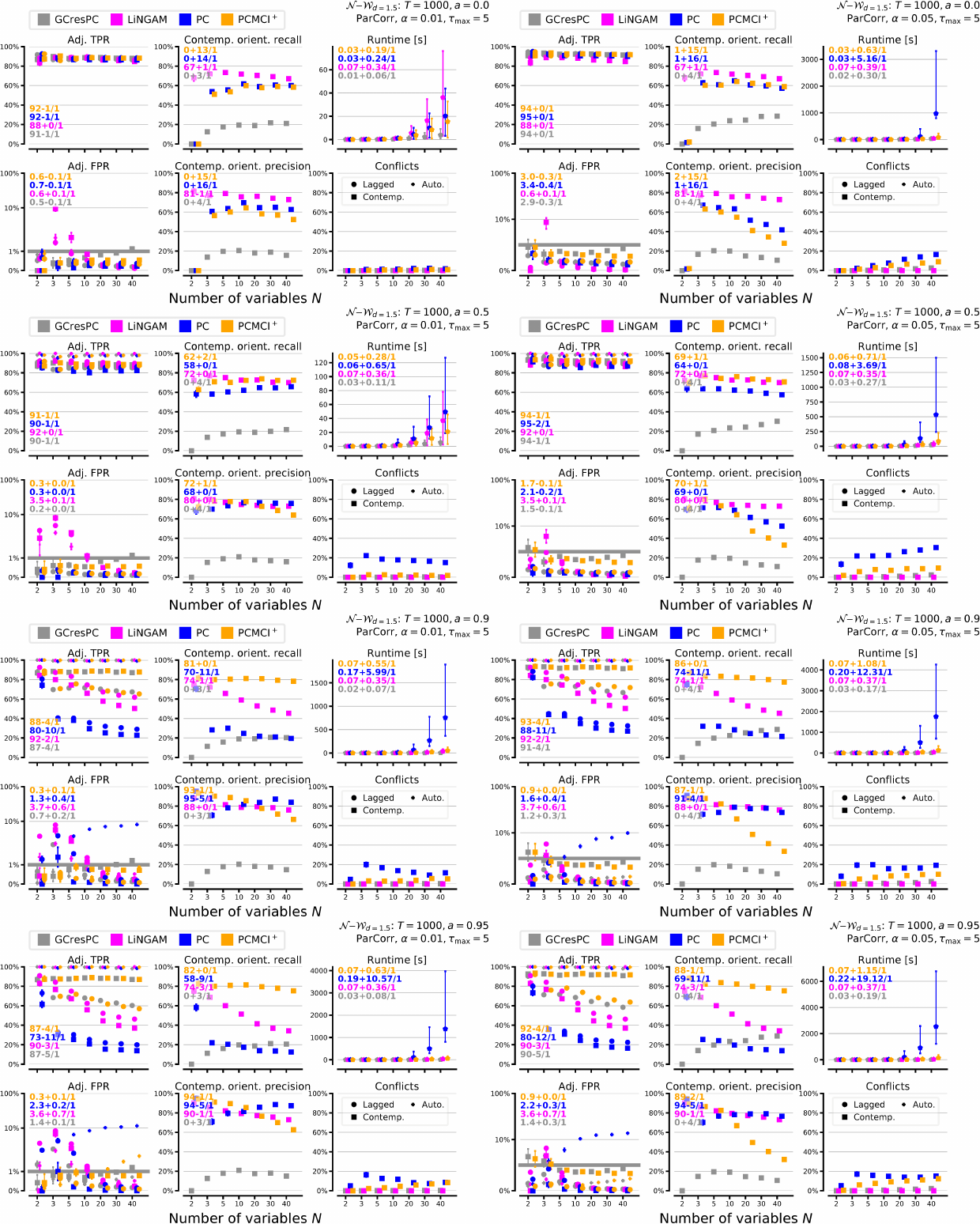}
\caption{
Numerical experiments with 
linear mixed noise setup 
for varying 
 number of variables $N$
   and $T=1000$
. The left (right) column shows results for significance level $\alpha=0.01$ ($\alpha=0.05$). 
The rows depict results for increasing autocorrelations $a$ (top to bottom).
All model and method parameters are indicated in the upper right of each panel.
}
\label{fig:linearmixed_SM_highdim3}
\end{figure*}

\clearpage
\begin{figure*}[t]  
\centering
\includegraphics[width=1\linewidth,page=1]{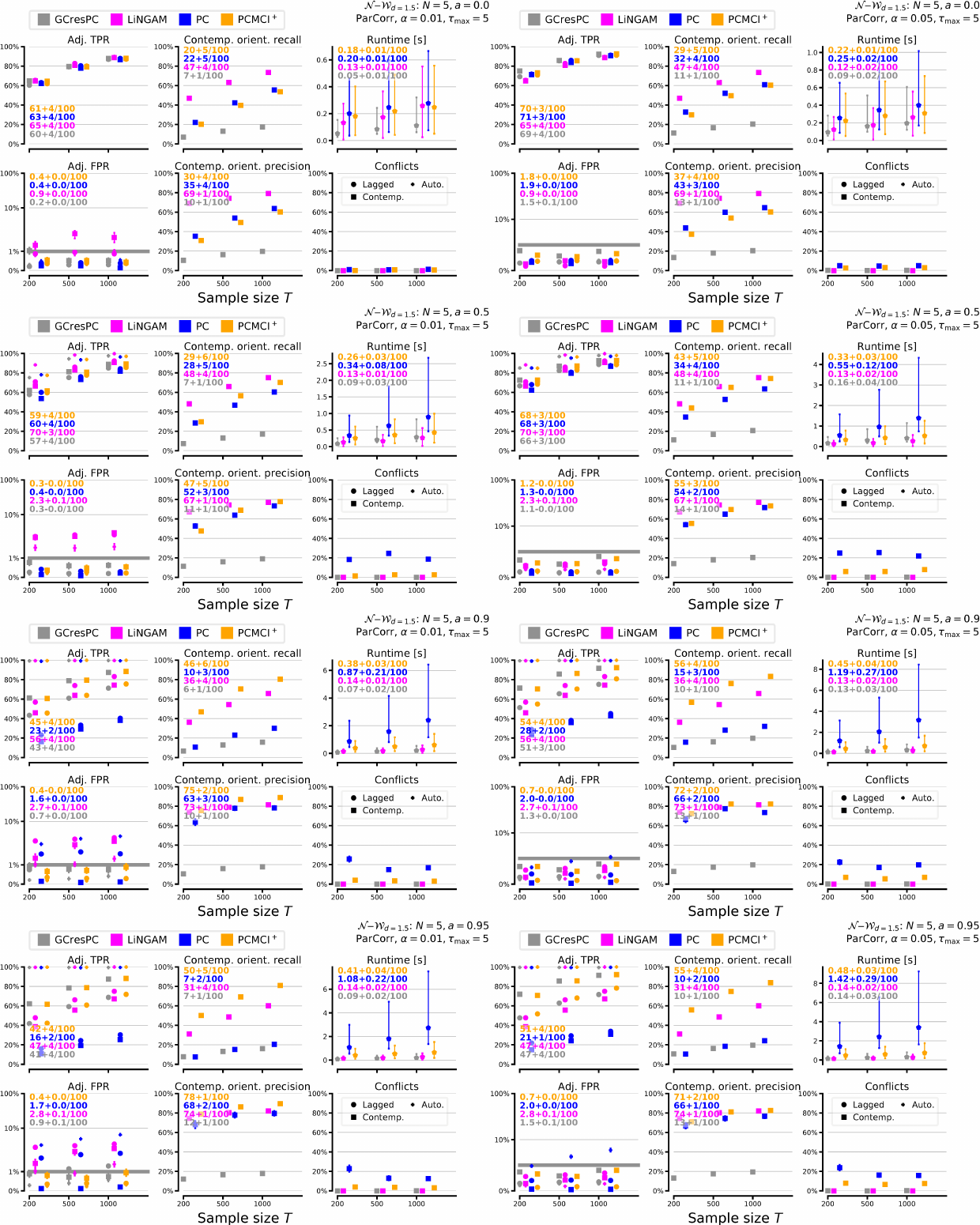}
\caption{
Numerical experiments with 
linear mixed noise setup 
for varying 
 sample size $T$
 for $N=5$ 
. The left (right) column shows results for significance level $\alpha=0.01$ ($\alpha=0.05$). 
The rows depict results for increasing autocorrelations $a$ (top to bottom).
All model and method parameters are indicated in the upper right of each panel.
}
\label{fig:linearmixed_SM_sample_size1}
\end{figure*}

\clearpage
\begin{figure*}[t]  
\centering
\includegraphics[width=1\linewidth,page=2]{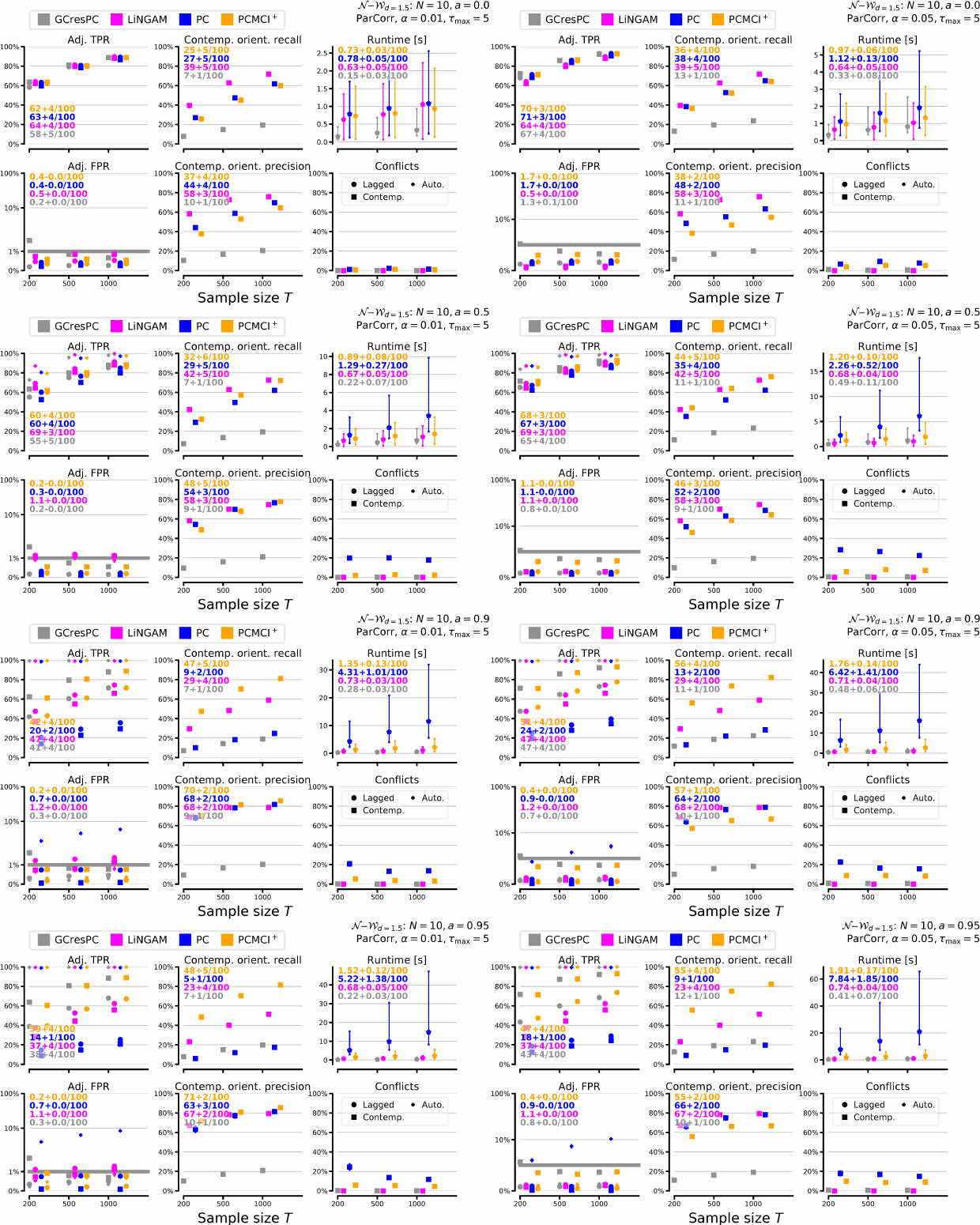}
\caption{
Numerical experiments with 
linear mixed noise setup  
for varying 
 sample size $T$
 for $N=10$ 
. The left (right) column shows results for significance level $\alpha=0.01$ ($\alpha=0.05$). 
The rows depict results for increasing autocorrelations $a$ (top to bottom).
All model and method parameters are indicated in the upper right of each panel.
}
\label{fig:linearmixed_SM_sample_size2}
\end{figure*}

\clearpage
\begin{figure*}[t]  
\centering
\includegraphics[width=1\linewidth,page=3]{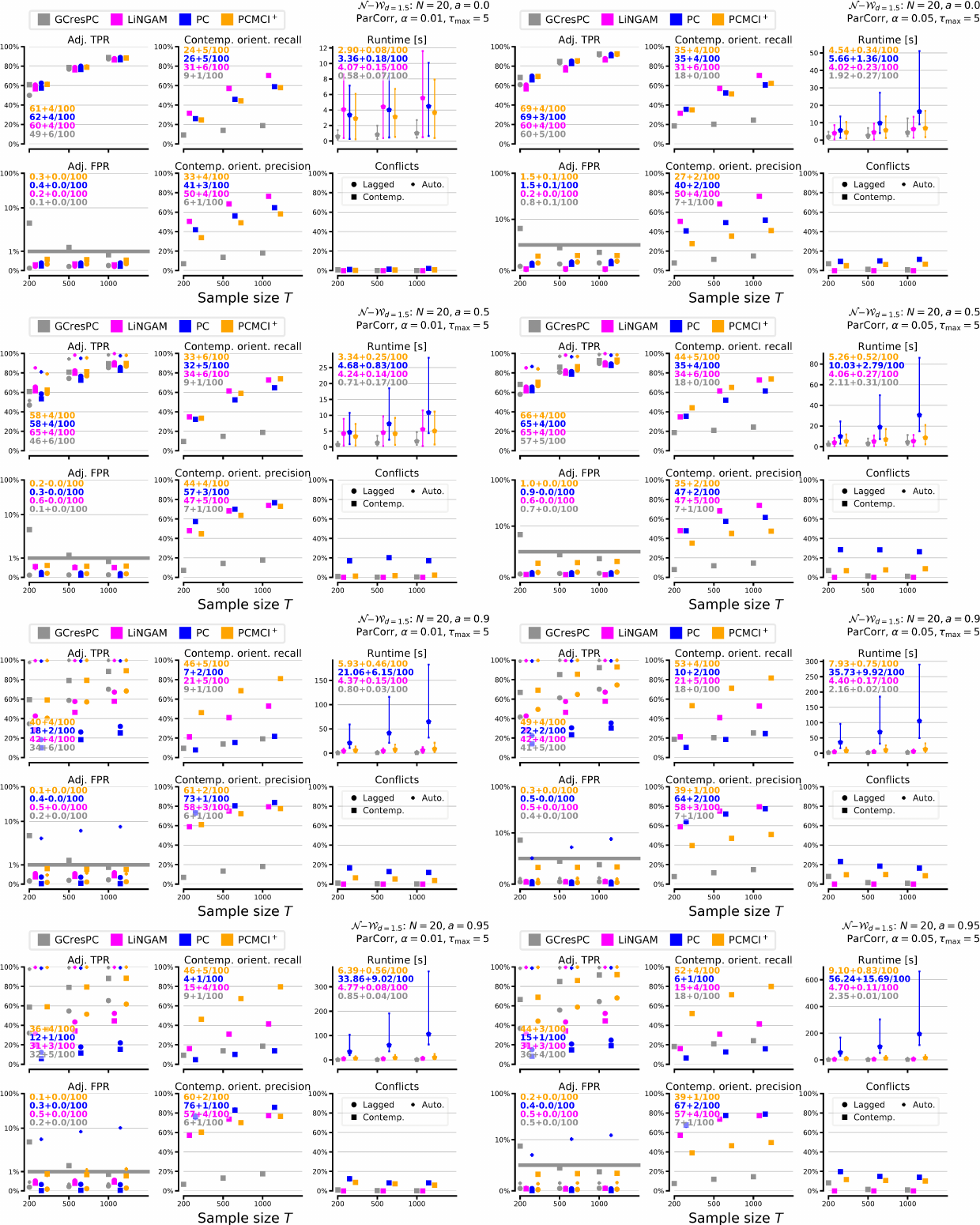}
\caption{
Numerical experiments with 
linear mixed noise setup 
for varying 
 sample size $T$
 for $N=20$ 
. The left (right) column shows results for significance level $\alpha=0.01$ ($\alpha=0.05$). 
The rows depict results for increasing autocorrelations $a$ (top to bottom).
All model and method parameters are indicated in the upper right of each panel.
}
\label{fig:linearmixed_SM_sample_size3}
\end{figure*}

\clearpage
\begin{figure*}[t]  
\centering
\includegraphics[width=1\linewidth,page=1]{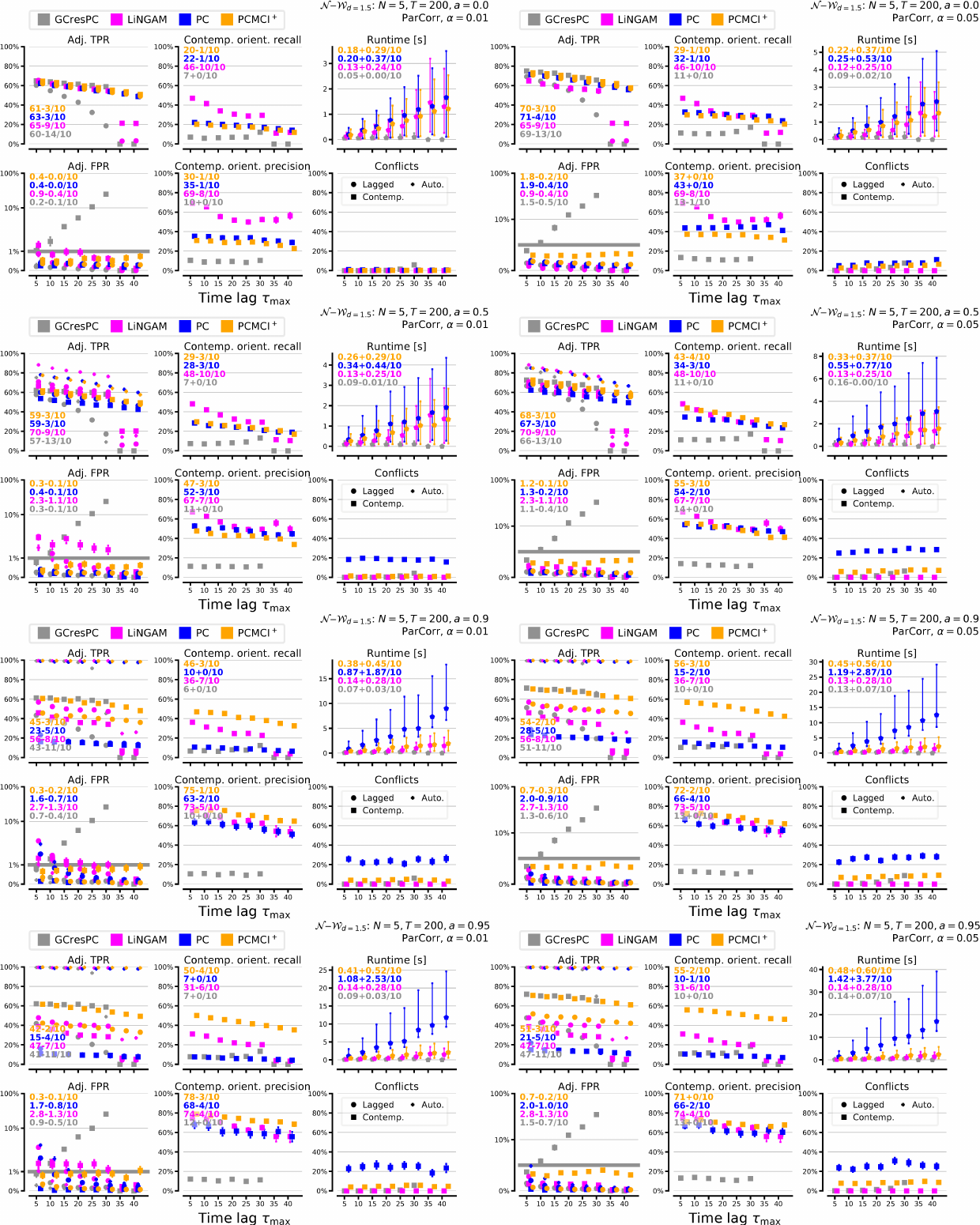}
\caption{
Numerical experiments with 
linear mixed noise setup 
for varying 
 maximum time lag $\tau_{\max}$ 
 and $T=200$
. The left (right) column shows results for significance level $\alpha=0.01$ ($\alpha=0.05$). 
The rows depict results for increasing autocorrelations $a$ (top to bottom).
All model and method parameters are indicated in the upper right of each panel.
}
\label{fig:linearmixed_SM_tau_max1}
\end{figure*}

\clearpage
\begin{figure*}[t]  
\centering
\includegraphics[width=1\linewidth,page=2]{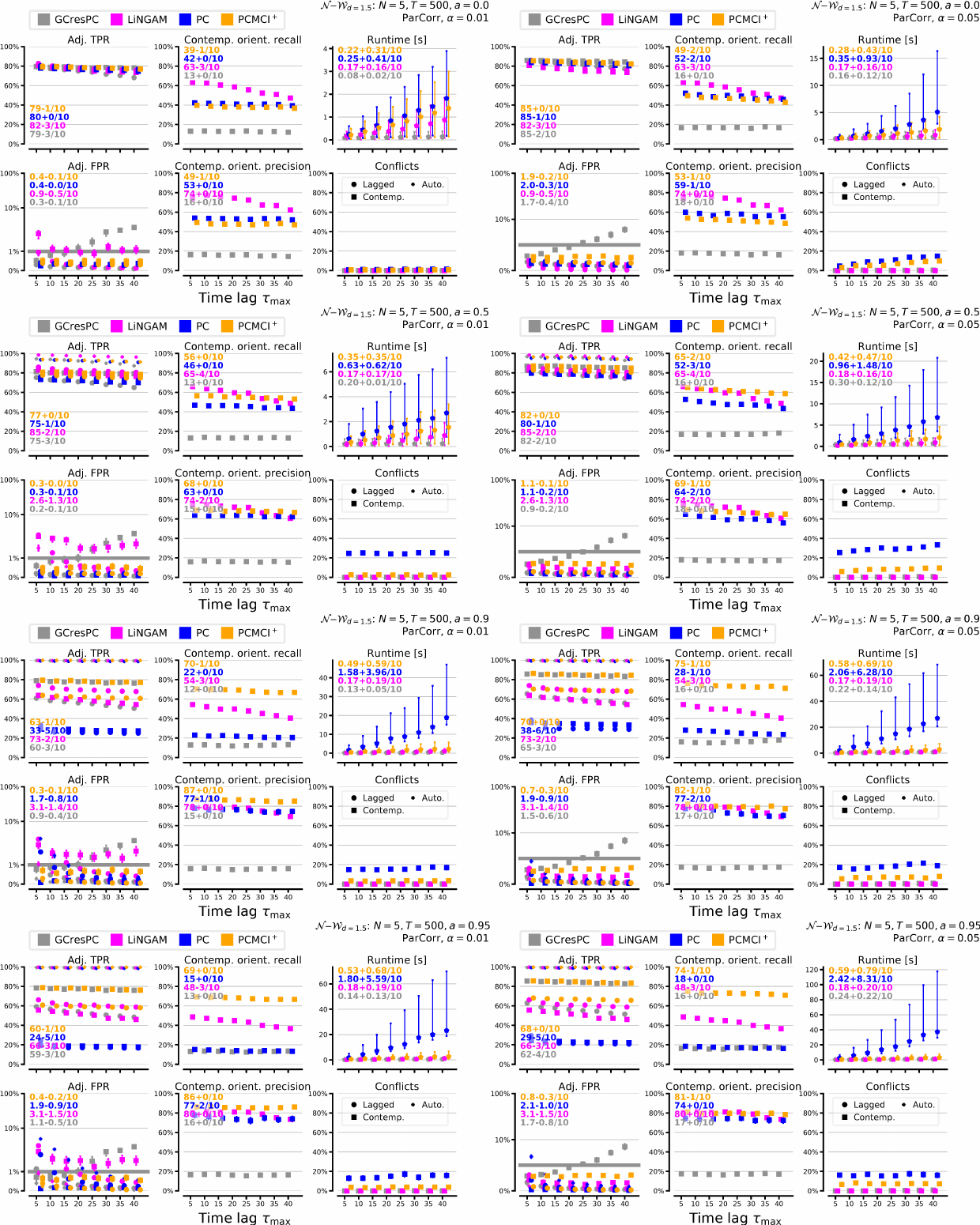}
\caption{
Numerical experiments with 
linear mixed noise setup 
for varying 
 maximum time lag $\tau_{\max}$ 
  and $T=500$
. The left (right) column shows results for significance level $\alpha=0.01$ ($\alpha=0.05$). 
The rows depict results for increasing autocorrelations $a$ (top to bottom).
All model and method parameters are indicated in the upper right of each panel.
}
\label{fig:linearmixed_SM_tau_max2}
\end{figure*}

\clearpage
\begin{figure*}[t]  
\centering
\includegraphics[width=1\linewidth,page=3]{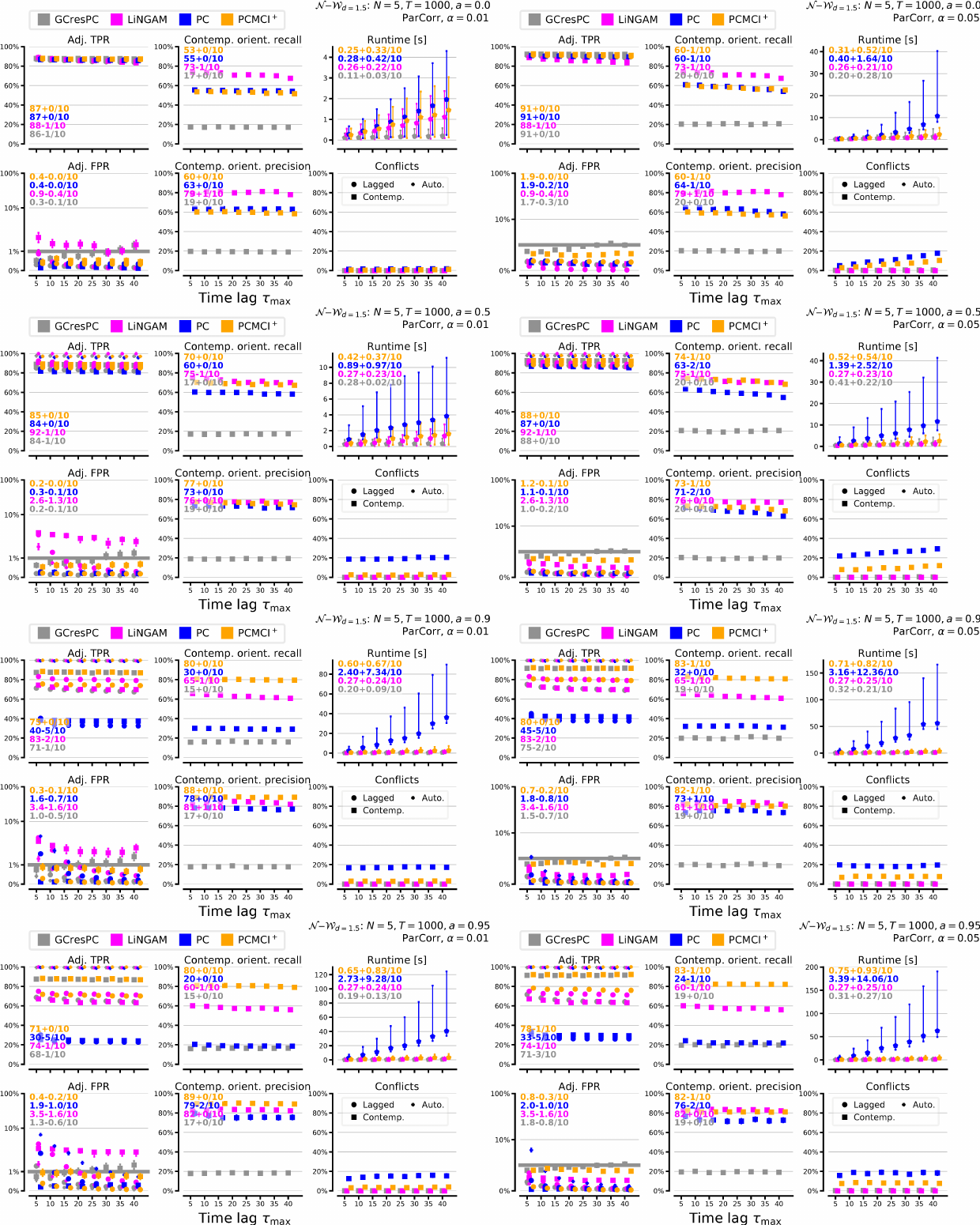}
\caption{
Numerical experiments with 
linear mixed noise setup 
for varying 
 maximum time lag $\tau_{\max}$ 
   and $T=1000$
. The left (right) column shows results for significance level $\alpha=0.01$ ($\alpha=0.05$). 
The rows depict results for increasing autocorrelations $a$ (top to bottom).
All model and method parameters are indicated in the upper right of each panel.
}
\label{fig:linearmixed_SM_tau_max3}
\end{figure*}

\end{document}